\documentclass[letter,12pt]{amsart}
\usepackage{amsmath,bm}
\usepackage{amsfonts}
\usepackage{amssymb}
\usepackage{graphicx}
\usepackage{rotating}
\usepackage[height=8.7in,left=1in,right=1in,bottom=1in]{geometry}
\usepackage[breaklinks=true,bookmarksopen=true,colorlinks=true,citecolor=blue]{hyperref}
\usepackage{color}
\usepackage{colortbl}
\usepackage{paralist}
 \usepackage{booktabs} 
  \usepackage{bbm}

\RequirePackage{hypernat}
\newtheorem{theorem}{Theorem}[section]

\newtheorem{algorithm}[theorem]{Algorithm}
\newtheorem{assumption}{Assumption}

\newtheorem{definition}[theorem]{Definition}

\newtheorem{lemma}[theorem]{Lemma}

\newtheorem{proposition}{Proposition}[section]
\newtheorem{remark}{Remark}[section]

\def \@seccntformat#1{\csname the#1\endcsname.\quad}
\makeatother
\numberwithin{equation}{section}

\newcommand{\indep}{\raisebox{0.05em}{\rotatebox[origin=c]{90}{$\models$}}}

\DeclareMathOperator*{\argmin}{\arg\!\min}
\DeclareMathOperator*{\argmax}{\arg\!\max}
\DeclareMathOperator{\1}{\mathbbm 1}

\usepackage{url}            
\usepackage{float} 
\usepackage{caption}
\usepackage{subcaption}
\usepackage{setspace} 
\usepackage{footmisc}
\makeatother

\let\footnotesize\normalsize

\usepackage{setspace}
\begin{document}
\title[Kernel Sample Selection]{Generalized Kernel Ridge Regression for Causal Inference with Missing-at-Random Sample Selection}
\thanks{{\footnotesize{I thank Anna Mikusheva and Whitney Newey for helpful comments.}}}
\author{Rahul Singh}
\thanks{rahul.singh@mit.edu, MIT Department of Economics.} 

\date{November 9, 2021}

\begin{abstract}
I propose kernel ridge regression estimators for nonparametric dose response curves and semiparametric treatment effects in the setting where an analyst has access to a selected sample rather than a random sample; only for select observations, the outcome is observed. I assume selection is as good as random conditional on treatment and a sufficiently rich set of observed covariates, where the covariates are allowed to cause treatment or be caused by treatment---an extension of missingness-at-random (MAR). I propose estimators of means, increments, and distributions of counterfactual outcomes with closed form solutions in terms of kernel matrix operations, allowing treatment and covariates to be discrete or continuous, and low, high, or infinite dimensional. For the continuous treatment case, I prove uniform consistency with finite sample rates. For the discrete treatment case, I prove $\sqrt{n}$ consistency, Gaussian approximation, and semiparametric efficiency.

\textbf{Keywords}: RKHS, dose response, doubly robust, missing-at-random

\end{abstract}

\maketitle

\section{Introduction}\label{section:intro}

In causal inference with observational data, an analyst may have access to a selected sample rather than a random sample, from which the analyst may wish to estimate the dose response curve of a continuous treatment or the scalar treatment effect of a discrete treatment. For selected observations, the outcome is observed; otherwise it is missing. This issue is also called \textit{outcome attrition}, and it may afflict even experimental data collection. I study the setting where the dataset includes a sufficiently rich set of covariates to ameliorate the problem: conditional on treatment and these covariates, selection is as good as random. I allow the covariates of the selection mechanism to either cause treatment or be caused by treatment, an extension of the canonical setting called \textit{missing-at-random} (MAR). I refer to the \textit{static} sample selection case when covariates cause treatment. I refer to the \textit{dynamic} sample selection case when some covariates cause treatment and other covariates are caused by treatment.

For this setting of data limitations (observational study, sample selection) and compensatory data access (sufficiently many and well-placed covariates in the causal graph), I propose a family of estimators for causal inference based on kernel ridge regression. The estimators are easy to compute from closed form expressions, though treatment and covariates are allowed to be discrete or continous and low, high, or infinite dimensional. For the continuous treatment case, I propose dose response curve estimators and prove uniform consistency with finite sample rates under standard RKHS assumptions. For the discrete treatment case, I propose treatment effect estimators with confidence intervals and prove $\sqrt{n}$-consistency, Gaussian approximation, and semiparametric efficiency under standard RKHS assumptions. 

I extend the framework in three ways: (i) from causal parameters of the full population to those of subpopulations and alternative populations; (ii) from means to increments and distributions of counterfactual outcomes; (iii) from static sample selection to dynamic sample selection. Tables~\ref{tab:overview_static} and~\ref{tab:overview_dynamic} summarize my contributions. All estimators and guarantees are new.

\begin{table}[ht]
  \caption{New estimators and guarantees for static sample selection}
  \begin{tabular}{c|cccc}
       \toprule
    Causal parameter & Symbol & Guarantee & Rate & Section \\
    \midrule
    Dose response & ATE & uniform consistency & $n^{-\frac{1}{6}}$ & \ref{section:static}\\
    Dose response with dist. shift & DS & uniform consistency & $n^{-\frac{1}{6}}$ & \ref{section:static}\\
    Conditional dose response & ATT & uniform consistency & $n^{-\frac{1}{6}}$ & \ref{section:static}\\
        Heterogeneous dose response & CATE & uniform consistency & $n^{-\frac{1}{6}}$ & \ref{section:static}\\
      \midrule
     Incremental response & $\nabla$:ATE & uniform consistency & $n^{-\frac{1}{6}}$ & \ref{section:static}\\
    Incremental response with dist. shift & $\nabla$:DS & uniform consistency & $n^{-\frac{1}{6}}$ & \ref{section:static}\\
    Conditional incremental response & $\nabla$:ATT & uniform consistency & $n^{-\frac{1}{6}}$ & \ref{section:static}\\
        Heterogeneous incremental response & $\nabla$:CATE & uniform consistency & $n^{-\frac{1}{6}}$ & \ref{section:static}\\
      \midrule
    Counterfactual dist. & D:ATE & convergence in dist. & $n^{-\frac{1}{6}}$ & \ref{section:dist}\\
    Counterfactual dist. with dist. shift & D:DS & convergence in dist. & $n^{-\frac{1}{6}}$ & \ref{section:dist}\\
        Conditional counterfactual dist. & D:ATT & convergence in dist. & $n^{-\frac{1}{6}}$ & \ref{section:dist}\\
        Heterogeneous counterfactual dist. & D:CATE & convergence in dist. & $n^{-\frac{1}{6}}$ & \ref{section:dist}\\
       \midrule
     Average treatment effect & ATE & Gaussian approx. & $n^{-\frac{1}{2}}$ & \ref{section:static}\\
    Average treatment effect with dist. shift & DS & Gaussian approx. & $n^{-\frac{1}{2}}$ & \ref{section:static}\\
    Average treatment on the treated & ATT & Gaussian approx. & $n^{-\frac{1}{2}}$ & \ref{section:static}\\
        Heterogeneous treatment effect & CATE & Gaussian approx. & $n^{-\frac{1}{2}}$ & \ref{section:static}\\
    \bottomrule
    \end{tabular}
  \label{tab:overview_static}%
\end{table}

\begin{table}[ht]
  \caption{New estimators and guarantees for dynamic sample selection}
  \begin{tabular}{c|cccc}
       \toprule
    Causal parameter & Symbol & Guarantee & Rate & Section \\
    \midrule
   Dose response & ATE & uniform consistency & $n^{-\frac{1}{6}}$ & \ref{section:dynamic}\\
    Dose response with dist. shift & DS & uniform consistency & $n^{-\frac{1}{6}}$ & \ref{section:dynamic}\\
      \midrule
    Counterfactual dist. & D:ATE & convergence in dist. & $n^{-\frac{1}{6}}$ & \ref{section:dist}\\
    Counterfactual dist. with dist. shift & D:DS & convergence in dist. & $n^{-\frac{1}{6}}$ & \ref{section:dist}\\
       \midrule
     Average treatment effect & ATE & Gaussian approx. & $n^{-\frac{1}{2}}$ & \ref{section:dynamic}\\
    Average treatment effect with dist. shift & DS & Gaussian approx. & $n^{-\frac{1}{2}}$ & \ref{section:dynamic}\\
    \bottomrule
    \end{tabular}
  \label{tab:overview_dynamic}%
\end{table}

\subsection{Preview}

I preview the framework with a vignette about the static case. Suppose an analyst observes a random sample in which some observations are selected ($S=1$) and others are not ($S=0$). Selected observations are complete in the sense that they include outcome $Y$, treatment $D$, and covariates $X$. Unselected observations are incomplete in the sense that they include only treatment $D$ and covariates $X$. The challenge for causal inference is that the selection mechanism may be confounded, so naive treatment effect estimation with complete cases may be misleading.

The baseline goal is to infer causal parameters of the full population, though only complete observations from the selected population are accessible. Consider the causal parameter 
$
\theta^{ATE}_0(d):=\mathbb{E}\{Y^{(d)}\}
$, 
where $Y^{(d)}$ is the potential outcome given the intervention $D=d$. With machine learning, one can estimate the conditional expectation function
$
\gamma_0(1,d,x)=\mathbb{E}(Y|S=1,D=d,X=x)
$ because selected observations include $(Y,D,X)$. It turns out that an appropriate formula for the causal parameter is of the form
$
\theta^{ATE}_0(d)=\int \gamma_0(1,d,x)\mathrm{d}\mathbb{Q}
$
where $\mathbb{Q}$ is a distribution that reweights $\gamma_0$ in order to adjust for confounding of the treatment assignment and sample selection mechanisms. This distribution $\mathbb{Q}$ may be estimated from both complete and incomplete observations.

As my first contribution, I prove that nonparametric estimation of a dose response curve from a selected sample can be reduced to the inner product of generalized kernel ridge regressions. Suppose $\gamma_0$ is correctly specified as a function in an RKHS $\mathcal{H}$ over selection $S$, treatment $D$, and covariates $X$. I estimate $\hat{\gamma}$ by classic kernel ridge regression. Next, I prove that the reweighting distribution $\mathbb{Q}$ can be represented by another function $\mu$ in the RKHS. The injective mapping $\mathbb{Q}\mapsto \mu$ is called the kernel mean embedding \cite{smola2007hilbert}. I estimate $\hat{\mu}$ by a procedure that involves generalized kernel ridge regression. Finally, I take the inner product of $\hat{\gamma}$ and $\hat{\mu}$, which corresponds to taking the expectation with respect to the counterfactual distribution $\mathbb{Q}$: $\hat{\theta}^{ATE}(d)=\int \hat{\gamma}(1,d,x)\mathrm{d}\hat{\mathbb{Q}}=\langle \hat{\gamma},\hat{\mu} \rangle_{\mathcal{H}}$.

As my second contribution, I prove that semiparametric inference for a treatment effect from a selected sample can be reduced to combinations of generalized kernel ridge regressions as well. In particular, confidence intervals involve Riesz representers, which can be estimated by generalized kernel ridge regressions and sequential mean embeddings. I construct confidence intervals from the Riesz representers by bias correction and sample splitting. Across settings, I appeal to RKHS geometry to articulate smoothness and spectral decay approximation assumptions for the various generalized kernel ridge regressions. These approximation assumptions generalize the standard smoothness and spectral decay assumptions in RKHS learning theory. 

The structure of the paper is as follows. Section~\ref{sec:related_main} describes related work. Section~\ref{sec:rkhs_main} presents the assumptions, which are standard in RKHS learning theory. Section~\ref{section:static} presents nonparametric and semiparametric theory for the case where covariates cause treatment.  Section~\ref{section:dynamic} presents nonparametric and semiparametric theory for the case where covariates may cause or be caused by treatment. Section~\ref{section:conclusion} concludes. I extend nonparametric results to counterfactual distributions in Appendix~\ref{section:dist}.

\section{Related work}\label{sec:related_main}

\textbf{Sample selection with continuous treatment.} I express dose response curves with static and dynamic sample selection as reweightings of an underlying conditional expectation function, extending the g-formula framework in biostatistics \cite{robins1986new} and the partial means framework in econometrics \cite{newey1994kernel}. In the static case, the reweighting is a simple marginal or conditional distribution. In the dynamic case, the reweighting is a product of intertemporal distributions. Importantly, I reweight for both treatment and selection mechanisms using covariates, a setting sometimes called double selection with missing-at-random (MAR) data \cite{rubin1976inference}. To formulate dose response curves in this way, I generalize nonparametric identification theorems of \cite{huber2012identification,huber2014treatment}. Existing work on continuous treatment studies uses parametric assumptions and proposes parametric estimators \cite{heckman1979sample,hausman1979attrition} or uses conditional moment and rank restrictions and proposes series estimators \cite{das2003nonparametric}. A machine learning dose response estimator for double selection does not appear to exist. I pose as a question for future work how to adapt the dose response estimators to other sample selection problems such as missing-not-at-random data (MNAR) \cite{scharfstein1999adjusting,rotnitzky1998semiparametric}, also called non-ignorable attrition, where auxiliary information such as instrumental variables \cite{das2003nonparametric,huber2012identification} and shadow variables \cite{d2010new,miao2015identification} are used instead.

\textbf{Sample selection with discrete treatment.} I express treatment effects with static and dynamic sample selection as functionals of an underlying conditional expectation function. In the static case, the functionals are linear; in the dynamic case, they are nonlinear.  Importantly, I adjust for both treatment and selection mechanisms \cite{negi2020doubly,bia2020double} using covariates in a multiply robust moment function \cite{robins1994estimation,robins1995semiparametric}. A rich literature provides results adjusting only for the selection mechanism in a doubly robust moment function \cite{robins1994estimation,bang2005doubly}. In the static case, I derive the doubly robust moment function for a broad class of treatment effects, extending the ATE characterization of \cite{bia2020double}. For the dynamic case, I quote the multiply robust moment function \cite{bia2020double}. I combine the multiply robust moments with sample splitting \cite{bickel1982adaptive}, following the principles of targeted maximum likelihood estimation (TMLE) \cite{van2006targeted} and debiased machine learning (DML) \cite{chernozhukov2016locally,chernozhukov2018original,chernozhukov2018global,chernozhukov2018learning}. I combine generalized kernel ridge regressions to estimate nonparametric quantities without explicit density estimation. I prove new, sufficiently fast nonparametric rates to verify abstract rate conditions from \cite{chernozhukov2018original,bia2020double,chernozhukov2021simple}. I pose as a question for future work how to adapt the treatment effect estimators to other sample selection problems where data are MNAR rather than MAR.

\textbf{RKHS in causal inference.} Existing work incorporates the RKHS into static and dynamic causal inference without sample selection. \cite{nie2017quasi} propose the R learner (named for \cite{robinson1988root}) to nonparametrically estimate the heterogenous treatment effect of binary treatment. The authors prove mean square error rates, which \cite{foster2019orthogonal} situate within a broader theoretical framework. \cite{singh2019kernel} study the nonparametric instrumental variable problem and elucidate the role of conditional expectation operators and conditional mean embeddings in causal inference. \cite{kallus2020generalized} proposes kernel optimal matching, a semiparametric estimator of average treatment on the treated with binary treatment, and proves $\sqrt{n}$ consistency. \cite{hirshberg2019minimax} propose a semiparametric balancing weight estimator for average treatment on the treated and its variants, and provide bias aware confidence intervals. \cite{singh2020kernel,singh2020negative,singh2021workshop} propose nonparametric and semiparametric estimators with closed form solutions based on kernel ridge regression for a broad variety of static and dynamic treatment effects including the dose response curve and heterogeneous treatment effect. It appears that this project is the first to propose RKHS algorithms for causal inference with sample selection.

\textbf{RKHS and sample selection.} In the machine learning literature, the term \textit{sample selection} is used to refer to prediction problems in which the training set and testing set are drawn from different distributions and the analyst possibly has access to unlabeled (i.e. excluding outcome) observations from the testing set. A rich variety of nonparametric reweighting techniques have been developed using RKHS methods \cite{huang2006correcting,cortes2008sample,gretton2009covariate}. By contrast, the version of the sample selection problem I study is causal. The goal is to estimate nonparametric dose response curves, semiparametric treatment effects, and counterfactual distributions. To forge a connection between the present work and the problem studied in the machine learning literature, I extend my results to the setting with distribution shift: using labeled and unlabeled observations drawn from a population $\mathbb{P}$ and unlabeled observations drawn from an alternative population $\tilde{\mathbb{P}}$, the analyst seeks to estimate causal parameters for the alternative population.

\textbf{Counterfactual distributions.} Previous work on counterfactual distributions considers static and dynamic settings without sample selection. Many papers study distributional generalizations of average treatment effect (ATE) or average treatment on the treated (ATT) for binary treatment \cite{dinardo1996labor,firpo2007efficient,imbens2009identification,cattaneo2010efficient,chernozhukov2013inference,muandet2020counterfactual}. \cite{singh2020kernel,singh2021workshop} present RKHS estimators for a richer class of static and dynamic treatment effects where treatment may be discrete or continuous. It appears that this project is the first to propose nonparametric algorithms for counterfactual distributions with sample selection.

Unlike previous work, I (i) provide a framework that encompasses nonparametric and semiparametric treatment effects in static and dynamic sample selection problems; (ii) propose new estimators with closed form solutions based on a new RKHS construction; and (iii) prove uniform consistency for the continuous treatment case and $\sqrt{n}$ consistency, Gaussian approximation, and semiparametric efficiency for the discrete treatment case via original arguments that handle low, high, or infinite dimensional data.
\section{RKHS assumptions}\label{sec:rkhs_main}

I summarize RKHS notation, interpretation, and assumptions from \cite[Section 3]{cucker2002mathematical}, \cite[Chapter 4]{steinwart2008support}, \cite[Section 2]{singh2020kernel} and \cite[Section 3]{singh2021workshop}.
A reproducing kernel Hilbert space (RKHS) $\mathcal{H}$ has elements that are functions $f:\mathcal{W}\rightarrow\mathbb{R}$ where $\mathcal{W}$ is a Polish space, i.e a separable and completely metrizable topological space. Let $k:\mathcal{W}\times\mathcal{W}\rightarrow \mathbb{R}$ be a function that is continuous, symmetric, and positive definite. I call $k$ the \textit{kernel}, and I call $\phi:w\mapsto k(w,\cdot)$ the \textit{feature map}. The kernel is the inner product of features in the sense that $k(w,w')=\langle \phi(w),\phi(w')\rangle_{\mathcal{H}}$. The RKHS is the closure of the span of the kernel \cite[Theorem 4.21]{steinwart2008support}.

The key approximation assumptions for the RKHS are smoothness and spectral decay. I articulate the definition of the RKHS as well as these assumptions in terms of the convolution operator in which $k$ serves as the kernel, i.e. $
L:L_{\nu}^2(\mathcal{W})\rightarrow L_{\nu}^2(\mathcal{W}),\; f\mapsto \int k(\cdot,w)f(w)\mathrm{d}\nu(w)
$, where $L^2_{\nu}(\mathcal{W})$ the space of square integrable functions with respect to measure $\nu$. By the spectral theorem, I express the eigendecomposition of the convolution operator $L$ as
$
Lf=\sum_{j=1}^{\infty} \lambda_j\langle \varphi_j,f \rangle_{L^2_{\nu}(\mathcal{W})}\cdot  \varphi_j
$ where $(\lambda_j)$ are weakly decreasing eigenvalues and $(\varphi_j)$ are orthonormal eigenfunctions that form a basis of $L_{\nu}^2(\mathcal{W})$.

First, I define the RKHS in terms of the orthonormal basis $(\varphi_j)$. For comparison, I also define the more familiar space $L^2_{\nu}(\mathcal{W})$ in terms of this orthonormal basis. For any $f,g\in L_{\nu}^2(\mathcal{W})$, write $
f=\sum_{j=1}^{\infty}a_j\varphi_j
$ and $
g=\sum_{j=1}^{\infty}b_j\varphi_j
$. By \cite[Theorem 4]{cucker2002mathematical},
\begin{align*}
 L^2_{\nu}(\mathcal{W})&=\left(f=\sum_{j=1}^{\infty}a_j\varphi_j:\; \sum_{j=1}^{\infty}a_j^2<\infty\right),\quad \langle f,g \rangle_{L^2_{\nu}(\mathcal{W})}=\sum_{j=1}^{\infty} a_jb_j; \\
 \mathcal{H}&=\left(f=\sum_{j=1}^{\infty}a_j\varphi_j:\;\sum_{j=1}^{\infty} \frac{a_j^2}{\lambda_j}<\infty\right),\quad \langle f,g \rangle_{\mathcal{H}}=\sum_{j=1}^{\infty} \frac{a_jb_j}{\lambda_j}.
\end{align*}
In summary, the RKHS $\mathcal{H}$ can be defined as the subset of $L^2_{\nu}(\mathcal{W})$ for which higher order terms in the series $(\varphi_j)$ have a smaller contribution, subject to $\nu$ satisfying the conditions of Mercer's theorem \cite{steinwart2012mercer}.

Next, I define the smoothness and spectral decay assumptions in terms of the eigenvalues $(\lambda_j)$. Denote the statistical target $f_0$ that is estimated in the RKHS $\mathcal{H}$. A smoothness assumption on $f_0$ can be formalized as \begin{equation}\label{eq:prior}
    f_0\in \mathcal{H}^c:=\left(f=\sum_{j=1}^{\infty}a_j\varphi_j:\;\sum_{j=1}^{\infty} \frac{a_j^2}{\lambda^c_j}<\infty\right)\subset \mathcal{H},\quad c\in(1,2].
\end{equation}
In words, $f_0$ is well approximated by the leading terms in the series $(\varphi_j)$. A larger value of $c$ corresponds to a smoother target $f_0$. This condition is also called the source condition. A spectral decay assumption on $f_0$ can be formalized as $\lambda_j\asymp j^{-b}$. A larger value of $b$ corresponds to a faster rate of spectral decay and therefore a lower effective dimension. Both $(b,c)$ are joint assumptions on the kernel and data distribution \cite{smale2007learning,caponnetto2007optimal,carrasco2007linear}. 

To interpret smoothness and spectral decay, I recap the classic example of Sobolev spaces. Let $\mathcal{W}\subset \mathbb{R}^p$. Denote by $\mathbb{H}_2^{\nu}$ the Sobolev space with $\nu>\frac{p}{2}$ derivatives that are square integrable. This space is an RKHS generated by the Mat\`ern kernel. Suppose $\mathcal{H}=\mathbb{H}_2^{\nu}$ with $\nu>\frac{p}{2}$ is chosen as the RKHS for estimation. If $f_0\in \mathbb{H}_2^{\nu_0}$, then $c=\frac{\nu_0}{\nu}$ \cite{fischer2017sobolev}; $c$ quantifies the smoothness of $f_0$ relative to $\mathcal{H}$. In this Sobolev space, $b=\frac{2\nu}{p}>1$ \cite{edmunds2008function}. The effective dimension is increasing in the original dimension $p$ and decreasing in the degree of smoothness $\nu$.

Finally, I state the five assumptions I place in this paper, generalizing the standard RKHS learning theory assumptions for sample selection models.

\textbf{Identification.} I assume selection on observables and missingness-at-random (MAR) in order to express dose response curves as reweightings of the conditional expectation function of outcome given selection, treatment, and covariates. For distribution shift settings, I assume the shift is only in the distribution of selection, treatment, and covariates.

\textbf{RKHS regularity.} I construct an tensor product RKHS for the conditional expectation function that facilitates the analysis of partial means. The RKHS construction requires kernels that are bounded and characteristic. For incremental responses, I require kernels that are differentiable.

\textbf{Original space regularity.} I allow treatment $D$ and covariates $X$ to be discrete or continuous and low, high, or infinite dimensional. I require that each of these variables is supported on a Polish space. For simplicity, I assume outcome $Y$ is bounded.

\textbf{Smoothness.} I estimate the conditional expectation function $\gamma_0$ by kernel ridge regression and I estimate the kernel mean embedding $\mu$ by a kernel mean or generalized kernel ridge regression. To analyze the bias from ridge regularization, I assume smoothness of each object estimated by a (generalized) kernel ridge regression.

\textbf{Effective dimension.} For semiparametric inference, I assume an effective dimension condition in the RKHS construction. I unconver a double spectral robustness: some kernels in the RKHS construction may not satisfy the effective dimension condition, as long as other kernels do.

The initial four assumptions suffice for uniformly consistent nonparametric estimation. I instantiate these assumptions for static sample selection (Section~\ref{section:static}), dynamic sample selection (Section~\ref{section:dynamic}), and counterfactual distributions (Appendix~\ref{section:dist}). The fifth assumption is not necessary for uniform consistency, but is necessary for semiparametric inference. I instantiate this assumption for static sample selection (Section~\ref{section:static}) and dynamic sample selection (Section~\ref{section:dynamic}).

\section{Static sample selection}\label{section:static}

\subsection{Causal parameters}

I denote an observation by $(SY,S,D,X)$, which encodes both the selected case $(S=1)$ and unselected case $(S=0)$. In this section, I study a basic sample selection model in which the selection mechanism is determined by baseline covariates only. In other words, covariates may cause treatment but may not be caused by treatment. I denote the counterfactual outcome $Y^{(d)}$ given a hypothetical intervention on treatment $D=d$. I consider the enitre class of static causal parameters studied by \cite{singh2020kernel}.

\begin{definition}[Causal parameters]\label{def:target}
I define the following dose response curves, incremental response curves, counterfactual distributions, and treatment effects.
\begin{enumerate}
    \item $\theta_0^{ATE}(d):=\mathbb{E}[Y^{(d)}]$ is the counterfactual mean outcome given intervention $D=d$ for the entire population.
     \item $ \theta_0^{DS}(d,\tilde{\mathbb{P}}):=\mathbb{E}_{\tilde{\mathbb{P}}}[Y^{(d)}]$ is the counterfactual mean outcome given intervention $D=d$ for an alternative population with data distribution $\tilde{\mathbb{P}}$ (elaborated in Assumptions~\ref{assumption:covariate} and~\ref{assumption:covariate_planning}).
    \item $ \theta_0^{ATT}(d,d'):=\mathbb{E}[Y^{(d')}|D=d]$ is the counterfactual mean outcome given intervention $D=d'$ for the subpopulation who actually received treatment $D=d$.
     \item $\theta_0^{CATE}(d,v):=\mathbb{E}[Y^{(d)}|V=v]$ is the counterfactual mean outcome given intervention $D=d$ for the subpopulation with subcovariate value $V=v$.
\end{enumerate}
Likewise I define incremental responses, e.g. $\theta_0^{\nabla:ATE}(d):=\mathbb{E}[\nabla_d Y^{(d)}]$ and $\theta_0^{\nabla:ATT}(d,d'):=\mathbb{E}[\nabla_{d'} Y^{(d')}|D=d]$. See Appendix~\ref{section:dist} for counterfactual distributions.
\end{definition}

When treatment $D$ (and subcovariate $V$) are discrete, these causal parameters are treatment effects. When treatment $D$ (or subcovariate $V$) is continuous, these causal parameters are dose response curves. For the continuous treatment case, I additionally consider incremental response curves. For both cases, I consider counterfactual distributions. My results for means, increments, and distributions of potential outcomes immediately imply results for differences thereof. See \cite{singh2020kernel} for further discussion of these causal parameters.

\subsection{Identification}

\cite{huber2012identification,huber2014treatment} articulates distribution-free sufficient conditions under which average treatment effect with discrete treatment can be measured from observations $(SY,S,D,X)$ in a selected sample. I refer to this collection of sufficient conditions as \textit{static sample selection on observables} and demonstrate that it identifies the entire class of causal parameters in Definition~\ref{def:target}.

\begin{assumption}[Static sample selection on observables]\label{assumption:selection_static}
Assume
\begin{enumerate}
    \item No interference: if $D=d$ then $Y=Y^{(d)}$.
    \item Conditional exchangeability: $\{Y^{(d)}\}\indep D |X$ and $\{Y^{(d)}\}\indep S|D,X$.
    \item Overlap: if $f(x)>0$ then $f(d|x)>0$ and if $f(d,x)>0$ then $\mathbb{P}(S=1|d,x)>0$;  where $f(x)$, $f(d|x)$, and $f(d,x)$, are densities.
\end{enumerate}
\end{assumption}

No interference (also called the stable unit treatment value assumption) rules out network effects (also called spillovers). Conditional exchangeability states that conditional on covariates, treatment assignment is as good as random. Moreover, conditional on treatment and covariates, selection is as good as random; outcomes are missing-at-random (MAR) \cite{rubin1976inference}. Overlap ensures that there is no covariate stratum $X=x$ such that treatment has a restricted support, and there is no treatment-covariate stratum $(D,X)=(d,x)$ such that selection is impossible.

For $\theta_0^{DS}$, I also make a standard assumption in transfer learning.
\begin{assumption}[Static distribution shift]\label{assumption:covariate}
Assume
\begin{enumerate}
    \item $\tilde{\mathbb{P}}(Y,S,D,X)=\mathbb{P}(Y|S,D,X)\tilde{\mathbb{P}}(S,D,X)$;
    \item $\tilde{\mathbb{P}}(S,D,X)$ is absolutely continuous with respect to $\mathbb{P}(S,D,X)$.
\end{enumerate}
\end{assumption}
The difference in population distributions $\mathbb{P}$ and $\tilde{\mathbb{P}}$ is only in the distribution of selection, treatment, and covariates. Moreover, the support of $\mathbb{P}$ contains the support of $\tilde{\mathbb{P}}$.  An immediate consequence is that the regression function for complete cases $\mathbb{E}[Y|S=1,D=d,X=x]$ remains the same across the different populations $\mathbb{P}$ and $\tilde{\mathbb{P}}$. 

I present a general identification result below. Whereas \cite{huber2012identification,huber2014treatment} identifies $\theta_0^{ATE}(d)$, I identify the entire class of causal parameters in Definition~\ref{def:target}, which includes many new results. I prove that each causal parameter is a reweighting of the regression using complete cases, i.e.
$
\mathbb{E}(Y|S=1,D=d,X=x)
$. This nuance is critical, since the data limitations of the problem setting imply that $\mathbb{E}(Y|S=1,D=d,X=x)$ is estimable while $\mathbb{E}(Y|S=0,D=d,X=x)$ is not.

\begin{remark}[Regression notation]\label{remark:notation}
To maintain the semantics of the observation model $(SY,S,D,X)$, I formally define $\gamma_0(s,d,x):=\mathbb{E}[SY|S=s,D=d,X=x]$. Because $\gamma_0(1,d,x)=\mathbb{E}[SY|S=1,D=d,X=x]=\mathbb{E}(Y|S=1,D=d,X=x)$, I use $\mathbb{E}[SY|S=1,D=d,X=x]$ and $\mathbb{E}(Y|S=1,D=d,X=x)$ interchangeably.
\end{remark}

\begin{theorem}[Identification of causal parameters under static sample selection]\label{theorem:id_static}
If Assumption~\ref{assumption:selection_static} holds then
\begin{enumerate}
    \item $\theta_0^{ATE}(d)=\int \gamma_0(1,d,x)\mathrm{d}\mathbb{P}(x)$ \cite{huber2012identification,huber2014treatment}.
    \item If in addition Assumption~\ref{assumption:covariate} holds, then $\theta_0^{DS}(d,\tilde{\mathbb{P}})=\int \gamma_0(1,d,x)\mathrm{d}\tilde{\mathbb{P}}(x)$.
    \item $\theta_0^{ATT}(d,d')=\int \gamma_0(1,d',x)\mathrm{d}\mathbb{P}(x|d)$.
    \item $\theta_0^{CATE}(d,v)=\int \gamma_0(1,d,v,x)\mathrm{d}\mathbb{P}(x|v)$.
\end{enumerate}
For $\theta_0^{CATE}$, I use $\gamma_0(1,d,v,x)=\mathbb{E}[Y|S=1,D=d,V=v,X=x]
$. Likewise I identify incremental effects, e.g. $\theta_0^{\nabla:ATE}(d)=\int \nabla_d \gamma_0(1,d,x)\mathrm{d}\mathbb{P}(x)$. See Appendix~\ref{section:dist} for counterfactual distributions. 
\end{theorem}
See Appendix~\ref{sec:id} for the proof. 

\subsection{Algorithm: Continuous treatment}

Next, I propose nonparametric algorithm for the continuous treatment case and semiparametric algorithms for the discrete treatment case. In preparation, I propose the following RKHS construction. Define scalar valued RKHSs for selection $S$, treatment $D$, and covariates $X$. For example, the RKHS for selection is $\mathcal{H}_{\mathcal{S}}$ with feature map $\phi(s)=\phi_{\mathcal{S}}(s)$. I construct the tensor product RKHS $\mathcal{H}=\mathcal{H}_{\mathcal{S}}\otimes \mathcal{H}_{\mathcal{D}} \otimes \mathcal{H}_{\mathcal{X}}$ with feature map $\phi(s)\otimes \phi(d) \otimes \phi(x)$. In this construction, the kernel for $\mathcal{H}$ is $k(s,d,x;s',d',x')=k_\mathcal{S}(s,s')\cdot k_\mathcal{D}(d,d')\cdot k_{\mathcal{X}}(x,x')$.

For the continuous treatment case, I assume $\gamma_0\in\mathcal{H}$. Therefore by the reproducing property
$
\gamma_0(s,d,x)=\langle \gamma_0, \phi(s)\otimes \phi(d)\otimes \phi(x)\rangle_{\mathcal{H}} 
$. Likewise for $\theta_0^{CATE}$, with the extended defintion $\mathcal{H}=\mathcal{H}_{\mathcal{S}}\otimes \mathcal{H}_{\mathcal{D}} \otimes \mathcal{H}_{\mathcal{V}}\otimes \mathcal{H}_{\mathcal{X}}$. I place regularity conditions on this RKHS construction in order to represent the causal parameters as inner products in $\mathcal{H}$. In anticipation of Section~\ref{section:dynamic} and Appendix~\ref{section:dist}, I include conditions for a follow-up covariate RKHS $\mathcal{H}_{\mathcal{M}}$ and an outcome RKHS $\mathcal{H}_{\mathcal{Y}}$ in parentheses.

\begin{assumption}[RKHS regularity conditions]\label{assumption:RKHS}
Assume 
\begin{enumerate}
    \item $k_{\mathcal{S}}$ is bounded. $k_{\mathcal{D}}$, $k_{\mathcal{V}}$, $k_{\mathcal{X}}$ (and $k_{\mathcal{M}}$, $k_{\mathcal{Y}}$) are continuous and bounded. Formally,
    $
    \sup_{s\in\mathcal{S}}\|\phi(s)\|_{\mathcal{H}_{\mathcal{S}}}\leq \kappa_s$, 
    $ \sup_{d\in\mathcal{D}}\|\phi(d)\|_{\mathcal{H}_{\mathcal{D}}}\leq \kappa_d$, 
    $ \sup_{v\in\mathcal{V}}\|\phi(v)\|_{\mathcal{H}_{\mathcal{V}}}\leq \kappa_v$,
    $ \sup_{x\in\mathcal{X}}\|\phi(x)\|_{\mathcal{H}_{\mathcal{X}}}\leq \kappa_x
    $ \{and $ \sup_{m\in\mathcal{M}}\|\phi(m)\|_{\mathcal{H}_{\mathcal{M}}}\leq \kappa_m$, $ \sup_{y\in\mathcal{Y}}\|\phi(y)\|_{\mathcal{H}_{\mathcal{Y}}}\leq \kappa_y$\}.
    \item $\phi(s)$, $\phi(d)$, $\phi(v)$, $\phi(x)$ \{and $\phi(m)$, $\phi(y)$\} are measurable.
    \item $k_{\mathcal{X}}$ (and $k_{\mathcal{M}}$, $k_{\mathcal{Y}}$) are characteristic.
\end{enumerate}
For incremental responses, further assume $\mathcal{D}\subset \mathbb{R}$ is an open set and $\nabla_d  \nabla_{d'} k_{\mathcal{D}}(d,d')$ exists and is continuous, hence $\sup_{d\in\mathcal{D}}\|\nabla_d\phi(d)\|_{\mathcal{H}}\leq \kappa_d'$. For the semiparametric case, take $k_{\mathcal{D}}(d,d')=\1_{d=d'}$ instead.
\end{assumption}
Popular kernels are continuous and bounded, and measurability is a weak condition. The characteristic property ensures injectivity of the mean embeddings \cite{sriperumbudur2010relation}, so that the encoding of the reweighting distributions will be without loss. For example, the Gaussian kernel is characteristic over a continuous domain.

\begin{theorem}[Representation of dose and incremental response curves via kernel mean embeddings]\label{theorem:representation_static}
Suppose the conditions of Theorem~\ref{theorem:id_static} hold. Further suppose Assumption~\ref{assumption:RKHS} holds and $\gamma_0\in\mathcal{H}$. Then
\begin{enumerate}
    \item $\theta_0^{ATE}(d)=\langle \gamma_0, \phi(1)\otimes \phi(d)\otimes \mu_x\rangle_{\mathcal{H}} $ where $\mu_x:=\int\phi(x) \mathrm{d}\mathbb{P}(x) $.
    \item $\theta_0^{DS}(d,\tilde{\mathbb{P}})=\langle \gamma_0, \phi(1)\otimes\phi(d)\otimes \nu_x\rangle_{\mathcal{H}} $ where $\nu_x:=\int\phi(x) \mathrm{d}\tilde{\mathbb{P}}(x) $.
    \item $\theta_0^{ATT}(d,d')=\langle \gamma_0, \phi(1)\otimes\phi(d')\otimes \mu_x(d)\rangle_{\mathcal{H}} $ where $\mu_x(d):=\int\phi(x) \mathrm{d}\mathbb{P}(x|d)$.
    \item $\theta_0^{CATE}(d,v)=\langle \gamma_0, \phi(1)\otimes \phi(d)\otimes \phi(v)\otimes \mu_{x}(v)\rangle_{\mathcal{H}} $ where $\mu_{x}(v):= \int \phi(x) \mathrm{d}\mathbb{P}(x|v)$.
\end{enumerate}
Likewise for incremental responses, e.g. $\theta_0^{\nabla:ATE}(d)=\langle \gamma_0, \phi(1)\otimes\nabla_d\phi(d)\otimes \mu_x\rangle_{\mathcal{H}} $. See Appendix~\ref{section:dist} for counterfactual distributions.
\end{theorem}
See Appendix~\ref{section:alg_deriv} for the proof. The kernel mean embedding $\mu_x:=\int\phi(x) \mathrm{d}\mathbb{P}(x)$ encodes the distribution $\mathbb{P}(x)$ as an element $\mu_x\in\mathcal{H}_{\mathcal{X}}$ \cite{smola2007hilbert}. The inner product representation suggests estimators that are inner products, e.g. $\hat{\theta}^{ATE}(d)=\langle \hat{\gamma}, \phi(1)\otimes \phi(d)\otimes \hat{\mu}_x\rangle_{\mathcal{H}}$ where $\hat{\gamma}$ is a standard kernel ridge regression using the selected sample and $\hat{\mu}_x$ is an empirical mean. When a conditional distribution is encoded by a conditional mean embedding such as $\mu_x(d):=\int\phi(x) \mathrm{d}\mathbb{P}(x|d)$, I estimate $\hat{\mu}_x(d)$ by a generalized kernel ridge regression.

\begin{algorithm}[Nonparametric estimation of dose and incremental response curves]\label{algorithm:static}
Denote the empirical kernel matrices
$
K_{SS},$ $K_{DD},$ $K_{VV},$ $K_{XX}\in\mathbb{R}^{n\times n}
$
calculated from observations drawn from population $\mathbb{P}$. Let $\{\tilde{X}_i\}_{i\in[\tilde{n}]}$ be observations drawn from population $\tilde{\mathbb{P}}$. Denote by $\odot$ the elementwise product. Dose response estimators have the closed form solutions
\begin{enumerate}
    \item $\hat{\theta}^{ATE}(d)=\frac{1}{n}\sum_{i=1}^n (S\odot Y)^{\top}(K_{SS}\odot K_{DD}\odot K_{XX}+n\lambda I)^{-1}(K_{S1}\odot K_{Dd}\odot K_{Xx_i})  $;
     \item $\hat{\theta}^{DS}(d,\tilde{\mathbb{P}})=\frac{1}{\tilde{n}}\sum_{i=1}^{\tilde{n}} (S\odot Y)^{\top}(K_{SS}\odot K_{DD}\odot K_{XX}+n\lambda I)^{-1}(K_{S1}\odot K_{Dd}\odot K_{X\tilde{x}_i}) $;
    \item $\hat{\theta}^{ATT}(d,d')=(S\odot Y)^{\top}(K_{SS}\odot K_{DD}\odot K_{XX}+n\lambda I)^{-1}(K_{S1}\odot K_{Dd'}\odot [K_{XX}(K_{DD}+n\lambda_1 I)^{-1}K_{Dd}])$;
    \item $\hat{\theta}^{CATE}(d,v)=(S\odot Y)^{\top}(K_{SS}\odot K_{DD}\odot K_{VV}\odot K_{XX} +n\lambda I)^{-1}(K_{S1}\odot K_{Dd}\odot K_{Vv}\odot K_{XX}(K_{VV}+n\lambda_2 I)^{-1}K_{Vv} ) $;
\end{enumerate}
where $(\lambda,\lambda_1,\lambda_2)$ are ridge regression penalty hyperparameters.  Likewise for incremental responses, e.g. $\hat{\theta}^{\nabla:ATE}(d)=\frac{1}{n}\sum_{i=1}^n (S\odot Y)^{\top}(K_{S1}\odot K_{DD}\odot K_{XX}+n\lambda I)^{-1}(\nabla_d K_{D{d}}\odot K_{Xx_i})  $ where $[\nabla_d  K_{D{d}}]_i=\nabla_d k(D_i,d)$. See Appendix~\ref{section:dist} for counterfactual distributions.
\end{algorithm}
See Appendix~\ref{section:alg_deriv} for the derivation. See Theorem~\ref{theorem:consistency_static} for theoretical values of $(\lambda,\lambda_1,\lambda_2)$ that balance bias and variance. See Appendix~\ref{section:tuning} for a practical tuning procedure based on the closed form solution of leave-one-out cross validation to empirically balance bias and variance.

\subsection{Algorithm: Discrete treatment}

So far, I have presented nonparametric algorithms for the continuous treatment cases. Now I present semiparametric algorithms for the discrete treatment case. Towards this end, I characterize the doubly robust moment function for the static sample selection model, generalizing the result of \cite{bia2020double} to treatment effects defined for subpopulations and alternative populations. 

\begin{theorem}[Doubly robust moment for static sample selection]\label{theorem:dr_static}
Suppose treatment $D$ and subcovariate $V$ are discrete. If Assumption~\ref{assumption:selection_static} holds then each treatment effect can be expressed as
\begin{align*}
\theta^{(m)}_0
    &=\mathbb{E}\bigg[
    m(1,D,X;\gamma_0) 
   + 
   \alpha^{(m)}_0(S,D,X)\{SY-\gamma_0(1,D,X)\}
\bigg],
\end{align*}
where $\gamma_0(1,D,X)=\mathbb{E}[Y|S=1,D,X]$ is the regression for complete cases, $\gamma\mapsto \mathbb{E}[m(1,D,X;\gamma)]$ is a linear functional for complete cases, and $\alpha^{(m)}_0(S,D,X)$ is the Riesz representer to the functional. In particular,
\begin{enumerate}
   \item For $\theta^{(m)}_0=\theta_0^{ATE}(d)$, $m(1,D,X;\gamma)=\gamma(1,d,X)$ and $\alpha_0^{ATE}(S,D,X)=\frac{S\cdot \mathbbm{1}_{D=d}}{\mathbb{P}(S=1|d,X)\mathbb{P}(d|X)}$ \cite{bia2020double}.
    \item If in addition Assumption~\ref{assumption:covariate} holds, then for $\theta^{(m)}_0=\theta_0^{DS}(d,\tilde{\mathbb{P}})$, $m(1,D,X;\gamma)=\gamma(1,d,X)$ and $\alpha_0^{DS}(S,D,X)=\frac{S\cdot \mathbbm{1}_{D=d}}{\tilde{\mathbb{P}}(S=1|d,X)\tilde{\mathbb{P}}(d|X)}$.
    \item For $\theta_0^{ATT}(d,d')=\frac{\theta^{(m)}_0}{\mathbb{P}(d)}$, $m(1,D,X;\gamma)=\gamma(1,d',X)\mathbbm{1}_{D=d}$ and $\alpha_0^{ATT}(S,D,X)=\frac{S\cdot \mathbbm{1}_{D=d'}\mathbb{P}(d|X)}{\mathbb{P}(S=1|d',X)\mathbb{P}(d'|X)}$.
    \item For $\theta_0^{CATE}(d,v)=\frac{\theta^{(m)}_0}{\mathbb{P}(v)}$,  $m(1,D,V,X;\gamma)=\gamma(1,d,v,X)\mathbbm{1}_{V=v}$ and
    $\alpha_0^{CATE}(S,D,V,X)=\frac{S\cdot \mathbbm{1}_{D=d}\mathbbm{1}_{V=v}}{\mathbb{P}(S=1|d,v,X)\mathbb{P}(d|v,X)}$.
\end{enumerate}
\end{theorem}
See Appendix~\ref{sec:id} for the proof. Note that inference for the abstract treatment effect $\theta^{(m)}_0$ implies inference for $\theta_0^{ATT}(d,d')$ and $\theta_0^{CATE}(d,v)$ by delta method.

The equation in Theorem~\ref{theorem:dr_static} has three important properties. First, it consists of estimable quantities given the data limitations of the static sample selection problem. Second, it is doubly robust with respect to the nonparametric objects $(\gamma_0,\alpha^{(m)}_0)$ in the sense that it continues to hold if one of the nonparametric objects is misspecified. For example, for any $\gamma$, 
\begin{align*}
    \theta_0^{ATE}(d)
    &=\mathbb{E}\bigg[
    \gamma(1,d,X) + \alpha^{ATE}_0(S,D,X)\{SY-\gamma(1,d,X)\}
\bigg].
\end{align*}
Double robustness implies that the semiparametric guarantees below hold even if one of the nonparametric objects $(\gamma_0,\alpha^{(m)}_0)$ is not actually an element of an RKHS. Third, double robustness is with respect to $\alpha^{(m)}_0$ rather than the treatment and selection propensity scores. The technique of the kernel ridge Riesz representer permits estimation of $\alpha^{(m)}_0$ without estimation of the treatment and selection propensity scores \cite{singh2020kernel}. I adapt the kernel ridge Riesz representer next.

\begin{algorithm}[Kernel ridge Riesz representer]\label{algorithm:Riesz}
Given observations $\{W_i\}$, formula $m$, and evaluation location $w$,
\begin{enumerate}
    \item Calculate $\{K_h\}\in\mathbb{R}^{n\times n}$ where $[K_h]_{ij}=k_h(W_i,W_j)$ is specified below.
    \item Calculate $K,\Omega \in  \mathbb{R}^{2n\times 2n}$ and $u,z\in\mathbb{R}^{2n}$ by
    $$
   K=\begin{bmatrix} K_1 & K_2\\ K_3 & K_4 \end{bmatrix},\quad \Omega=\begin{bmatrix} K_1K_1 & K_1K_2\\ K_3K_1 & K_3K_2 \end{bmatrix},\quad z=\begin{bmatrix} K_2 \\  K_4 \end{bmatrix} \mathbbm{1}_{n},\quad u_i= k_1(W_i,w),\quad u_{i+n}= k_3(W_i,w)
    $$
    for $i\in\{1,...,n\}$ where $\mathbbm{1}_{n}\in\mathbb{R}^{n}$ is a vector of ones.
    \item Set
    $
    \hat{\alpha}^{(m)}(w)=z^{\top}(\Omega+n\lambda_3 K)^{-1} u.
    $
\end{enumerate}
All that remains is to specify $\{k_h\}$ for the leading examples.
\begin{enumerate}
    \item[(1,2)] For $\theta_0^{ATE}(d)$ and $\theta_0^{DS}(d,\tilde{\mathbb{P}})$,
   \begin{align*}
       k_1(W_i,W_j)&=k(S_i,S_j)k(D_i,D_j)k(X_i,X_j),\quad 
       k_2(W_i,W_j)=k(S_i,1)k(D_i,d)k(X_i,X_j) \\
       k_3(W_i,W_j)&=k(1,S_j)k(d,D_j)k(X_i,X_j),\quad 
       k_4(W_i,W_j)=k(1,1)k(d,d)k(X_i,X_j).
   \end{align*}

    \item[(3)] For the numerator of $\theta_0^{ATT}(d,d')$,
       \begin{align*}
       k_1(W_i,W_j)&=k(S_i,S_j)k(D_i,D_j)k(X_i,X_j),\quad 
       k_2(W_i,W_j)=\mathbbm{1}_{D_j=d}k(S_i,1)k(D_i,d')k(X_i,X_j) \\
       k_3(W_i,W_j)&=\mathbbm{1}_{D_i=d}k(1,S_j)k(d',D_j)k(X_i,X_j),\quad 
       k_4(W_i,W_j)=\mathbbm{1}_{D_i,D_j=d}k(1,1)k(d',d')k(X_i,X_j).
   \end{align*}
      \item[(4)] For the numerator of $\theta_0^{CATE}(d,v)$,
        \begin{align*}
       k_1(W_i,W_j)&=k(S_i,S_j)k(D_i,D_j)k(V_i,V_j)k(X_i,X_j) \quad
       k_2(W_i,W_j)=\mathbbm{1}_{V_j=v}k(S_i,1)k(D_i,d)k(V_i,v)k(X_i,X_j) \\
       k_3(W_i,W_j)&=\mathbbm{1}_{V_i=v}k(1,S_j)k(d,D_j)k(v,V_i)k(X_i,X_j) \quad 
       k_4(W_i,W_j)=\mathbbm{1}_{V_i,V_j=v}k(1,1)k(d,d)k(v,v)k(X_i,X_j).
   \end{align*}
   \end{enumerate}
\end{algorithm}

The meta algorithm that combines the doubly robust moment function from Theorem~\ref{theorem:dr_static} with sample splitting is called DML for static sample selection \cite{chernozhukov2018original,bia2020double,chernozhukov2021simple}. I instantiate DML for static sample selection using the new nonparametric estimator from Algorithm~\ref{algorithm:Riesz}.

\begin{algorithm}[Semiparametric inference of static sample selection]\label{algorithm:static_dml}
Partition the sample into folds $(I_{\ell})_{\ell=1:L}$. Denote by $I_{\ell}^c$ observations not in fold $I_{\ell}$.
\begin{enumerate}
    \item For each fold $\ell$, estimate $(\hat{\gamma}_{\ell},\hat{\alpha}_{\ell}^{(m)})$ from observations in $I_{\ell}^c$.
    \item Estimate $\hat{\theta}^{(m)}$ as
    \begin{align*}
        \hat{\theta}^{(m)}
        &=\frac{1}{n}\sum_{\ell=1}^L\sum_{i\in I_{\ell}} \bigg[
        m(1,D_i,X_i;\hat{\gamma}_{\ell}) +\hat{\alpha}^{(m)}_{\ell}(S_i,D_i,X_i)\{S_iY_i-\hat{\gamma}_{\ell}(1,D,X_i)\} ]\bigg].
   \end{align*}
    \item Estimate its $(1-a)\cdot 100$\% confidence interval as $\hat{\theta}^{(m)}\pm c_{a}\frac{\hat{\sigma}}{\sqrt{n}}$ where
    \begin{align*}
       \hat{\sigma}^2
        &=\frac{1}{n}\sum_{\ell=1}^L\sum_{i\in I_{\ell}} \bigg[
        m(1,D_i,X_i;\hat{\gamma}_{\ell}) +\hat{\alpha}^{(m)}_{\ell}(S_i,D_i,X_i)\{S_iY_i-\hat{\gamma}_{\ell}(1,D,X_i)\} -\hat{\theta}^{(m)}\bigg]^2,
    \end{align*}
    and $c_{a}$ is the $1-\frac{a}{2}$ quantile of the standard Gaussian. 
\end{enumerate}
\end{algorithm}

Unlike \cite[Algorithm 1]{bia2020double}, the proposed procedure does not require explicit propensity score estimation. Note that $\hat{\alpha}^{(m)}$ is computed according to Algorithm~\ref{algorithm:Riesz} with regularization parameter $\lambda_3$. $\hat{\gamma}$ is a standard kernel ridge regressions with regularization parameter $\lambda$. See Appendix~\ref{section:alg_deriv} for explicit computations. See Theorem~\ref{theorem:inference_static} for theoretical values of regularization parameters that balance bias and variance. See Appendix~\ref{section:tuning} for a practical tuning procedure based on the closed form solution for leave-one-out cross validation that empirically balances bias and variance. 

\subsection{Guarantees: Continuous treatment}

Towards formal guarantees, I place weak regularity conditions on the original spaces. In anticipation of Section~\ref{section:dynamic} and Appendix~\ref{section:dist}, I also include conditions for the follow-up covariate space and outcome space in parentheses.
\begin{assumption}[Original space regularity conditions]\label{assumption:original}
Assume
\begin{enumerate}
    \item $\mathcal{D}$, $\mathcal{X}$ (and $\mathcal{M}$, $\mathcal{Y}$) are Polish spaces, i.e, separable and completely metrizable topological spaces.
    \item $Y\in\mathbb{R}$ and it is bounded, i.e. there exists $C<\infty$ such that $|Y|\leq C$ almost surely.
\end{enumerate}
\end{assumption}
This weak regularity condition allows treatment and covariates to be discrete or continuous and low, high, or infinite dimensional.  The assumption that $Y\in\mathbb{R}$ is bounded simplifies the analysis, but may be relaxed.

Next, I assume that the various objects estimated by kernel ridge regressions and generalized kernel ridge regressions are smooth in the sense of~\eqref{eq:prior} in Section~\ref{sec:rkhs_main}. For dose response and incremental response curves, these are the regression $\gamma_0$ and conditional mean embeddings $(\mu_{x}(d),\mu_x(v))$. The smoothness assumption for the regression is as in Section~\ref{sec:rkhs_main}.

\begin{assumption}[Smoothness of regression]\label{assumption:smooth_gamma}
Assume $\gamma_0\in\mathcal{H}^c$.
\end{assumption}

For conditional mean embeddings, I follow the abstract statement of smoothness from \cite{singh2020kernel}. Define the abstract conditional mean embedding $\mu_{a}(b):=\int \phi(a)\mathrm{d}\mathbb{P}(a|b)$ where $a\in\mathcal{A}_j$ and $b\in\mathcal{B}_j$. I parametrize the smoothness of $\mu_{a}(b)$, with $a\in\mathcal{A}_j$ and $b\in\mathcal{B}_j$, by $c_j$. Define the abstract conditional expectation operator $E_j:\mathcal{H}_{\mathcal{A}_j}\rightarrow\mathcal{H}_{\mathcal{B}_j}$, $f(\cdot)\mapsto \mathbb{E}\{f(A_j)|B_j=\cdot\}$. The relationship between $E_j$ and $\mu_{a}(b)$ is given by
$$
\mu_{a}(b)=\int \phi(a)\mathrm{d}\mathbb{P}(a|b) =\mathbb{E}_{A_j|B_j=b} \{\phi(A_j)\} =E_j^* \{\phi(b)\},\quad a\in\mathcal{A}_j, b\in\mathcal{B}_j
$$
where $E_j^*$ is the adjoint of $E_j$. Let $\mathcal{L}_2(\mathcal{H}_{\mathcal{A}_j},\mathcal{H}_{\mathcal{B}_j})$ be the space of Hilbert-Schmidt operators between $\mathcal{H}_{\mathcal{A}_j}$ and $\mathcal{H}_{\mathcal{B}_j}$ by $\mathcal{L}_2(\mathcal{H}_{\mathcal{A}_j},\mathcal{H}_{\mathcal{B}_j})$, which is also an RKHS \cite{grunewalder2013smooth,singh2019kernel}. With this notation, I state the abstract statement of smoothness.

\begin{assumption}[Smoothness of mean embedding]\label{assumption:smooth_op}
Assume $E_j\in \{\mathcal{L}_2(\mathcal{H}_{\mathcal{A}_j},\mathcal{H}_{\mathcal{B}_j})\}^{^{c_j}}$.
\end{assumption}

With these assumptions, I arrive at the first main result.
\begin{theorem}[Uniform consistency of dose and incremental responses curves]\label{theorem:consistency_static}
Suppose Assumptions~\ref{assumption:selection_static},~\ref{assumption:RKHS},~\ref{assumption:original}, and~\ref{assumption:smooth_gamma} hold. Set $(\lambda,\lambda_1,\lambda_2)=(n^{-\frac{1}{c+1}},n^{-\frac{1}{c_1+1}},n^{-\frac{1}{c_2+1}})$.
\begin{enumerate}
    \item Then
    $
    \|\hat{\theta}^{ATE}-\theta_0^{ATE}\|_{\infty}=O_p\left(n^{-\frac{1}{2}\frac{c-1}{c+1}}\right).
    $
    \item If in addition Assumption~\ref{assumption:covariate} holds, then 
      $
    \|\hat{\theta}^{DS}(\cdot,\tilde{\mathbb{P}})-\theta_0^{DS}(\cdot,\tilde{\mathbb{P}})\|_{\infty}=O_p\left( n^{-\frac{1}{2}\frac{c-1}{c+1}}+\tilde{n}^{-\frac{1}{2}}\right).
    $
    \item If in addition Assumption~\ref{assumption:smooth_op} holds with $\mathcal{A}_1=\mathcal{X}$ and $\mathcal{B}_1=\mathcal{D}$, then
      $
    \|\hat{\theta}^{ATT}-\theta_0^{ATT}\|_{\infty}=O_p\left(n^{-\frac{1}{2}\frac{c-1}{c+1}}+n^{-\frac{1}{2}\frac{c_1-1}{c_1+1}}\right).
    $
    \item If in addition Assumption~\ref{assumption:smooth_op} holds with $\mathcal{A}_2=\mathcal{X}$ and $\mathcal{B}_2=\mathcal{V}$, then
      $
    \|\hat{\theta}^{CATE}-\theta_0^{CATE}\|_{\infty}=O_p\left(n^{-\frac{1}{2}\frac{c-1}{c+1}}+n^{-\frac{1}{2}\frac{c_2-1}{c_2+1}}\right).
    $
\end{enumerate}
Likewise for incremental responses, e.g.  $
    \|\hat{\theta}^{\nabla:ATE}-\theta_0^{\nabla:ATE}\|_{\infty}=O_p\left(n^{-\frac{1}{2}\frac{c-1}{c+1}}\right)
    $. See Appendix~\ref{section:dist} for counterfactual distributions.
\end{theorem}
See Appendix~\ref{section:consistency_proof} for the proof and exact finite sample rates. The rates are at best $n^{-\frac{1}{6}}$ when $(c,c_1,c_2)=2$. The slow rates reflect the challenge of a $\sup$ norm guarantee, which encodes caution about worst case scenarios when informing policy decisions. 

\subsection{Guarantees: Discrete treatment}

Consider the semiparametric case where $D$ is binary. Then $\theta_0^{ATE}(d)$ is a vector in $\mathbb{R}^{2}$ and Theorem~\ref{theorem:consistency_static} simplifies to a guarantee on the maximum element of the vector of differences $|\hat{\theta}^{ATE}(d)-\theta_0^{ATE}(d)|$. For the semiparametric case, I improve the rate from $n^{-\frac{1}{6}}$ to $n^{-\frac{1}{2}}$ by using Algorithm~\ref{algorithm:static_dml} and imposing additional assumptions. In particular, I place additional assumptions on the regression propensity scores, Riesz representer, and corresponding RKHSs. 

\begin{assumption}[Bounded propensities and Riesz representer]\label{assumption:bounded_Riesz}
Assume there exist $\epsilon>0$ and $\bar{\alpha}'<\infty$ such that
\begin{enumerate}
    \item $\epsilon \leq \mathbb{P}(D=1|X)\leq 1-\epsilon $ and $\epsilon \leq \mathbb{P}(S=1|D,X)\leq 1-\epsilon $ with probability one;
    \item $\|\hat{\alpha}^{(m)}\|_{\infty}\leq \bar{\alpha}'$ with probability approaching one.
\end{enumerate}
Likewise for the propensity scores in the cases of $\theta_0^{DS}$ and $\theta_0^{CATE}$.
\end{assumption}
Part one of Assumption~\ref{assumption:bounded_propensity_planning} is a mild strengthening of the overlap condition in Assumption~\ref{assumption:selection_planning}. It implies boundedness of the Riesz representer $\alpha_0^{(m)}$ and mean square continuity of the corresponding functional \cite[Assumption 3.1]{chernozhukov2021simple}. Part two can be imposed by censoring extreme evaluations of the estimator $\hat{\alpha}^{(m)}$. Note that censoring can only improve prediction quality because $\alpha^{(m)}_0$ is bounded as an implication of part one. 

Finally, I assume smoothness and spectral decay assumptions for the Riesz representer as described in Section~\ref{sec:rkhs_main}.

\begin{assumption}[Smoothness of Riesz representer]\label{assumption:smooth_Riesz}
Assume $\alpha^{(m)}_0\in \mathcal{H}^{c_3}$.
\end{assumption}

To quantify the effective dimension of the Riesz representer kernel, recall the convolution operator notation from Section~\ref{sec:rkhs_main}: $\lambda_j(k)$ is the $j$-th eigenvalue of the convolution operator $L:f\mapsto \int k(\cdot, w)f(w)\mathrm{d}\mathbb{P}(w)$ of the kernel $k$.

\begin{assumption}[Effective dimension of Riesz representer]\label{assumption:spectral_Riesz}
Assume $\lambda_j(k_{\mathcal{S}}\cdot k_{\mathcal{D}} \cdot  k_{\mathcal{X}})\asymp j^{-b_3}$.
\end{assumption}

The latter statement quantifies the effective dimension of $\alpha_0^{(m)} \in \mathcal{H}$. Recall from Section~\ref{sec:rkhs_main} that Sobolev kernels possess this property. See \cite[Table 1]{scetbon2021spectral} for additional examples. In Appendix~\ref{sec:inference_proof}, I argue that it is sufficient for \textit{either} the kernel of the Riesz representer $\alpha_0^{(m)}$ to have polynomial spectral decay \textit{or} for the kernel of the regression $\gamma_0$ to have polynomial spectral decay---a kind of double spectral robustness \cite{singh2020kernel,singh2021workshop}.

\begin{theorem}[Semiparametic consistency, Gaussian approximation, and efficiency for static sample selection]\label{theorem:inference_static}
Suppose the conditions of Theorem~\ref{theorem:consistency_static} hold, as well as Assumptions~\ref{assumption:bounded_Riesz},~\ref{assumption:smooth_Riesz}, and~\ref{assumption:spectral_Riesz}. For a given treatment effect, denote the moments of 
$$
\psi^{(m)}_0(W):=m(1,D,X;\gamma_0)+\alpha^{(m)}_0(S,D,X)[SY-\gamma_0(1,D,X)]-\theta^{(m)}_0$$ by
$\sigma^2=\mathbb{E}[\psi_0(W)^2]$, $\eta^3=\mathbb{E}[|\psi_0(W)|^3]$, and $\chi^4=\mathbb{E}[\psi_0(W)^4]$. Assume the regularity conditions that $\sigma^2$ is bounded away from zero and $\left\{\left(\frac{\eta}{\sigma}\right)^3+\chi^2\right\}n^{-\frac{1}{2}}\rightarrow0$. Set $\lambda=n^{-\frac{1}{2}\frac{c-1}{c+1}}$. Set $\lambda_3=n^{-\frac{1}{2}}$ if $b_3=\infty$ and $\lambda_3=n^{-\frac{b_3}{b_3c_3+1}}$ otherwise. Then for any $c,c_3\in(1,2]$ and $b_3\in(1,\infty]$ satisfying
$$
\frac{b_3c_3}{b_3c_3+1}+\frac{c-1}{c+1}>1,
$$
the abstract estimator $\hat{\theta}^{(m)}$ is consistent, i.e. $\hat{\theta}^{(m)}\overset{p}{\rightarrow}\theta_0^{(m)}$, and  the confidence interval includes $\theta_0^{(m)}$ with probability approaching the nominal level, i.e. $$\lim_{n\rightarrow\infty} \mathbb{P}\left[\theta_0^{(m)}\in \left\{\hat{\theta}^{(m)}\pm c_a\frac{\hat{\sigma}}{\sqrt{n}}\right\}\right]=1-a.$$
\end{theorem}
See Appendix~\ref{sec:inference_proof} for the proof and a finite sample Gaussian approximation. Observe the tradeoff among the smoothness and effective dimension assumptions across nonparametric objects. The spectral decay of the Riesz representer $\alpha_0^{(m)}$ parametrized by $b_3$ must be sufficiently fast relative to the smoothness of the nonparametric objects $(\gamma_0,\alpha^{(m)}_0)$ parametrized by $(c,c_3)$. To see that this parameter region is nonempty, observe that if $(c,c_3)=2$ then the condition simplifies to $b_3>1$. In other words, when the various nonparametric objects are smooth, it is sufficient for the Riesz representer to belong to an RKHS with any polynomial rate of spectral decay. 
\section{Dynamic sample selection}\label{section:dynamic}

\subsection{Causal parameters}

So far I have studied the static sample selection model where selection is as good as random conditional on treatment $D$ and baseline covariates $X$. A richer class of sample selection models allows for selection to be as good as random conditional on treatment $D$, baseline covariates $X$, and follow-up covariates $M$ measured after treatment. In other words I open up the possibility that, among the covariates in the selection mechanism, some are baseline covariates $X$ that may cause treatment while others are follow-up covariates $M$ that may be caused by treatment. I refer to this setting as dynamic sample selection. The dynamic sample selection case is more challenging from a statistical perspective. I study a subset of the causal parameters from Section~\ref{section:static} due to this additional complexity, though in principle the same techniques should apply to the entire class of causal parameters.

\subsection{Identification}

\cite{bia2020double} articulate distribution-free sufficient conditions under which average treatment effect can be measured from observations $(SY,S,D,X,M)$. I refer to this collection of sufficient conditions as \textit{dynamic sample selection on observables} and argue that it identifies additional causal parameters from Definition~\ref{def:target}.

\begin{assumption}[Dynamic sample selection on observables]\label{assumption:selection_planning}
Assume
\begin{enumerate}
    \item No interference: if $D=d$ then $Y=Y^{(d)}$.
    \item Conditional exchangeability: $\{Y^{(d)}\}\indep D |X$ and $\{Y^{(d)}\}\indep S|D,X,M$.
    \item Overlap: if $f(x)>0$ then $f(d|x)>0$ and if $f(d,x,m)>0$ then $\mathbb{P}(S=1|d,x,m)>0$;  where $f(x)$, $f(d|x)$, and $f(d,x,m)$, are densities.
\end{enumerate}
\end{assumption}

Observe that Assumption~\ref{assumption:selection_planning} is a generalization of Assumption~\ref{assumption:selection_static}. In particular, $\{Y^{(d)}\}\indep S|D,X,M$ generalizes $\{Y^{(d)}\}\indep S|D,X$. In this sense, the dynamic sample selection model generalizes missing-at-random (MAR) \cite{rubin1976inference}. The updated overlap condition is more stringent: there is no treatment-baseline-follow up stratum $(D,X,M)=(d,x,m)$ such that selection is impossible. Similarly, I adapt Assumption~\ref{assumption:covariate} to handle $\theta_0^{DS}$.

\begin{assumption}[Dynamic distribution shift]\label{assumption:covariate_planning}
Assume
\begin{enumerate}
    \item $\tilde{\mathbb{P}}(Y,S,D,X,M)=\mathbb{P}(Y|S,D,X,M)\tilde{\mathbb{P}}(S,D,X,M)$;
    \item $\tilde{\mathbb{P}}(S,D,X,M)$ is absolutely continuous with respect to $\mathbb{P}(S,D,X,M)$.
\end{enumerate}
\end{assumption}
As before, difference in population distributions $\mathbb{P}$ and $\tilde{\mathbb{P}}$ is only in the distribution of selection, treatments, and covariates.  An immediate consequence is that the regression function for complete cases $\gamma_0(1,d,x,m)=\mathbb{E}(Y|S=1,D=d,X=x,M=m)$ remains the same across the different populations $\mathbb{P}$ and $\tilde{\mathbb{P}}$. 

Similar to Theorem~\ref{theorem:id_static}, I argue that different causal parameters are reweightings of the regression using complete cases, i.e. $\gamma_0(1,d,x,m)$. This nuance is critical, since the data limitations of the problem setting imply that $\mathbb{E}(Y|S=1,D=d,X=x,M=m)$ is estimable while $\mathbb{E}(Y|S=0,D=d,X=x,M=m)$ is not. The reweightings are more complex due to dynamic sample selection. I modestly extend the argument of \cite{bia2020double} to cover the additional cases of distribution shift and counterfactual distributions.

\begin{theorem}[Identification of causal parameters under dynamic sample selection]\label{theorem:id_planning}
If Assumption~\ref{assumption:selection_planning} holds then
\begin{enumerate}
    \item $
\theta_0^{ATE}(d)=\int \gamma_0(1,d,x,m) \mathrm{d}\mathbb{P}(m|d,x) \mathrm{d}\mathbb{P}(x)
$ \cite{bia2020double}.
    \item If in addition Assumption~\ref{assumption:covariate_planning} holds then  $
\theta_0^{DS}(d,\tilde{\mathbb{P}})=\int \gamma_0(1,d,x,m) \mathrm{d}\tilde{\mathbb{P}}(m|d,x) \mathrm{d}\tilde{\mathbb{P}}(x)
$.
\end{enumerate}
See Appendix~\ref{section:dist} for counterfactual distributions.
\end{theorem}
See Appendix~\ref{sec:id} for the proof.

\subsection{Algorithm: Continuous treatment}

I extend the RKHS construction from Section~\ref{section:static} to accommodate the follow-up covariates $M$. In particular, I define an additional scalar valued RKHS $\mathcal{H}_{\mathcal{M}}$ and construct the tensor product RKHS $\mathcal{H}=\mathcal{H}_{\mathcal{S}}\otimes \mathcal{H}_{\mathcal{D}} \otimes \mathcal{H}_{\mathcal{X}}\otimes \mathcal{H}_{\mathcal{M}}$ with feature map $\phi(s)\otimes \phi(d) \otimes \phi(x) \times \phi(m)$. In this construction, the kernel for $\mathcal{H}$ is $k(s,d,x,m;s',d',x',m')=k_\mathcal{S}(s,s')\cdot k_\mathcal{D}(d,d')\cdot  k_{\mathcal{X}}(x,x')\cdot k_{\mathcal{M}}(m,m')$.

As before, by the reproducing property, $\gamma_0\in \mathcal{H}$ implies $
\gamma_0(s,d,x,m)=\langle \gamma_0, \phi(s)\otimes \phi(d)\otimes \phi(x) \otimes \phi(m) \rangle_{\mathcal{H}} 
$. The regularity conditions in Assumption~\ref{assumption:RKHS} help to represent the causal parameters as inner products in $\mathcal{H}$.

\begin{theorem}[Representation of dose response curves via kernel mean embeddings]\label{theorem:representation_planning}
Suppose the conditions of Theorem~\ref{theorem:id_planning} hold. Further suppose Assumption~\ref{assumption:RKHS} holds and $\gamma_0\in\mathcal{H}$. Then
\begin{enumerate}
    \item $\theta_0^{ATE}(d)=\langle \gamma_0,  \phi(1)\otimes\phi(d) \otimes \int \{\phi(x)\otimes \mu_{m}(d,x)\}\mathrm{d} \mathbb{P}(x) \rangle_{\mathcal{H}}$ where $\mu_{m}(d,x):=\int \phi(m)\mathrm{d}\mathbb{P}(m|d,x)$,
    \item $\theta_0^{DS}(d,\tilde{\mathbb{P}})=\langle \gamma_0,  \phi(1)\otimes\phi(d) \otimes \int \{\phi(x)\otimes \nu_{m}(d,x)\}\mathrm{d} \tilde{\mathbb{P}}(x) \rangle_{\mathcal{H}}$ where $\nu_{m}(d,x):=\int \phi(m)\mathrm{d}\tilde{\mathbb{P}}(m|d,x)$.
\end{enumerate}
See Appendix~\ref{section:dist} for counterfactual distributions.
\end{theorem}
See Appendix~\ref{section:alg_deriv} for the proof. The conditional mean embedding $\mu_{m}(d,x):=\int \phi(m)\mathrm{d}\mathbb{P}(m|d,x)$ encodes the conditional distribution $\mathbb{P}(m|d,x)$. The sequential mean embedding $\int \{\phi(x)\otimes \mu_{m}(d,x)\}\mathrm{d}\mathbb{P}(x)$ encodes the counterfactual distribution of the sequence of covariates $(X,M)$ when treatment is $D=d$, adapting a technique introduced by \cite{singh2021workshop}. The inner product representation suggests estimators that are inner products, e.g. $\hat{\theta}^{ATE}(d)=\langle \hat{\gamma},  \phi(1)\otimes\phi(d) \otimes \frac{1}{n}\sum_{i=1}^n\phi(X_i)\otimes \hat{\mu}_{m}(d,X_i)\} \rangle_{\mathcal{H}}$ where $\hat{\gamma}$ is a standard kernel ridge regression using the selected sample and $\hat{\mu}_{m}(d,x)$ is a generalized kernel ridge regression.

\begin{algorithm}[Nonparametric estimation of dose response curves]\label{algorithm:planning}
Denote the empirical kernel matrices
$
K_{SS}$, $K_{DD}$, $K_{XX}$, $K_{MM}\in\mathbb{R}^{n\times n}
$
calculated from observations drawn from population $\mathbb{P}$.  Denote the empirical kernel matrices
$
K_{\tilde{D}\tilde{D}}$, $K_{\tilde{X}\tilde{X}}$, $K_{\tilde{M}\tilde{M}}\in\mathbb{R}^{\tilde{n}\times \tilde{n}}
$
calculated from observations drawn from population $\tilde{\mathbb{P}}$. Dose response estimators have the closed form solutions
\begin{enumerate}
\item $\hat{\omega}(1,d;x)=(S \odot Y)^{\top}(K_{SS}\odot K_{DD}\odot K_{XX}\odot K_{MM}+n\lambda I)^{-1}$ \\
        $[K_{S1}\odot K_{Dd} \odot  K_{Xx}\odot 
        \{K_{MM}(K_{DD}\odot K_{XX}+n\lambda_4 I)^{-1}(K_{Dd}\odot K_{Xx})\}]  $,
    \item $\hat{\theta}^{ATE}(d)=\frac{1}{n}\sum_{i=1}^n \hat{\omega}(1,d;X_i)$,
    \item $\hat{\theta}^{DS}(d,\tilde{\mathbb{P}})=\frac{1}{n}\sum_{i=1}^{\tilde{n}} (S \odot Y)^{\top}(K_{SS}\odot K_{DD}\odot K_{XX}\odot K_{MM}+n\lambda I)^{-1}$ \\
        $[K_{S1}\odot K_{Dd} \odot  K_{X\tilde{x}_{i}}\odot \{K_{M \tilde{M}}(K_{\tilde{D}\tilde{D}}\odot K_{\tilde{X}\tilde{X}}+\tilde{n}\lambda_5 I)^{-1}(K_{\tilde{D}d}\odot K_{\tilde{X}\tilde{x}_{i}})\}]  $,
\end{enumerate}
where $(\lambda,\lambda_4,\lambda_5)$ are ridge regression penalty parameters.\footnote{Appendix~\ref{section:alg_deriv} presents alternative closed form expressions that run faster in numerical packages.} See Appendix~\ref{section:dist} for counterfactual distributions.
\end{algorithm}

See Appendix~\ref{section:alg_deriv} for the derivation. See Theorem~\ref{theorem:consistency_planning} for theoretical values of $(\lambda,\lambda_4,\lambda_5)$ that balance bias and variance. See Appendix~\ref{section:tuning} for a practical tuning procedure based on the closed form solution of leave-one-out cross validation to empirically balance bias and variance.

\subsection{Algorithm: Discrete treatment}

So far, I have presented nonparametric algorithms for the continuous treatment cases. Now I present semiparametric algorithms for the discrete treatment case. Towards this end, I quote the multiply robust moment function for the dynamic sample selection model parametrized to avoid density estimation.
\begin{lemma}[Multiply robust moment for dynamic sample selection \cite{bia2020double}]\label{theorem:dr_planning}
\; \\ Suppose treatment $D$ is binary, so that $d\in\{0,1\}$. If Assumption~\ref{assumption:selection_planning} holds then
\begin{align*}
\theta^{ATE}_0(d)
    &=\mathbb{E}\bigg[
    \omega_0(1,d;X) 
   + \frac{\1_{D=d}S}{\pi_0(d;X)\rho_0(1;d,X,M)}\{SY-\gamma_0(1,d,X,M)\}\\
    &\quad + \frac{\1_{D=d}}{\pi_0(d;X)}\left\{\gamma_0(1,d,X,M)-\omega_0(1,d;X)\right\}
\bigg],
\end{align*}
where $\omega_0(1,d;X):=\int \gamma_0(1,d,X,m) \mathrm{d}\mathbb{P}(m|d,X)$ is the sequential mean embedding, $ \pi_0(d;X):=\mathbb{P}(d|X)$ is the treatment propensity score, and $\rho_0(1;d,X,M):=\mathbb{P}(S=1|d,X,M)$ is the selection propensity score.
\end{lemma}
The equation in Lemma~\ref{theorem:dr_planning} has three important properties. First, it consists of estimable quantities given the data limitations of the dynamic sample selection problem. Second, it is multiply robust with respect to the nonparametric objects $(\gamma_0,\omega_0,\pi_0,\rho_0)$ in the sense that it continues to hold if one of the nonparametric objects is misspecified. For example, for any $\gamma$, 
\begin{align*}
    \theta_0^{ATE}(d)
    &=\mathbb{E}\bigg[
    \omega_0(1,d;X)  + \frac{\1_{D=d}S}{\pi_0(d;X)\rho_0(1;d,X,M)}\{SY-\gamma(1,d,X,M)\} \\
    &\quad + \frac{\1_{D=d}}{\pi_0(d;X)}\left\{\gamma(1,d,X,M)-\omega_0(1,d;X)\right\}
\bigg].
\end{align*}
Multiple robustness implies that the semiparametric guarantees below hold even if one of the nonparametric objects $(\gamma_0,\omega_0,\pi_0,\rho_0)$ is not actually an element of an RKHS. Third, multiple robustness is with respect to $\omega_0$ rather than the conditional density $f(m|d,x)$. The technique of sequential mean embeddings permits estimation of $\omega_0$ without estimation of $f(m|d,x)$ \cite{singh2021workshop}.

The meta algorithm that combines the multiply robust moment function from Lemma~\ref{theorem:dr_planning} with sample splitting is called DML for dynamic sample selection \cite{chernozhukov2018original,bia2020double,chernozhukov2021simple}. I instantiate DML for dynamic sample selection using the new nonparametric estimator from Algorithm~\ref{algorithm:planning}.

\begin{algorithm}[Semiparametric inference for dynamic sample selection]\label{algorithm:planning_dml}
Partition the sample into folds $(I_{\ell})_{\ell=1:L}$. Denote by $I_{\ell}^c$ observations not in fold $I_{\ell}$.
\begin{enumerate}
    \item For each fold $\ell$, estimate $(\hat{\gamma}_{\ell},\hat{\omega}_{\ell},\hat{\pi}_{\ell},\hat{\rho}_{\ell})$ from observations in $I_{\ell}^c$. 
    \item Estimate $\hat{\theta}^{ATE}(d)$ as
    \begin{align*}
        \hat{\theta}^{ATE}(d) 
        &=\frac{1}{n}\sum_{\ell=1}^L\sum_{i\in I_{\ell}} \bigg[
        \hat{\omega}_{\ell}(1,d;X_i)
        +\frac{\1_{D_i=d}S_i}{\hat{\pi}_{\ell}(d;X_i)\hat{\rho}_{\ell}(1;d,X_i,M_i)}\{S_iY_i-\hat{\gamma}_{\ell}(1,d,X_i,M_i)\} \\
        &\quad +\frac{\1_{D_i=d}}{\hat{\pi}_{\ell}(d;X_i)}\left\{\hat{\gamma}_{\ell}(1,d,X_i,M_i)-\hat{\omega}_{\ell}(1,d;X_i)\right\}\bigg].
   \end{align*}
    \item Estimate its $(1-a)\cdot 100$\% confidence interval as $\hat{\theta}^{ATE}(d)\pm c_{a}\frac{\hat{\sigma}(d)}{\sqrt{n}}$ where
    \begin{align*}
       \hat{\sigma}^2(d)
        &=\frac{1}{n}\sum_{\ell=1}^L\sum_{i\in I_{\ell}} \bigg[
        \hat{\omega}_{\ell}(1,d;X_i)
        +\frac{\1_{D_i=d}S_i}{\hat{\pi}_{\ell}(d;X_i)\hat{\rho}_{\ell}(1;d,X_i,M_i)}\{S_iY_i-\hat{\gamma}_{\ell}(1,d,X_i,M_i)\}\} \\
        &\quad +\frac{\1_{D_i=d}}{\hat{\pi}_{\ell}(d;X_i)}\left\{\hat{\gamma}_{\ell}(1,d,X_i,M_i)-\hat{\omega}_{\ell}(1,d;X_i)\right\} -\hat{\theta}^{ATE}(d)\bigg]^2,  
    \end{align*}
    and $c_{a}$ is the $1-\frac{a}{2}$ quantile of the standard Gaussian. 
\end{enumerate}
\end{algorithm}

See Appendix~\ref{section:alg_deriv} for explicit computations. The sequential mean embedding $\hat{\omega}$ is computed according to Algorithm~\ref{algorithm:planning} with regularization parameters $(\lambda,\lambda_4)$. The remaining nonparametric subroutines $(\hat{\gamma},\hat{\pi},\hat{\rho})$ are standard kernel ridge regressions with regularization parameters $(\lambda,\lambda_6,\lambda_7)$. See Theorem~\ref{theorem:inference_planning} for  theoretical values of the regularization parameters that balance bias and variance. See Appendix~\ref{section:tuning} for a practical tuning procedure based on the closed form solution of leave-one-out cross validation to empirically balance bias and variance. 

Unlike \cite[Algorithm 3]{bia2020double}, Algorithm~\ref{algorithm:planning_dml} does not require multiple levels of sample splitting. Crucially, I am able to prove a new, sufficiently fast rate of convergence on the sequential mean embedding $\hat{\omega}$ from Algorithm~\ref{algorithm:planning} without additional sample splitting, which allows for this simplification. 

\subsection{Guarantees: Continuous treatment}

As in Section~\ref{section:static}, I place regularity conditions on the original spaces, assume the regression $\gamma_0$ is smooth, and assume the conditional mean embeddings $\mu_{m}(d,x)$ and $\nu_{m}(d,x)$ are smooth in order to prove uniform consistency for the continuous treatment case.

\begin{theorem}[Uniform consistency of dose and incremental response curves]\label{theorem:consistency_planning}
Suppose Assumptions~\ref{assumption:selection_planning},~\ref{assumption:RKHS},~\ref{assumption:original}, and~\ref{assumption:smooth_gamma} hold. Set $(\lambda,\lambda_4,\lambda_5)$= $(n^{-\frac{1}{c+1}},n^{-\frac{1}{c_4+1}},\tilde{n}^{-\frac{1}{c_5+1}})$.
\begin{enumerate}
    \item If in addition Assumption~\ref{assumption:smooth_op} holds with $\mathcal{A}_4=\mathcal{X}$ and $\mathcal{B}_4=\mathcal{D}\times \mathcal{X}$, then
$
\|\hat{\theta}^{ATE}-\theta_0^{ATE}\|_{\infty}=O_p\left(n^{-\frac{1}{2}\frac{c-1}{c+1}}+n^{-\frac{1}{2}\frac{c_4-1}{c_4+1}}\right).
$
    \item If in addition Assumptions~\ref{assumption:covariate_planning} and~\ref{assumption:smooth_op} hold with $\mathcal{A}_5=\mathcal{X}$ and $\mathcal{B}_5=\mathcal{D}\times \mathcal{X}$, then
    $
\|\hat{\theta}^{DS}-\theta_0^{DS}\|_{\infty}=O_p\left(n^{-\frac{1}{2}\frac{c-1}{c+1}}+\tilde{n}^{-\frac{1}{2}\frac{c_5-1}{c_5+1}}\right)
$.
\end{enumerate}
See Appendix~\ref{section:dist} for counterfactual distributions.
\end{theorem}
See Appendix~\ref{section:consistency_proof} for the proof and exact finite sample rates. As before, these rates are at best $n^{-\frac{1}{6}}$ when $(c,c_4,c_5)=2$.

\subsection{Guarantees: Discrete treatment}

Consider the semiparametric case where $D$ is binary. Then $\theta_0^{ATE}(d)$ is a vector in $\mathbb{R}^{2}$ and Theorem~\ref{theorem:consistency_planning} simplifies to a guarantee on the maximum element of the vector of differences $|\hat{\theta}^{ATE}(d)-\theta_0^{ATE}(d)|$. For the semiparametric case, I improve the rate from $n^{-\frac{1}{6}}$ to $n^{-\frac{1}{2}}$ by using Algorithm~\ref{algorithm:planning_dml} and imposing additional assumptions.

Since $\gamma_0(1,d,x,m)=\mathbb{E}(Y|S=1,D=d,X=x,M=m)$, I refer to $Y-\gamma_0(1,D,X,M)$ as the regression residual. I assume that the regression residual has non-degenerate variance.
\begin{assumption}[Non-degenerate residual]\label{assumption:resid_planning}
Assume there exists $\chi>0$ such that $\inf_{d}\mathbb{E}[\{Y-\gamma_0(1,d,X,M)\}^2]\geq \chi$.
\end{assumption}

To lighten notation for the binary treatment case, let $\pi_0(x):=\pi_0(1;x)=\mathbb{P}(D=1|X=x)$ and $\rho_0(d,x,m):=\rho_0(1;d,x,m)=\mathbb{P}(S=1|D=d,X=x,M=m)$. I assume these treatment and selection propensity scores are bounded away from zero and one.

\begin{assumption}[Bounded propensities]\label{assumption:bounded_propensity_planning}
Assume propensity scores are bounded away from zero and one, i.e. there exists $\epsilon>0$ such that
\begin{enumerate}
    \item $\epsilon \leq \pi_0(X)\leq 1-\epsilon $ and $\epsilon \leq \rho_0(D,X,M)\leq 1-\epsilon $ with probability one;
    \item $\epsilon \leq \inf_{x}  \hat{\pi}(x)\leq \sup_{x} \hat{\pi}(x) \leq  1-\epsilon$ and $\epsilon \leq \inf_{d,x,m}  \hat{\rho}(d,x,m)\leq \sup_{d,x,m}  \hat{\rho}(d,x,m) \leq  1-\epsilon$ with probability approaching one.
\end{enumerate}
\end{assumption}
Part one of Assumption~\ref{assumption:bounded_propensity_planning} is a mild strengthening of the overlap condition in Assumption~\ref{assumption:selection_planning}. Part two can be imposed by censoring extreme evaluations of the estimators $(\hat{\pi},\hat{\rho})$. Note that censoring can only improve prediction quality because $(\pi_0,\rho_0)$ are bounded away from zero and one by hypothesis. 

Finally, I assume smoothness and spectral decay assumptions for the treatment and selection propensity scores as described in Section~\ref{sec:rkhs_main}.

\begin{assumption}[Smoothness of propensities]\label{assumption:smooth_propensity_planning}
Assume $\pi_0\in \mathcal{H}_{\mathcal{X}}^{c_6}$ and $\rho_0\in (\mathcal{H}_{\mathcal{D}}\otimes \mathcal{H}_{\mathcal{X}}\otimes \mathcal{H}_{\mathcal{M}})^{c_7}$.
\end{assumption}

To quantify the effective dimension of the propensity kernels, recall the convolution operator notation from Section~\ref{sec:rkhs_main}: $\lambda_j(k)$ is the $j$-th eigenvalue of the convolution operator $L:f\mapsto \int k(\cdot, w)f(w)\mathrm{d}\mathbb{P}(w)$ of the kernel $k$.

\begin{assumption}[Effective dimension of propensities]\label{assumption:spectral_propensity_planning}
Assume $\lambda_j(k_{\mathcal{X}})\asymp j^{-b_6}$ and $\lambda_j(k_{\mathcal{D}}\cdot k_{\mathcal{X}}\cdot k_{\mathcal{M}})\asymp j^{-b_7}$.
\end{assumption}

The former statement quantifies the effective dimension of $\pi_0 \in \mathcal{H}_{\mathcal{X}}$, and the latter quantifies the effective dimension of $\rho_0 \in \mathcal{H}_{\mathcal{D}}\otimes \mathcal{H}_{\mathcal{X}}\otimes \mathcal{H}_{\mathcal{M}}$. In Appendix~\ref{sec:inference_proof}, I argue that it is sufficient for \textit{either} the kernels of the propensities $(\pi_0,\rho_0)$ to have polynomial spectral decay \textit{or} for the kernels of the regression and sequential mean embedding $(\gamma_0,\omega_0)$ to have polynomial spectral decay---a kind of double spectral robustness.

\begin{theorem}[Semiparametic consistency, Gaussian approximation, and efficiency for dynamic sample selection]\label{theorem:inference_planning}
Suppose the conditions of Theorem~\ref{theorem:consistency_planning} hold, as well as Assumptions~\ref{assumption:resid_planning},~\ref{assumption:bounded_propensity_planning},~\ref{assumption:smooth_propensity_planning}, and~\ref{assumption:spectral_propensity_planning}. Set $\lambda_6=n^{-\frac{1}{2}}$ if $b_6=\infty$ and $\lambda_6=n^{-\frac{b_6}{b_6c_6+1}}$ otherwise. Set $\lambda_7=n^{-\frac{1}{2}}$ if $b_7=\infty$ and $\lambda_7=n^{-\frac{b_7}{b_7c_7+1}}$ otherwise.
Then for any $c,c_4,c_6,c_7\in(1,2]$ and $b_6,b_7\in(1,\infty]$ satisfying
$$
\frac{\min(b_6c_6,b_7c_7)}{\min(b_6c_6,b_7c_7)+1}+\frac{\min(c,c_4)-1}{\min(c,c_4)+1}>1,
$$
the estimator $\hat{\theta}^{ATE}(d)$ is consistent, i.e. $\hat{\theta}^{ATE}(d)\overset{p}{\rightarrow}\theta_0^{ATE}(d)$, and  the confidence interval includes $\theta_0^{ATE}(d)$ with probability approaching the nominal level, i.e. $$\lim_{n\rightarrow\infty} \mathbb{P}\left[\theta_0^{ATE}(d)\in \left\{\hat{\theta}^{ATE}(d)\pm c_a\frac{\hat{\sigma}(d)}{\sqrt{n}}\right\}\right]=1-a.$$
\end{theorem}
See Appendix~\ref{sec:inference_proof} for the proof. Observe the tradeoff among the smoothness and effective dimension assumptions across nonparametric objects. The spectral decay of propensity scores $(\pi_0,\rho_0)$ parametrized by $(b_6,b_7)$ must be sufficiently fast relative to the smoothness of the various nonparametric objects $(\gamma_0,\mu_{m},\pi_0,\rho_0)$ parametrized by $(c,c_4,c_6,c_7)$. To see that this parameter region is nonempty, observe that if $(c,c_4,c_6,c_7)=2$ then the condition simplifies to $\min(b_6,b_7)>1$. In other words, when the various nonparametric objects are smooth, it is sufficient for the treatment and selection propensities to belong to RKHSs with any polynomial rate of spectral decay. 
\section{Conclusion}\label{section:conclusion}

I propose a family of novel estimators for nonparametric estimation of dose response curves from selected samples as well as semiparametric estimation of treatment effects from selected samples. The estimators are easily implemented from closed form solutions in terms of kernel matrix operations. As a contribution to the sample selection literature, I propose simple estimators with closed form solutions for new causal estimands. As a contribution to the kernel methods literature, extend the approximation and estimation theory of kernel ridge regression to causal sample selection models. I pose as a question for future research how to extend this framework to the setting with missing-not-at-random (MNAR) data.

\newpage

\appendix

\begin{center}
    \uppercase{Online appendix}
\end{center}

\section{Counterfactual distributions}\label{section:dist}

\subsection{Learning problem}

In the main text, I study target parameters defined as \textit{means} or \textit{increments} of potential outcomes. In this section, I study \textit{distributions} of potential outcomes, extending the algorithms and analyses presented in the main text. I relax the assumption that $\mathcal{Y}\subset\mathbb{R}$ to the more general setting that $\mathcal{Y}$ is a Polish space. En route to estimating each counterfactual distribution, I estimate its kernel mean embedding in an outcome RKHS $\mathcal{H}_{\mathcal{Y}}$ with feature map $\phi(y)$.

\begin{definition}[Counterfactual distributions and embeddings] 
I define the following counterfactual distributions.
\begin{enumerate}
    \item $\theta_0^{D:ATE}(d):=\mathbb{P}[Y^{(d)}]$ is the counterfactual outcome distribution given intervention $D=d$ for the entire population.
     \item $ \theta_0^{D:DS}(d,\tilde{\mathbb{P}}):=\tilde{\mathbb{P}}[Y^{(d)}]$ is the counterfactual outcome distribution given intervention $D=d$ for an alternative population with data distribution $\tilde{\mathbb{P}}$.
    \item $ \theta_0^{D:ATT}(d,d'):=\mathbb{P}[Y^{(d')}|D=d]$ is the counterfactual outcome distribution given intervention $D=d'$ for the subpopulation who actually received treatment $D=d$.
     \item $\theta_0^{D:CATE}(d,v):=\mathbb{P}[Y^{(d)}|V=v]$ is the counterfactual outcome distribution given intervention $D=d$ for the subpopulation with subcovariate value $V=v$.
\end{enumerate}
Likewise I define embeddings of the counterfactual distributions, e.g. $
\check{\theta}_0^{D:ATE}(d):=\mathbb{E}[\phi(Y^{(d)})].
$
\end{definition}

\subsection{Identification}

The same identification results apply to distribution target parameters.

\begin{theorem}[Identification of counterfactual distributions and embeddings]\label{theorem:id_dist}
If Assumption~\ref{assumption:selection_static} holds then
\begin{enumerate}
    \item $[\theta_0^{D:ATE}(d)](y)=\int \mathbb{P}(y|1,d,x)\mathrm{d}\mathbb{P}(x)$.
    \item If in addition Assumption~\ref{assumption:covariate} holds, then $[\theta_0^{D:DS}(d,\tilde{\mathbb{P}})](y)=\int \mathbb{P}(y|1,d,x) \mathrm{d}\tilde{\mathbb{P}}(x)$.
    \item $[\theta_0^{D:ATT}(d,d')](y)=\int \mathbb{P}(y|1,d',x)\mathrm{d}\mathbb{P}(x|d)$.
    \item $[\theta_0^{D:CATE}(d,v)](y)=\int \mathbb{P}(y|1,d,v,x)\mathrm{d}\mathbb{P}(x|v)$.
\end{enumerate}
If Assumption~\ref{assumption:selection_planning} holds then
\begin{enumerate}
    \item $
[\theta_0^{D:ATE}(d)](y)=\int \mathbb{P}(y|1,d,x,m) \mathrm{d}\mathbb{P}(m|d,x) \mathrm{d}\mathbb{P}(x).
$
    \item If in addition Assumption~\ref{assumption:covariate_planning} holds then  $
[\theta_0^{D:DS}(d,\tilde{\mathbb{P}})](y)=\int \mathbb{P}(y|1,d,x,m) \mathrm{d}\tilde{\mathbb{P}}(m|d,x) \mathrm{d}\tilde{\mathbb{P}}(x)
$.
\end{enumerate}
Likewise for embedding of counterfactual distributions. For example, if in addition Assumption~\ref{assumption:RKHS} holds then $\check{\theta}_0^{ATE}(d)=\int \mathbb{E}\{\phi(Y)|S=1,D=d,X=x\}\mathrm{d}\mathbb{P}(x)
$.
\end{theorem}
To relate Theorem~\ref{theorem:id_dist} with Theorems~\ref{theorem:id_static} and~\ref{theorem:id_planning}, consider the generalized regressions
$
\gamma_0(1,d,x)=\mathbb{E}\{\phi(Y)|S=1,D=d,X=x\}
$ and $\gamma_0(1,d,x,m)=\mathbb{E}\{\phi(Y)|S=1,D=d,X=x,M=m\}$. Then I can express these results in the familiar form, e.g. $\check{\theta}_0^{D:ATE}(d)=\int \gamma_0(1,d,x)\mathrm{d}\mathbb{P}(x)$.

\subsection{Algorithm}

In what follows, I present a closed form solution for the embedding of a counterfactual distribution, then sample from the counterfactual distribution \cite{welling2009herding}. Alternatively, one could use the embedding to estimate moments \cite{kanagawa2014recovering}.

As previewed above, I introduce an additional scalar valued RKHS for outcome $Y$. Note that the regression being reweighted is now a conditional mean embedding, so I define the corresponding conditional expectation operators. For static sample selection, define the conditional expectation operator
$
E_8: \mathcal{H}_{\mathcal{Y}}\rightarrow \mathcal{H}_{\mathcal{S}}\otimes \mathcal{H}_{\mathcal{D}}\otimes  \mathcal{H}_{\mathcal{X}},\; f(\cdot)\mapsto \mathbb{E}\{f(SY)|S=\cdot, D=\cdot,X=\cdot \}
$. By construction
$
\gamma_0(1,d,x)=E_8^*\{\phi(1)\otimes \phi(d)\otimes \phi(x)\}
$.  For dynamic sample selection, define the conditional expectation operator
$
E_9:\mathcal{H}_{\mathcal{Y}}\rightarrow \mathcal{H}_{\mathcal{S}}\otimes \mathcal{H}_{\mathcal{D}}\otimes  \mathcal{H}_{\mathcal{X}}\otimes  \mathcal{H}_{\mathcal{M}},\; f(\cdot)\mapsto \mathbb{E}\{f(SY)|S=\cdot,D=\cdot,X=\cdot,M=\cdot \}
$.  I represent counterfactual distributions as evaluations of $(E_8^*,E_9^*)$. 

\begin{theorem}[Representation of counterfactual distribution embeddings]\label{theorem:representation_dist}
Suppose the conditions of Theorem~\ref{theorem:id_static} hold. Further suppose Assumption~\ref{assumption:RKHS} holds and $E_8\in\mathcal{L}_2(\mathcal{H}_{\mathcal{Y}},\mathcal{H}_{\mathcal{S}}\otimes\mathcal{H}_{\mathcal{D}}\otimes \mathcal{H}_{\mathcal{X}})$ (or $\mathcal{L}_2(\mathcal{H}_{\mathcal{Y}},\mathcal{H}_{\mathcal{S}}\otimes\mathcal{H}_{\mathcal{D}}\otimes\mathcal{H}_{\mathcal{V}} \otimes  \mathcal{H}_{\mathcal{X}})$ for $\theta_0^{D:CATE}$). Then
\begin{enumerate}
    \item $\check{\theta}_0^{D:ATE}(d)=E_8^*[\phi(1)\otimes \phi(d)\otimes \mu_x] $;
    \item $\check{\theta}_0^{D:DS}(d,\tilde{\mathbb{P}})=E_8^*[\phi(1)\otimes\phi(d)\otimes \nu_x]$;
    \item $\check{\theta}_0^{D:ATT}(d,d')=E_8^*[\phi(1)\otimes\phi(d')\otimes \mu_x(d)] $;
    \item $\check{\theta}_0^{D:CATE}(d,v)=E_8^*[\phi(1)\otimes\phi(d)\otimes \phi(v)\otimes \mu_{x}(v)]$.
\end{enumerate}
Suppose the conditions of Theorem~\ref{theorem:id_planning} hold. Further suppose Assumption~\ref{assumption:RKHS} holds and $E_9\in\mathcal{L}_2(\mathcal{H}_{\mathcal{Y}},\mathcal{H}_{\mathcal{S}}\otimes \mathcal{H}_{\mathcal{D}}\otimes \mathcal{H}_{\mathcal{X}}\otimes \mathcal{H}_{\mathcal{M}})$. Then
\begin{enumerate}
    \item $\check{\theta}_0^{D:ATE}(d)=E_9^*\left[\phi(1)\otimes\phi(d) \otimes \int \{\phi(x)\otimes \mu_{m}(d,x)\} \mathrm{d}\mathbb{P}(x) \right]$; 
    \item $\check{\theta}_0^{D:DS}(d,\tilde{\mathbb{P}})=E_9^*\left[\phi(1)\otimes\phi(d) \otimes \int \{\phi(x)\otimes \nu_{m}(d,x)\} \mathrm{d}\tilde{\mathbb{P}}(x) \right]$.
\end{enumerate}
\end{theorem}

See Appendix~\ref{section:alg_deriv} for the proof. The mean embeddings are the same as in the main text. These representations imply analogous estimators, e.g.
$\hat{\theta}^{D:ATE}(d)=\hat{E}_8^*\left[\phi(1)\otimes\phi(d)\otimes\hat{\mu}_x\right]$, where  $\hat{E}_9$ is a generalized kernel ridge regression and $\hat{\mu}_x$ is an empirical mean.

\begin{algorithm}[Nonparametric estimation of counterfactual distribution embeddings]\label{algorithm:dist}
Denote the empirical kernel matrices
$
K_{SS},$ $K_{DD},$ $K_{XX}$, $K_{MM}$, $K_{YY}\in\mathbb{R}^{n\times n}
$ calculated from observations drawn from population $\mathbb{P}$. Denote the empirical kernel matrices $K_{\tilde{D}\tilde{D}},$ $K_{\tilde{X}\tilde{X}}$, $K_{\tilde{M}\tilde{M}}\in\mathbb{R}^{\tilde{n}\times \tilde{n}}
$ calculated from observations drawn from population $\tilde{\mathbb{P}}$. Denote by $\odot$ the elementwise product. 
For static sample selection,
\begin{enumerate}
    \item $[\hat{\theta}^{D:ATE}(d)](y)=\frac{1}{n}\sum_{i=1}^n (S\odot K_{Yy})^{\top} (K_{SS}\odot K_{DD}\odot K_{XX}+n\lambda_8 I)^{-1}(K_{S1}\odot K_{Dd}\odot K_{Xx_i})  $;
     \item $[\hat{\theta}^{D:DS}(d,\tilde{\mathbb{P}})](y)=\frac{1}{\tilde{n}}\sum_{i=1}^{\tilde{n}} (S\odot K_{Yy})^{\top}(K_{SS}\odot K_{DD}\odot K_{XX}+n\lambda_8 I)^{-1}(K_{S1}\odot K_{Dd}\odot K_{X\tilde{x}_i}) $;
    \item $[\hat{\theta}^{D:ATT}(d,d')](y)=(S\odot K_{Yy})^{\top}(K_{SS}\odot K_{DD}\odot K_{XX}+n\lambda_8 I)^{-1}(K_{S1}\odot K_{Dd'}\odot [K_{XX}(K_{DD}+n\lambda_1 I)^{-1}K_{Dd}])$;
    \item $[\hat{\theta}^{D:CATE}(d,v)](y)=(S\odot K_{Yy})^{\top}(K_{SS}\odot K_{DD}\odot K_{VV}\odot K_{XX} +n\lambda_8 I)^{-1}(K_{S1}\odot K_{Dd}\odot K_{Vv}\odot K_{XX}(K_{VV}+n\lambda_2 I)^{-1}K_{Vv} ) $.
\end{enumerate}
For dynamic sample selection,
\begin{enumerate}
    \item $[\hat{\theta}^{D:ATE}(d)](y)=\frac{1}{n}\sum_{i=1}^n (S\odot K_{Yy})^{\top}(K_{SS}\odot K_{DD}\odot K_{XX}\odot K_{MM}+n\lambda_9 I)^{-1}$ \\
        $[K_{S1}\odot K_{Dd} \odot  K_{Xx_{i}}\odot 
        \{K_{MM}(K_{DD}\odot K_{XX}+n\lambda_4 I)^{-1}(K_{Dd}\odot K_{Xx_{i}})\}]  $,
    \item $[\hat{\theta}^{D:DS}(d,\tilde{\mathbb{P}})](y)=\frac{1}{n}\sum_{i=1}^{\tilde{n}} (S\odot K_{Yy})^{\top}(K_{SS}\odot K_{DD}\odot K_{XX}\odot K_{MM}+n\lambda_9 I)^{-1}$ \\
        $[K_{S1}\odot K_{Dd} \odot  K_{X\tilde{x}_{i}}\odot \{K_{M \tilde{M}}(K_{\tilde{D}\tilde{D}}\odot K_{\tilde{X}\tilde{X}}+\tilde{n}\lambda_5 I)^{-1}(K_{\tilde{D}d}\odot K_{\tilde{X}\tilde{x}_{i}})\}]  $,
\end{enumerate}
where $(\lambda_1,\lambda_4,\lambda_5,\lambda_8,\lambda_9)$ are ridge regression penalty hyperparameters.
\end{algorithm}
See Appendix~\ref{section:alg_deriv} for the derivation. See Theorem~\ref{theorem:consistency_dist}  for theoretical values of $(\lambda_1,\lambda_4,\lambda_5,\lambda_8,\lambda_9)$ that balance bias and variance. See Appendix~\ref{section:tuning} for a practical tuning procedure based on the closed form solution for leave-one-out cross validation to empirically balance bias and variance.

Finally, I present a kernel herding procedure  \cite{welling2009herding,muandet2020counterfactual} that uses the counterfactual distribution embedding to provide samples $(\tilde{Y_j})$ from the counterfactual distribution.

\begin{algorithm}[Nonparametric estimation of counterfactual distributions]\label{algorithm:herding}
Recall that $\hat{\theta}^{D:ATE}(d)$ is a mapping from $\mathcal{Y}$ to $\mathbb{R}$. 
Given $\hat{\theta}^{D:ATE}(d)$, calculate
\begin{enumerate}
    \item $\tilde{Y}_1=\argmax_{y\in\mathcal{Y}} \left[\{\hat{\theta}^{D:ATE}(d)\}(y)\right]$,
    \item $\tilde{Y}_{j}=\argmax_{y\in\mathcal{Y}} \left[\{ \hat{\theta}^{D:ATE}(d)\}(y)-\frac{1}{j+1}\sum_{\ell=1}^{j-1}k_{\mathcal{Y}}(\tilde{Y}_{\ell},y)\right]$ for $j>1$.
\end{enumerate}
Likewise for the other counterfactual distributions, replacing $\hat{\theta}^{D:ATE}(d)$ with the other quantities in Algorithm~\ref{algorithm:dist}.
\end{algorithm}
This procedure provides samples from counterfactual distributions, with which one may estimate the counterfactual densities or test statistical hypotheses.

\subsection{Guarantees}

As in Sections~\ref{section:static} and~\ref{section:dynamic}, I place regularity conditions on the original spaces (now allowing $\mathcal{Y}$ to be Polish), assume the generalized regression $\gamma_0$ is smooth, and assume the various conditional mean embeddings are smooth in order to prove uniform consistency.

\begin{theorem}[Nonparametric consistency of counterfactual distribution embeddings]\label{theorem:consistency_dist}
For static sample selection, suppose the conditions of Theorem~\ref{theorem:consistency_static} hold as well as Assumption~\ref{assumption:smooth_op} with $\mathcal{A}_8=\mathcal{Y}$ and $\mathcal{B}_8=\mathcal{S}\times\mathcal{D}\times \mathcal{X}$ (or $\mathcal{B}_8=\mathcal{S}\times \mathcal{D}\times \mathcal{V}\times \mathcal{X}$ for $\theta_0^{D:CATE}$). Set $(\lambda_1,\lambda_2,\lambda_8)=(n^{-\frac{1}{c_1+1}},n^{-\frac{1}{c_2+1}},n^{-\frac{1}{c_8+1}})$.
\begin{enumerate}
    \item Then
    $$
    \sup_{d\in\mathcal{D}}\|\hat{\theta}^{D:ATE}(d)-\check{\theta}_0^{D:ATE}(d)\|_{\mathcal{H}_{\mathcal{Y}}}=O_p\left(n^{-\frac{1}{2}\frac{c_8-1}{c_8+1}}\right).
$$
 \item If in addition Assumption~\ref{assumption:covariate} holds, then 
      $$
     \sup_{d\in\mathcal{D}}\|\hat{\theta}^{D:DS}(d,\tilde{\mathbb{P}})-\check{\theta}_0^{D:DS}(d,\tilde{\mathbb{P}})\|_{\mathcal{H}_{\mathcal{Y}}}=O_p\left( n^{-\frac{1}{2}\frac{c_8-1}{c_8+1}}+\tilde{n}^{-\frac{1}{2}}\right).
    $$
    \item If in addition Assumption~\ref{assumption:smooth_op} holds with $\mathcal{A}_1=\mathcal{X}$ and $\mathcal{B}_1=\mathcal{D}$, then
        $$
    \sup_{d,d'\in\mathcal{D}}\|\hat{\theta}^{D:ATT}(d,d')-\check{\theta}_0^{D:ATT}(d,d')\|_{\mathcal{H}_{\mathcal{Y}}}=O_p\left(n^{-\frac{1}{2}\frac{c_8-1}{c_8+1}}+n^{-\frac{1}{2}\frac{c_1-1}{c_1+1}}\right).
$$
 \item If in addition Assumption~\ref{assumption:smooth_op} holds with $\mathcal{A}_2=\mathcal{X}$ and $\mathcal{B}_2=\mathcal{V}$, then
      $$
     \sup_{d\in\mathcal{D},v\in\mathcal{V}}\|\hat{\theta}^{D:CATE}(d,v)-\check{\theta}_0^{D:CATE}(d,v)\|_{\mathcal{H}_{\mathcal{Y}}}=O_p\left(n^{-\frac{1}{2}\frac{c_8-1}{c_8+1}}+n^{-\frac{1}{2}\frac{c_2-1}{c_2+1}}\right).
    $$
\end{enumerate}
For dynamic sample selection, suppose the conditions of Theorem~\ref{theorem:consistency_planning} hold as well as Assumption~\ref{assumption:smooth_op} with $\mathcal{A}_9=\mathcal{Y}$ and $\mathcal{B}_9=\mathcal{S}\times \mathcal{D}\times \mathcal{X}\times \mathcal{M}$.  Set $(\lambda_4,\lambda_5,\lambda_9)$=$(n^{-\frac{1}{c_4+1}},\tilde{n}^{-\frac{1}{c_5+1}},n^{-\frac{1}{c_9+1}})$.
\begin{enumerate}
    \item Then
    $$
    \sup_{d \in\mathcal{D}}\|\hat{\theta}^{D:ATE}(d)-\check{\theta}_0^{D:ATE}(d)\|_{\mathcal{H}_{\mathcal{Y}}}=O_p\left(n^{-\frac{1}{2}\frac{c_9-1}{c_9+1}}+n^{-\frac{1}{2}\frac{c_4-1}{c_4+1}}\right).
$$
    \item If in addition Assumption~\ref{assumption:covariate_planning} holds, then 
     $$
     \sup_{d \in\mathcal{D}}\|\hat{\theta}^{D:DS}(d,\tilde{\mathbb{P}})-\check{\theta}_0^{D:DS}(d,\tilde{\mathbb{P}})\|_{\mathcal{H}_{\mathcal{Y}}}=O_p\left(n^{-\frac{1}{2}\frac{c_9-1}{c_9+1}}+\tilde{n}^{-\frac{1}{2}\frac{c_5-1}{c_5+1}}\right).
    $$
\end{enumerate}
\end{theorem}
See Appendix~\ref{section:consistency_proof} for the proof and exact finite sample rates. These rates are at best $n^{-\frac{1}{6}}$ when $(c_1,c_2,c_4,c_5,c_8,c_9)=2$. 

Finally, I prove that the samples $(\tilde{Y}_j)$ calculated from the counterfactual distribution embeddings weakly converge to the desired counterfactual distribution.
\begin{assumption}[Additional regularity]\label{assumption:regularity}
Assume
\begin{enumerate}
    \item $\mathcal{Y}$ is locally compact,
    \item $\mathcal{H}_{\mathcal{Y}}\subset\mathcal{C}$, where $\mathcal{C}$ is the space of bounded, continuous, real valued functions that vanish at infinity.
\end{enumerate}
\end{assumption}
See \cite{simon2020metrizing} for a discussion of the assumptions that $\mathcal{Y}$ is Polish and locally compact. $\mathcal{Y}=\mathbb{R}^p$ satisfies both conditions.
\begin{theorem}[Convergence in distribution of counterfactual distributions]\label{theorem:conv_dist}
Suppose the conditions of Theorem~\ref{theorem:consistency_dist} hold, as well as Assumption~\ref{assumption:regularity}. Suppose samples $(\tilde{Y}_j)$ are calculated for $\theta_0^{D:ATE}(d)$ as described in Algorithm~\ref{algorithm:herding}. Then $(\tilde{Y}_j)\overset{d}{\rightarrow} \theta_0^{D:ATE}(d)$. Likewise for the other counterfactual distributions, replacing $\hat{\theta}^{D:ATE}(d)$ with the other quantities in Algorithm~\ref{algorithm:dist}.
\end{theorem}
See Appendix~\ref{section:consistency_proof} for the proof.
\section{Identification}\label{sec:id}

\subsection{Static sample selection}

\begin{proof}[Proof of Theorem~\ref{theorem:id_static}]
By no interference and conditional exchangeability,
\begin{align*}
    \gamma_0(1,d,x)&=\mathbb{E}[Y|S=1,D=d,X=x] \\
    &= \mathbb{E}[Y^{(d)}|S=1,D=d,X=x] \\
    &=\mathbb{E}[Y^{(d)}|D=d,X=x]  \\
    &=\mathbb{E}[Y^{(d)}|X=x].
\end{align*}
Therefore by law of iterated expectations and Assumption~\ref{assumption:covariate}
\begin{align*}
    \theta_0^{ATE}(d)&=\mathbb{E}[Y^{(d)}]=\mathbb{E}[\mathbb{E}[Y^{(d)}|X]]=\int \gamma_0(1,d,x)\mathrm{d}\mathbb{P}(x); \\
    \theta_0^{DS}(d)&=\mathbb{E}_{\tilde{\mathbb{P}}}[Y^{(d)}]=\mathbb{E}_{\tilde{\mathbb{P}}}[\mathbb{E}_{\tilde{\mathbb{P}}}[Y^{(d)}|X]]=\int \gamma_0(1,d,x)\mathrm{d}\tilde{\mathbb{P}}(x).
\end{align*}
Similarly by law of iterated expectations and conditional exchangeability,
\begin{align*}
    \theta_0^{ATT}(d,d;)&=\mathbb{E}[Y^{(d')}|D=d]\\
    &=\mathbb{E}[\mathbb{E}[Y^{(d')}|D,X]|D=d]\\
    &=\mathbb{E}[\mathbb{E}[Y^{(d')}|X]|D=d]\\
    &=\int \gamma_0(1,d',x)\mathrm{d}\mathbb{P}(x|d); \\
    \theta_0^{CATE}(d,v)&=\mathbb{E}[Y^{(d)}|V=v]\\
    &=\mathbb{E}[\mathbb{E}[Y^{(d)}|D,V,X]|V=v]\\
    &=\mathbb{E}[\mathbb{E}[Y^{(d)}|V,X]|V=v]\\
    &=\int \gamma_0(1,d,v,x)\mathrm{d}\mathbb{P}(x|v).
\end{align*}

Next I turn to incremental effects using the nonseparable model notation $Y=Y(D,\eta)$, where $\eta$ is unobserved heterogeneity. I restructure the argument so that taking the derivative with respect to treatment does not lead to additional factors (as it otherwise would by chain rule). 

By no interference and conditional exchangeability,
\begin{align*}
\gamma_0(1,d,x)&=\mathbb{E}[Y|S=1,D=d,X=x] \\
    &= \int Y(d,\eta) \mathrm{d}\mathbb{P}(\eta|1,d,x) \\
    &=\int Y(d,\eta) \mathrm{d}\mathbb{P}(\eta|d,x) \\
    &=\int Y(d,\eta) \mathrm{d}\mathbb{P}(\eta|x).
\end{align*}
Taking the derivative yields
$$
\nabla_d\gamma_0(1,d,x)=\int \nabla_d Y(d,\eta) \mathrm{d}\mathbb{P}(\eta|x).
$$
Therefore by law of iterated expectation and Assumption~\ref{assumption:covariate}
\begin{align*}
    \theta_0^{\nabla:ATE}(d)&
    =\int \nabla_d Y(d,\eta) \mathrm{d}\mathbb{P}(\eta) \\
    &=\int \nabla_d Y(d,\eta) \mathrm{d}\mathbb{P}(\eta,x) \\
    &=\int \nabla_d Y(d,\eta) \mathrm{d}\mathbb{P}(\eta|x)\mathrm{d}\mathbb{P}(x)  \\
    &=\int \nabla_d\gamma_0(1,d,x)\mathrm{d}\mathbb{P}(x);\\
    \theta_0^{\nabla:DS}(d)&
    =\int \nabla_d Y(d,\eta) \mathrm{d}\tilde{\mathbb{P}}(\eta)\\
    &=\int \nabla_d Y(d,\eta) \mathrm{d}\tilde{\mathbb{P}}(\eta,x)\\
    &=\int \nabla_d Y(d,\eta) \mathrm{d}\tilde{\mathbb{P}}(\eta|x)\mathrm{d}\tilde{\mathbb{P}}(x) \\
    &=\int \nabla_d\gamma_0(1,d,x)\mathrm{d}\tilde{\mathbb{P}}(x).
\end{align*}
Similarly by law of iterated expectation and conditional exchangeability,
\begin{align*}
    \theta_0^{\nabla:ATT}(d,d')&
    =\int \nabla_{d'} Y(d,\eta) \mathrm{d}\mathbb{P}(\eta|d)\\
    &=\int \nabla_d Y(d,\eta) \mathrm{d}\mathbb{P}(\eta,x|d)\\
    &=\int \nabla_d Y(d,\eta) \mathrm{d}\mathbb{P}(\eta|d,x)\mathbb{P}(x|d) \\
    &=\int \nabla_d Y(d,\eta) \mathrm{d}\mathbb{P}(\eta|x)\mathrm{d}\mathbb{P}(x|d) \\
    &=\int \nabla_d\gamma_0(1,d',x)\mathrm{d}\mathbb{P}(x|d);
    \end{align*}
    \begin{align*}
    \theta_0^{\nabla:CATE}(d,v)&
    =\int \nabla_d Y(d,\eta) \mathrm{d}\mathbb{P}(\eta|v)\\
    &=\int \nabla_d Y(d,\eta) \mathrm{d}\mathbb{P}(\eta,x|v)\\
    &=\int \nabla_d Y(d,\eta) \mathrm{d}\mathbb{P}(\eta|v,x)\mathrm{d}\mathbb{P}(x|v) \\
    &=\int \nabla_d\gamma_0(1,d,v,x)\mathrm{d}\mathbb{P}(x|v).
\end{align*}
\end{proof}

\begin{proof}[Proof of Theorem~\ref{theorem:dr_static}]
For $\theta_0^{ATE}$, observe that by law of iterated expectations
\begin{align*}
    \gamma(1,d,X)
&=\mathbb{E}[S\cdot \mathbbm{1}_{D=d}\cdot \gamma(S,D,X)|S=1,D=d,X] \\
&=\mathbb{E}\left[\frac{S\cdot \mathbbm{1}_{D=d}\cdot\gamma(S,D,X)}{\mathbb{P}(S=1,D=d|X)}\mid S=1,D=d,X\right]\mathbb{P}(S=1,D=d|X) \\
&=\mathbb{E}\left[\frac{S\cdot \mathbbm{1}_{D=d}\cdot \gamma(S,D,X)}{\mathbb{P}(S=1,D=d|X)}\mid X\right] \\
&=\mathbb{E}\left[\frac{S\cdot \mathbbm{1}_{D=d}\cdot \gamma(S,D,X)}{\mathbb{P}(S=1|d,X)\mathbb{P}(d|X)}\mid X\right]
\end{align*}
Hence by law of iterated expectation and Assumption~\ref{assumption:bounded_Riesz},
$$
\mathbb{E}[\gamma(1,d,X)] =\mathbb{E}\left[\mathbb{E}\left[\frac{S\cdot \mathbbm{1}_{D=d}\cdot \gamma(S,D,X)}{\mathbb{P}(S=1|d,X)\mathbb{P}(d|X)}\mid X\right]\right]=\mathbb{E}\left[\frac{S\cdot \mathbbm{1}_{D=d}\cdot \gamma(S,D,X)}{\mathbb{P}(S=1|d,X)\mathbb{P}(d|X)}\right].
$$
Likewise for $\theta_0^{DS}$. 

Next, consider $\theta_0^{ATT}$. By Theorem~\ref{theorem:id_static} and law of iterated expectations 
\begin{align*}
    \theta_0^{ATT}(d,d')
    &=\int \gamma_0(1,d',x) \mathrm{d}\mathbb{P}(x|d)\\
    &=\frac{1}{\mathbb{P}(d)}\int \gamma_0(1,d',x) \mathbb{P}(d|x)\mathrm{d}\mathbb{P}(x)\\
    &=\frac{1}{\mathbb{P}(d)}\int \gamma_0(1,d',x) \mathbbm{1}_{D=d}\mathrm{d}\mathbb{P}(x).
\end{align*}
Moreover, by law of iterated expectations,
\begin{align*}
    \int \gamma(1,d',x) \mathbbm{1}_{D=d}\mathrm{d}\mathbb{P}(x) 
    &= \mathbb{E}\left[\mathbb{E}\left[\frac{S\cdot \mathbbm{1}_{D=d'} \cdot \gamma(S,D,X)}{\mathbb{P}(S=1|d',X)\mathbb{P}(d'|X)} \mid X\right]\mathbbm{1}_{D=d}\right] \\
    &=\mathbb{E}\left[\mathbb{E}\left[\frac{S\cdot \mathbbm{1}_{D=d'} \cdot \gamma(S,D,X)}{\mathbb{P}(S=1|d',X)\mathbb{P}(d'|X)} \mid X\right]\mathbb{P}(d|X)\right]\\
    &=\mathbb{E}\left[\frac{S\cdot \mathbbm{1}_{D=d'} \cdot \gamma(S,D,X)}{\mathbb{P}(S=1|d',X)\mathbb{P}(d'|X)} \mathbb{P}(d|X)\right]
\end{align*}

Finally, consider $\theta_0^{CATE}$. By Theorem~\ref{theorem:id_static} and law of iterated expectations 
\begin{align*}
    \theta_0^{CATE}(d,v)
    &=\int \gamma_0(1,d,v,x) \mathrm{d}\mathbb{P}(x|v)\\
    &=\frac{1}{\mathbb{P}(v)}\int \gamma_0(1,d,v,x) \mathbb{P}(v|x)\mathrm{d}\mathbb{P}(x)\\
    &=\frac{1}{\mathbb{P}(v)}\int \gamma_0(1,d,v,x) \mathbbm{1}_{V=v}\mathrm{d}\mathbb{P}(x).
\end{align*}
Moreover, by law of iterated expectations,
\begin{align*}
    \int \gamma(1,d,v,x) \mathbbm{1}_{V=v}\mathrm{d}\mathbb{P}(x) 
     &= \mathbb{E}\left[\gamma(1,d,V,X) \mathbbm{1}_{V=v}\right]\\
    &= \mathbb{E}\left[\mathbb{E}\left[\frac{S\cdot \mathbbm{1}_{D=d} \cdot \gamma(S,D,V,X)}{\mathbb{P}(S=1|d,V,X)\mathbb{P}(d|V,X)} \mid V,X\right]\mathbbm{1}_{V=v}\right] \\
    &=\mathbb{E}\left[\frac{S\cdot \mathbbm{1}_{D=d} \cdot \gamma(S,D,V,X)}{\mathbb{P}(S=1|d,V,X)\mathbb{P}(d|V,X)} \mathbbm{1}_{V=v}\right].
\end{align*}

\end{proof}

\subsection{Dynamic sample selection}

\begin{proof}[Proof of Theorem~\ref{theorem:id_planning}]
By no interference and conditional exchangeability,
\begin{align*}
    \gamma_0(1,d,x,m)&=\mathbb{E}[Y|S=1,D=d,X=x,M=m] \\
    &= \mathbb{E}[Y^{(d)}|S=1,D=d,X=x,M=m] \\
    &=\mathbb{E}[Y^{(d)}|D=d,X=x,M=m].
\end{align*}
By conditional exchangeability and law of iterated expectation
\begin{align*}
    \mathbb{E}[Y^{(d)}|X=x]
    = \mathbb{E}[Y^{(d)}|D=d,X=x] 
    =\mathbb{E}[\mathbb{E}[Y^{(d)}|D,X,M]|D=d,X=x].
\end{align*}
In summary
\begin{align*}
\mathbb{E}[Y^{(d)}|X=x]
&=\mathbb{E}[\gamma_0(1,D,X,M)|D=d,X=x]\\
&=\mathbb{E}[\gamma_0(1,d,x,M)|D=d,X=x]\\
&=\int \gamma_0(1,d,x,m) \mathrm{d}\mathbb{P}(m|d,x).
\end{align*}
Therefore by law of iterated expectations and Assumption~\ref{assumption:covariate}
\begin{align*}
    \theta_0^{ATE}(d)&=\mathbb{E}[Y^{(d)}]=\mathbb{E}[\mathbb{E}[Y^{(d)}|X]]=\int \gamma_0(1,d,x,m)\mathrm{d}\mathbb{P}(m|d,x)\mathrm{d}\mathbb{P}(x); \\
    \theta_0^{DS}(d)&=\mathbb{E}_{\tilde{\mathbb{P}}}[Y^{(d)}]=\mathbb{E}_{\tilde{\mathbb{P}}}[\mathbb{E}_{\tilde{\mathbb{P}}}[Y^{(d)}|X]]=\int \gamma_0(1,d,x,m)\mathrm{d}\tilde{\mathbb{P}}(m|d,x)\mathrm{d}\tilde{\mathbb{P}}(x).
\end{align*}
\end{proof}

\begin{proof}[Proof of Theorem~\ref{theorem:id_dist}]
Identical to the proofs of Theorems~\ref{theorem:id_static} and~\ref{theorem:id_planning}.
\end{proof}

\section{Algorithm derivation}\label{section:alg_deriv}

\subsection{Static sample selection}

\begin{proof}[Proof of Theorem~\ref{theorem:representation_static}]
I generalize \cite[Theorem 6.3]{singh2020kernel} for the static sample selection problem. Recall from Remark~\ref{remark:notation} that
$
\gamma_0(s,d,x)=\mathbb{E}[SY|S=s,D=d,X=x],
$ which by hypothesis is an element of a tensor product RKHS.

In Assumption~\ref{assumption:RKHS}, I impose that the scalar kernels are bounded. This assumption has several implications. First, the feature maps are Bochner integrable \cite[Definition A.5.20]{steinwart2008support}. Bochner integrability permits the exchange of expectation and inner product. Second, the mean embeddings exist. Third, the product kernel is also bounded and hence the tensor product RKHS inherits these favorable properties. By Theorem~\ref{theorem:id_static} and linearity of expectation
 \begin{align*}
    \theta_0^{ATE}(d)&= \int \gamma_0(1,d,x)\mathrm{d}\mathbb{P}(x)\\
    &=\int \langle \gamma_0, \phi(1)\otimes \phi(d)\otimes \phi(x)\rangle_{\mathcal{H}}  \mathrm{d}\mathbb{P}(x) \\
    &= \langle \gamma_0, \phi(1)\otimes\phi(d)\otimes \int\phi(x) \mathrm{d}\mathbb{P}(x) \rangle_{\mathcal{H}} \\
    &= \langle \gamma_0, \phi(1)\otimes\phi(d)\otimes \mu_x \rangle_{\mathcal{H}}.
\end{align*}
Likewise for $\theta_0^{DS}(d,\tilde{\mathbb{P}})$. Next,
\begin{align*}
    \theta_0^{ATT}(d,d')&= \int \gamma_0(1,d',x)\mathrm{d}\mathbb{P}(x|d)\\
    &=\int \langle \gamma_0, \phi(1)\otimes \phi(d')\otimes \phi(x)\rangle_{\mathcal{H}}  \mathrm{d}\mathbb{P}(x|d) \\
    &= \langle \gamma_0, \phi(1)\otimes \phi(d')\otimes \int\phi(x) \mathrm{d}\mathbb{P}(x|d) \rangle_{\mathcal{H}}\\
    &= \langle \gamma_0, \phi(1)\otimes \phi(d')\otimes \mu_x(d) \rangle_{\mathcal{H}}.
\end{align*}
Finally, 
 \begin{align*}
    \theta_0^{CATE}(d,v)&= \int \gamma_0(1,d,v,x)\mathrm{d}\mathbb{P}(x|v)\\
    &=\int \langle \gamma_0, \phi(1)\otimes \phi(d)\otimes \phi(v) \otimes \phi(x)\rangle_{\mathcal{H}}  \mathrm{d}\mathbb{P}(x|v) \\
    &= \langle \gamma_0, \phi(1)\otimes \phi(d)\otimes \phi(v) \otimes \int\phi(x) \mathrm{d}\mathbb{P}(x|v) \rangle_{\mathcal{H}} \\
    &= \langle \gamma_0, \phi(1)\otimes  \phi(d)\otimes \phi(v) \otimes \mu_{x}(v) \rangle_{\mathcal{H}}.
\end{align*}
\cite[Lemma 4.34]{steinwart2008support} guarantees that the derivative feature map $\nabla_d\phi(d)$ exists, is continuous, and is Bochner integrable (since $\kappa_d'=\sqrt{\sup_{d,d'\in\mathcal{D}}\nabla_d\nabla_{d'}k(d,d')}<\infty$). Therefore the derivations remain valid for incremental effects.
\end{proof}

\begin{proof}[Derivation of Algorithm~\ref{algorithm:static}]
I generalize \cite[Algorithm 6.4]{singh2020kernel} for the static sample selection problem.

Recall from Remark~\ref{remark:notation} that
$$
\gamma_0(1,d,x)=\mathbb{E}[Y|S=1,D=d,X=x]=\mathbb{E}[SY|S=1,D=d,X=x].
$$
The final expression is in terms of observables $(SY,S,D,X)$ as desired. Extending standard arguments \cite{kimeldorf1971some} according to this representation,
\begin{align*}
\hat{\gamma}(1,d,x)&=
\langle \hat{\gamma}, \phi(1)\otimes \phi(d)\otimes \phi(x) \rangle_{\mathcal{H}}\\
&=(S \odot Y)^{\top}(K_{SS} \odot K_{DD}\odot K_{XX}+n\lambda)^{-1}(K_{S1}\odot K_{Dd}\odot K_{Xx}).    
\end{align*}
    The results for $\hat{\theta}^{ATE}(d)$ holds by substitution:
   $$
   \hat{\mu}_x=\frac{1}{n}\sum_{i=1}^n \phi(x_i),\quad \hat{\theta}^{ATE}(d)=\langle \hat{\gamma}, \phi(1)\otimes  \phi(d)\otimes \hat{\mu}_x\rangle_{\mathcal{H}}.
   $$
    Likewise for $\hat{\theta}^{DS}(d,\tilde{\mathbb{P}})$. 
    
    The results for $\hat{\theta}^{ATT}(d,d')$ and $\hat{\theta}^{CATE}(d,v)$ use the closed form of the conditional mean embedding from \cite[Algorithm 1]{singh2019kernel}. Specifically,
    $$
    \hat{\mu}_x(d)=K_{\cdot X}(K_{DD}+n\lambda_1)^{-1}K_{Dd},\quad  \hat{\theta}^{ATT}(d,d')=\langle \hat{\gamma},  \phi(1)\otimes\phi(d')\otimes \hat{\mu}_x(d)\rangle_{\mathcal{H}} 
    $$
and
    $$
    \hat{\mu}_{x}(v)=K_{\cdot X}(K_{VV}+n\lambda_2)^{-1}K_{Vv},\quad \hat{\theta}^{CATE}(d,v)=\langle \hat{\gamma},  \phi(1)\otimes\phi(d)\otimes \phi(v)\otimes \hat{\mu}_{x}(v)\rangle_{\mathcal{H}}.
    $$
For incremental effects, replace $\hat{\gamma}(1,d,x)$ with
\begin{align*}
\nabla_d\hat{\gamma}(1, d,x)&=
\langle \hat{\gamma},  \phi(1)\otimes \nabla_d \phi(d)\otimes \phi(x) \rangle_{\mathcal{H}}\\
&=(S\odot Y)^{\top}(K_{DD}\odot K_{XX}+n\lambda)^{-1}(K_{S1}\odot \nabla_d K_{Dd}\odot K_{Xx}).
\end{align*}
\end{proof}

\begin{algorithm}[Explicit computations for Algorithm~\ref{algorithm:static_dml}]
Abstracting from sample splitting,
\begin{enumerate}
\item $\theta_0^{ATE},\theta_0^{DS},\theta_0^{ATT}$: $\hat{\gamma}(1,d,x)=(S\odot Y)^{\top}(K_{SS}\odot K_{DD}\odot  K_{XX}+n\lambda I)^{-1}(K_{S1}\odot K_{Dd}\odot  K_{Xx}) $,
  \item $\theta_0^{CATE}$: $\hat{\gamma}(1,d,v,x)=(S\odot Y)^{\top}(K_{SS}\odot K_{DD}\odot  K_{VV}\odot K_{XX}+n\lambda I)^{-1}(K_{S1}\odot K_{Dd}\odot  K_{Vv}\odot K_{Xx}) $,
\end{enumerate}
\end{algorithm}

\subsection{Dynamic sample selection}

\begin{proof}[Proof of Theorem~\ref{theorem:representation_planning}]
I adapt the techniques in \cite[Theorem 6]{singh2021workshop} to the dynamic sample selection problem. Recall from Remark~\ref{remark:notation} that
$
\gamma_0(s,d,x,m)=\mathbb{E}[SY|S=s,D=d,X=x,M=m],
$ which by hypothesis is an element of a tensor product RKHS.

Assumption~\ref{assumption:RKHS} implies Bochner integrability, which permits the exchange of expectation and inner product. Therefore
\begin{align*}
    \theta_0^{ATE}(d)
    &= \int \gamma_0(1,d,x,m) \mathrm{d}\mathbb{P}(m|d,x)\mathrm{d}\mathbb{P}(x)\\
    &=\int \langle \gamma_0, \phi(1) \otimes  \phi(d)\otimes\phi(x)\otimes \phi(m)\rangle_{\mathcal{H}}  \mathrm{d}\mathbb{P}(m|d,x)\mathrm{d}\mathbb{P}(x) \\
    &= \int \langle \gamma_0,  \phi(1) \otimes  \phi(d)\otimes\phi(x)\otimes \int \phi(m) \mathrm{d}\mathbb{P}(m|d,x)\rangle_{\mathcal{H}}  \mathrm{d}\mathbb{P}(x)\\
    &=\int\langle \gamma_0,  \phi(1) \otimes \phi(d)\otimes \phi(x)\otimes \mu_{m}(d,x) \rangle_{\mathcal{H}}\mathrm{d}\mathbb{P}(x)\\
    &=\langle \gamma_0, \phi(1) \otimes \phi(d)\otimes \int \{\phi(x)\otimes \mu_{m}(d,x)\} \mathrm{d}\mathbb{P}(x) \rangle_{\mathcal{H}}.
\end{align*}
The argument for $\theta_0^{DS}$ is identical.
\end{proof}

\begin{proof}[Derivation of Algorithm~\ref{algorithm:planning}]
I adapt the techniques of \cite[Algorithm 3]{singh2021workshop} to the dynamic sample selection problem.

Recall from Remark~\ref{remark:notation} that
$$
\gamma_0(1,d,x,m)=\mathbb{E}[Y|S=1,D=d,X=x,M=m]=\mathbb{E}[SY|S=1,D=d,X=x,M=m].
$$
The final expression is in terms of observables $(SY,S,D,X,M)$ as desired. Extending standard arguments \cite{kimeldorf1971some} according to this representation,
\begin{align*}
\hat{\gamma}(1,d,x,m)
&=
\langle \hat{\gamma}, \phi(1)\otimes \phi(d)\otimes  \phi(x)\otimes \phi(m) \rangle_{\mathcal{H}}\\
&=(S\odot Y)^{\top}( K_{SS}\odot K_{DD}\odot K_{XX}\odot K_{MM}+n\lambda I)^{-1} (K_{S1}\odot K_{Dd}\odot   K_{Xx}\odot K_{Mm}).    \end{align*}

    
    By \cite[Algorithm 1]{singh2019kernel}, write the conditional mean 
$$
    \hat{\mu}_{m}(d,x)=K_{\cdot M}(K_{DD}\odot K_{XX}+n\lambda_4 I)^{-1}(K_{Dd}\odot K_{Xx}).
    $$
    Therefore
    \begin{align*}
        \frac{1}{n}\sum_{i=1}^n[\phi(X_i) \otimes \hat{\mu}_{m}(d,X_i) ]
        &=\frac{1}{n}\sum_{i=1}^n[\phi(X_i)\otimes\{K_{\cdot M}(K_{DD}\odot K_{XX}+n\lambda_4 I)^{-1}(K_{Dd}\odot K_{Xx_i})\} ]
    \end{align*}
    and
    \begin{align*}
        \hat{\theta}^{ATE}(d)
        &=\langle \hat{\gamma},\phi(1)\otimes \phi(d)\otimes \frac{1}{n}\sum_{i=1}^n[\phi(X_i)\otimes \hat{\mu}_{m}(d,X_i)]  \rangle_{\mathcal{H}} \\
        &=\frac{1}{n}\sum_{i=1}^n (S\odot Y)^{\top}(K_{SS}\odot K_{DD}\odot  K_{XX}\odot K_{MM}+n\lambda I)^{-1} \\
        &\quad [K_{S1}\odot K_{Dd}\odot  K_{Xx_i}\odot \{K_{M M}(K_{DD}\odot K_{XX}+n\lambda_4 I)^{-1} (K_{Dd}\odot K_{Xx_i})\}].
    \end{align*}
     
     The argument for $\theta_0^{DS}$ is identical.
\end{proof}

\begin{algorithm}[Explicit computations for Algorithm~\ref{algorithm:planning_dml}]
Abstracting from sample splitting,
\begin{enumerate}
\item $\hat{\gamma}(1,d,x,m)=(S\odot Y)^{\top}(K_{SS}\odot K_{DD}\odot  K_{XX}\odot K_{MM}+n\lambda I)^{-1}(K_{S1}\odot K_{Dd}\odot  K_{Xx}\odot K_{Mm}) $,
  \item $\hat{\pi}(x)=D^{\top}(K_{XX}+n\lambda_6 I)^{-1}K_{Xx}$,
    \item $\hat{\rho}(d,x,m)=S^{\top}(K_{DD}\odot K_{XX}\odot K_{MM}+n\lambda_7 I)^{-1}(K_{Dd}\odot K_{Xx}\odot K_{Mm}) $.
\end{enumerate}
\end{algorithm}

\subsection{Counterfactual distributions}

\begin{proof}[Proof of Theorem~\ref{theorem:representation_dist}]
The argument is analogous to Theorems~\ref{theorem:representation_static} and~\ref{theorem:representation_planning}, recognizing that in the distributional setting
\begin{align*}
    \gamma_0(1,d,x)&= E_8^*\{\phi(1) \otimes \phi(d) \otimes \phi(x)\}, \\
     \gamma_0(1,d,x,m)&= E_9^*\{\phi(1) \otimes \phi(d)\otimes  \phi(x)\otimes \phi(m)\}.
\end{align*}
\end{proof}

\begin{proof}[Derivation of Algorithm~\ref{algorithm:dist}]
The argument is analogous to Algorithms~\ref{algorithm:static} and~\ref{algorithm:planning} appealing to \cite[Algorithm 1]{singh2019kernel}, which gives
\begin{align*}
    \hat{\gamma}(1,d,x)
    &=\hat{E}_8^*[\phi(1)\otimes \phi(d)\otimes   \phi(x)] \\
    &=(K_{1 S}\odot K_{\cdot Y})(K_{SS}\odot K_{DD}\odot  K_{XX}+n\lambda_8 I)^{-1}( K_{S1}\odot K_{Dd}\odot K_{Xx}),\\
    \hat{\gamma}(1,d,x,m)&=\hat{E}_9^*\{\phi(1) \otimes \phi(d)\otimes  \phi(x) \otimes \phi(m)\}\\
       &=(K_{1 S}\odot K_{\cdot Y})(K_{SS} \odot K_{DD}\odot   K_{XX}\odot K_{MM} +n\lambda_9 I)^{-1}( K_{S1}\odot K_{Dd}\odot K_{XX}\odot K_{Mm}).
\end{align*}
\end{proof}

\subsection{Alternative closed form}

I present alternative expressions for $\hat{\theta}^{ATE}$ and $\hat{\theta}^{DS}$ in Algorithm~\ref{algorithm:planning} that run faster by replacing the summations with matrix multiplications.
 
\begin{lemma}[Proposition 2 of \cite{singh2021workshop}] \label{lem:matrix_form}
If $A \in \mathbb{R}^{n\times n}$, $a,b \in \mathbb{R}^n$, and $\mathbf{1}_n \in \mathbb{R}^n$ is the vector of ones then
$$
A(a \odot b) = (A \odot (\mathbf{1}_n a^\top))b,\quad (Aa) \odot b = (A \odot (b \mathbf{1}_n^\top))a.
$$
\end{lemma}

\begin{algorithm}[Alternative expressions]\label{algorithm:planning2}
Denote the empirical kernel matrices
$
K_{SS}$, $
K_{DD}$, $K_{XX}$, $K_{MM}\in\mathbb{R}^{n\times n}
$
calculated from observations drawn from population $\mathbb{P}$.  Denote the empirical kernel matrices
$
K_{\tilde{D}\tilde{D}}$, $K_{\tilde{X}\tilde{X}}$, $K_{\tilde{M}\tilde{M}}\in\mathbb{R}^{n\times n}
$
calculated from observations drawn from population $\tilde{\mathbb{P}}$.  Dose response estimators have the closed form solutions
\begin{enumerate}
    \item $\hat{\theta}^{ATE}(d)= (S\odot Y)^{\top}( K_{SS}\odot K_{DD}\odot K_{XX}\odot K_{MM}+n\lambda I)^{-1}$ \\ $
        [K_{S1}\odot K_{Dd}\odot  \{K_{M M}(K_{DD}\odot K_{XX}+n\lambda_4)^{-1} \odot \frac{1}{n} K^2_{XX}\}K_{Dd}]  $,
    \item $\hat{\theta}^{DS}(d,\tilde{\mathbb{P}})= (S\odot Y)^{\top}( K_{SS}\odot K_{DD}\odot K_{XX}\odot K_{MM}+n\lambda I)^{-1}$ \\$
        [K_{S1}\odot  K_{Dd}\odot  \{K_{M \tilde{M}}(K_{\tilde{D}\tilde{D}}\odot K_{\tilde{X}\tilde{X}}+\tilde{n}\lambda_5 I)^{-1} \odot \frac{1}{n} (K_{X \tilde X}K_{\tilde X \tilde X})\}K_{\tilde{D}d}] $,
\end{enumerate}
where $(\lambda,\lambda_4,\lambda_5)$ are ridge regression penalty parameters.
\end{algorithm}

\begin{proof}[Derivation]
I adapt the techniques of \cite[Algorithm 10]{singh2021workshop} to the dynamic sample selection problem.

Denote
\begin{align*}
    R_1 &= (S\odot Y)^{\top}(K_{SS}\odot K_{DD}\odot K_{XX}\odot K_{XX}+n\lambda I)^{-1},\quad R_2 = K_{MM}(K_{DD}\odot K_{XX}+n\lambda_4 I )^{-1}
\end{align*}
Then by Algorithm~\ref{algorithm:static} and Lemma~\ref{lem:matrix_form}, 
\begin{align*}
\hat{\theta}^{ATE}(d)
    &= \frac{1}{n} \sum_{i=1}^n R_1 \{K_{S1}\odot K_{Dd}\odot K_{Xx_i}\odot  R_2(K_{Dd}  \odot K_{Xx_i}) \}\\
    &= \frac{1}{n} \sum_{i=1}^n R_1 [K_{S1}\odot K_{Dd}\odot K_{Xx_i} \odot  \{(R_2 \odot \mathbf{1}_n K^\top_{Xx_i}) K_{Dd}\} ]\\
    &= \frac{1}{n} \sum_{i=1}^n R_1 \{K_{S1}\odot K_{Dd}\odot  (R_2 \odot \mathbf{1}_n K^\top_{Xx_i} \odot K_{Xx_i} \mathbf{1}^\top_n) K_{Dd}\}\\
    &= R_1 \left\{K_{S1}\odot K_{Dd}\odot  \left(R_2 \odot \frac{1}{n} \sum_{i=1}^n K_{Xx_i} K^\top_{Xx_i}\right) K_{Dd}\right\} \\
    &= R_1 \left\{K_{S1}\odot K_{Dd}\odot  \left(R_2 \odot \frac{1}{n} K_{XX}^2\right) K_{Dd}\right\}.
\end{align*}
Note that I use the identity
$
    \mathbf{1}_n K^\top_{Xx_i} \odot K_{Xx_i} \mathbf{1}^\top_n =  K_{Xx_i} K_{Xx_i}^\top.
$
\end{proof}
\section{Nonparametric consistency proofs}\label{section:consistency_proof}

In this appendix, I (i) present technical lemmas for regression, unconditional mean embeddings, and conditional mean embeddings; (ii) appeal to these lemmas to prove uniform consistency of dose response curves, incremental response curves, and counterfactual distributions.

\subsection{Lemmas}

\subsubsection{Regression}

Recall classic results for the kernel ridge regression estimator $\hat{\gamma}$ of $\gamma_0(w):=\mathbb{E}(Y|W=w)$. As in Section~\ref{sec:rkhs_main}, $W$ is the concatenation of regressors. Consider the following notation:
\begin{align*}
    \gamma_0&=\argmin_{\gamma\in\mathcal{H}}\mathcal{E}(\gamma),\quad \mathcal{E}(\gamma)=\mathbb{E}[\{Y-\gamma(W)\}^2]; \\
    \gamma_{\lambda}&=\argmin_{\gamma\in\mathcal{H}}\mathcal{E}_{\lambda}(\gamma),\quad \mathcal{E}_{\lambda}(\gamma)=\mathcal{E}(\gamma)+\lambda\|\gamma\|^2_{\mathcal{H}}; \\
    \hat{\gamma}&=\argmin_{\gamma\in\mathcal{H}}\hat{\mathcal{E}}(\gamma),\quad \hat{\mathcal{E}}(\gamma)=\frac{1}{n}\sum_{i=1}^n\{Y_i-\gamma(W_i)\}^2+\lambda\|\gamma\|^2_{\mathcal{H}}.
\end{align*}

\begin{lemma}[Sampling error; Theorem 1 of \cite{smale2007learning}]\label{lemma:sampling}
Suppose Assumptions~\ref{assumption:RKHS} and~\ref{assumption:original} hold. Then with probability $1-\delta$,
$
\|\hat{\gamma}-\gamma_{\lambda}\|_{\mathcal{H}}\leq \frac{6 C \kappa_w \ln(2/\delta)}{\sqrt{n}\lambda}.
$
\end{lemma}

\begin{lemma}[Approximation error; Proposition 3 of \cite{caponnetto2007optimal}] \label{prop:approx_error}
Suppose Assumptions~\ref{assumption:original} and~\ref{assumption:smooth_gamma} hold. Then
$
\|\gamma_{\lambda}-\gamma_0\|_{\mathcal{H}}\leq \lambda^{\frac{c-1}{2}}\sqrt{\zeta}.
$
\end{lemma}

\begin{lemma}[Regression rate; Theorem D.1 of \cite{singh2020kernel}]\label{theorem:regression}
Suppose Assumptions~\ref{assumption:RKHS},~\ref{assumption:original}, and \ref{assumption:smooth_gamma} hold. Then with probability $1-\delta$,
$$
\|\hat{\gamma}-\gamma_0\|_{\mathcal{H}}\leq r_{\gamma}(n,\delta,c):=\dfrac{ \sqrt{\zeta}(c+1)}{4^{\frac{1}{c+1}}} \left\{\dfrac{6 C \kappa_w \ln(2/\delta)}{ \sqrt{n\zeta}(c-1)}\right\}^{\frac{c-1}{c+1}}.
$$
\end{lemma}

\begin{proof}
    Immediate from triangle inequality using Lemmas~\ref{lemma:sampling} and~\ref{prop:approx_error}.
\end{proof}

\begin{remark}\label{remark:1}
In sample selection problems, I consider $\gamma_0(w)=\mathbb{E}[SY|W=w]$ where $W$ and hence $\kappa_w$ varies.
\begin{enumerate}
    \item Static sample selection
    \begin{enumerate}
        \item $\theta_0^{ATE}$, $\theta_0^{DS}$, $\theta_0^{ATT}$: $\kappa_w=\kappa_s\kappa_d\kappa_x$;
        \item $\theta_0^{CATE}$: $\kappa_w=\kappa_s\kappa_d\kappa_v\kappa_x$.
    \end{enumerate}
    \item Dynamic sample selection: $\kappa_w=\kappa_s\kappa_d\kappa_x\kappa_m$.
\end{enumerate}
\end{remark}

\subsubsection{Unconditional mean embedding}

Next, I recall classic results for the unconditional mean embedding estimator $\hat{\mu}_w$ for $\mu_w:=\mathbb{E}\{\phi(W)\}$. Let $W$ be a placeholder random variable as in Section~\ref{sec:rkhs_main}.

\begin{lemma}[Bennett inequality; Lemma 2 of \cite{smale2007learning}]\label{lemma:prob}
Let $(\xi_i)$ be i.i.d. random variables drawn from distribution $\mathbb{P}$ taking values in a real separable Hilbert space $\mathcal{K}$. Suppose there exists $ \tilde{M}$ such that
$
    \|\xi_i\|_{\mathcal{K}} \leq \tilde{M}<\infty$ almost surely and $
    \sigma^2(\xi_i):=\mathbb{E}(\|\xi_i\|_{\mathcal{K}}^2)$. Then $\forall n\in\mathbb{N}, \forall \eta\in(0,1)$,
$$
\mathbb{P}\left\{\left\|\dfrac{1}{n}\sum_{i=1}^n\xi_i-\mathbb{E}(\xi)\right\|_{\mathcal{K}}\leq\dfrac{2\tilde{M}\ln(2/\eta)}{n}+\sqrt{\dfrac{2\sigma^2(\xi)\ln(2/\eta)}{n}}\right\}\geq 1-\eta.
$$
\end{lemma}

\begin{lemma}[Mean embedding rate; Theorem D.2 of \cite{singh2020kernel}]\label{theorem:mean}
Suppose Assumptions~\ref{assumption:RKHS} and~\ref{assumption:original} hold. Then with probability $1-\delta$, 
$$
\|\hat{\mu}_w-\mu_w\|_{\mathcal{H}_{\mathcal{W}}}\leq r_{\mu}(n,\delta):=\frac{4\kappa_w \ln(2/\delta)}{\sqrt{n}}.
$$
\end{lemma}
I quote a result that appeals to Lemma~\ref{lemma:prob}. \cite[Theorem 15]{altun2006unifying} originally prove this rate by McDiarmid inequality. See \cite[Theorem 2]{smola2007hilbert} for an argument via Rademacher complexity. See \cite[Proposition A.1]{tolstikhin2017minimax} for an improved constant and the proof that the rate is minimax optimal.

\begin{remark}\label{remark:2}
In various applications, $\kappa_w$ varies.
\begin{enumerate}
    \item Static sample selection.
    \begin{enumerate}
         \item $\theta_0^{ATE}$: with probability $1-\delta$, $
\|\hat{\mu}_x-\mu_x\|_{\mathcal{H}_{\mathcal{X}}}\leq r_{\mu}(n,\delta):=\frac{4\kappa_x \ln(2/\delta)}{\sqrt{n}}.
$
        \item $\theta_0^{DS}$: with probability $1-\delta$, 
$
\|\hat{\nu}_x-\nu_x\|_{\mathcal{H}_{\mathcal{X}}}\leq r_{\nu}(\tilde{n},\delta):=\frac{4\kappa_x \ln(2/\delta)}{\sqrt{\tilde{n}}}.
$
    \end{enumerate}
    \item Dynamic sample selection.
    \begin{enumerate}
        \item $\theta_0^{ATE}$: with probability $1-\delta$, $\forall d\in\mathcal{D}$
$$
\left\|\frac{1}{n}\sum_{i=1}^n\{\phi(X_i) \otimes \mu_{m}(d,X_i)\}-\int \{\phi(x)\otimes \mu_{m}(d,x)\} \mathrm{d}\mathbb{P}(x)\right\|_{\mathcal{H}_{\mathcal{X}}\otimes \mathcal{H}_{\mathcal{M}}} \leq r^{ATE}_{\mu}(n,\delta):=\frac{4\kappa_x \kappa_m \ln(2/\delta)}{\sqrt{n}}.$$
        \item $\theta_0^{DS}$: with probability $1-\delta$, $\forall d \in\mathcal{D}$
        $$
\left\|\frac{1}{n}\sum_{i=1}^n\{\phi(\tilde{X}_{i}) \otimes \mu_{m}(d,\tilde{X}_{i})\}-\int \{\phi(x)\otimes \mu_{m}(d,x)\} \mathrm{d}\tilde{\mathbb{P}}(x)\right\|_{\mathcal{H}_{\mathcal{X}}\otimes \mathcal{H}_{\mathcal{M}}}
\leq r^{DS}_{\nu}(\tilde{n},\delta):=\frac{4\kappa_x\kappa_m \ln(2/\delta)}{\sqrt{\tilde{n}}}.$$
    \end{enumerate}
\end{enumerate}
\end{remark}

\subsubsection{Conditional expectation operator and conditional mean embedding}

In Sections~\ref{section:static} and~\ref{section:dynamic} as well as Appendix~\ref{section:dist}, I consider the abstract conditonal expectation operator $E_j\in \mathcal{L}_2(\mathcal{H}_{\mathcal{A}_j},\mathcal{H}_{\mathcal{B}_j})$, where $\mathcal{A}_j$ and $\mathcal{B}_j$ are spaces that can be instantiated for different causal parameters.

\begin{lemma}[Conditional mean embedding rate; Theorem 2 of \cite{singh2019kernel}]\label{theorem:conditional}
Suppose Assumptions~\ref{assumption:RKHS},~\ref{assumption:original}, and~\ref{assumption:smooth_op} hold. Then with probability $1-\delta$, 
$$
\|\hat{E}_j-E_j\|_{\mathcal{L}_2}\leq r_E(\delta,n,c_j):=\dfrac{ \sqrt{\zeta_j}(c_j+1)}{4^{\frac{1}{c_j+1}}} \left\{\dfrac{4\kappa_b(\kappa_a+\kappa_b \|E_j\|_{\mathcal{L}_2}) \ln(2/\delta)}{ \sqrt{n\zeta_j}(c_j-1)}\right\}^{\frac{c_j-1}{c_j+1}}.
$$
Moreover, $\forall b\in\mathcal{B}_j$
$$
  \|\hat{\mu}_a(b)-\mu_a(b)\|_{\mathcal{H}_{\mathcal{A}_j}}\leq r_{\mu}(\delta,n,c_j):=\kappa_{b}\cdot
  \dfrac{ \sqrt{\zeta_j}(c_j+1)}{4^{\frac{1}{c_j+1}}} \left\{\dfrac{4\kappa_b(\kappa_a+\kappa_b \|E_j\|_{\mathcal{L}_2}) \ln(2/\delta)}{ \sqrt{n\zeta_j}(c_j-1)}\right\}^{\frac{c_j-1}{c_j+1}}.
    $$
\end{lemma}

\begin{remark}\label{remark:3}
Note that in various applications, $\kappa_a$ and $\kappa_b$ vary.
\begin{enumerate}
\item Static sample selection
\begin{enumerate}
    \item $\theta_0^{ATT}$: $\kappa_a=\kappa_x$, $\kappa_b=\kappa_d$;
        \item $\theta_0^{CATE}$: $\kappa_a=\kappa_x$, $\kappa_b=\kappa_v$;
\end{enumerate}
    \item Dynamic sample selection: 
    \begin{enumerate}
    \item $\theta_0^{ATE}$: $\kappa_a=\kappa_m$, $\kappa_b=\kappa_d\kappa_x$;
        \item $\theta_0^{DS}$: $\kappa_a=\kappa_m$, $\kappa_b=\kappa_d\kappa_x$;
\end{enumerate}
    \item Counterfactual distributions
    \begin{enumerate}
        \item Static $\theta_0^{D:ATE}$, $\theta_0^{D:DS}$, $\theta_0^{D:ATT}$: $\kappa_a=\kappa_y$, $\kappa_b=\kappa_s \kappa_d \kappa_x$.
        \item Static $\theta_0^{D:CATE}$: $\kappa_a=\kappa_y$, $\kappa_b=\kappa_s \kappa_d  \kappa_v \kappa_x$.
        \item Dynamic $\theta_0^{D:ATE}$, $\theta_0^{D:DS}$: $\kappa_a=\kappa_y$, $\kappa_b=\kappa_s \kappa_d \kappa_x \kappa_m$.
    \end{enumerate}
\end{enumerate}
\end{remark}

\subsection{Main results}

Appealing to Remarks~\ref{remark:notation},~\ref{remark:1},~\ref{remark:2}, and~\ref{remark:3}, I prove uniform consistency for (i) dose response curves, (ii) incremental response curves, and (iii) counterfactual distributions.

\subsubsection{Static sample selection}

\begin{proof}[Proof of Theorem~\ref{theorem:consistency_static}]
I generalize \cite[Theorem 6.5]{singh2020kernel} to the static sample selection problem.

Consider $\theta_0^{ATE}$. 
    \begin{align*}
        &\hat{\theta}^{ATE}(d)-\theta_0^{ATE}(d)
        =\langle \hat{\gamma} ,\phi(1)\otimes  \phi(d)\otimes \hat{\mu}_x \rangle_{\mathcal{H}} - \langle \gamma_0 , \phi(1)\otimes \phi(d)\otimes \mu_x \rangle_{\mathcal{H}} \\
        &=\langle \hat{\gamma} ,\phi(1)\otimes \phi(d)\otimes[\hat{\mu}_x-\mu_x] \rangle_{\mathcal{H}} + \langle [\hat{\gamma}-\gamma_0], \phi(1)\otimes \phi(d) \otimes \mu_x \rangle_{\mathcal{H}} \\
        &=\langle [\hat{\gamma}-\gamma_0],\phi(1)\otimes \phi(d)\otimes[\hat{\mu}_x-\mu_x] \rangle_{\mathcal{H}} \\
        &\quad + \langle \gamma_0,\phi(1)\otimes \phi(d)\otimes[\hat{\mu}_x-\mu_x] \rangle_{\mathcal{H}}\\
        &\quad+\langle [\hat{\gamma}-\gamma_0],\phi(1)\otimes \phi(d) \otimes \mu_x \rangle_{\mathcal{H}}.
    \end{align*}
    Therefore by Lemmas~\ref{theorem:regression} and~\ref{theorem:mean}, with probability $1-2\delta$
   \begin{align*}
       &|\hat{\theta}^{ATE}(d)-\theta_0^{ATE}(d)|\leq 
       \|\hat{\gamma}-\gamma_0\|_{\mathcal{H}}\|\phi(1)\|_{\mathcal{H}_{\mathcal{S}}} \|\phi(d)\|_{\mathcal{H}_{\mathcal{D}}} \|\hat{\mu}_x-\mu_x\|_{\mathcal{H}_{\mathcal{X}}}\\
       &\quad +
       \|\gamma_0\|_{\mathcal{H}}\|\phi(1)\|_{\mathcal{H}_{\mathcal{S}}}\|\phi(d)\|_{\mathcal{H}_{\mathcal{D}}}\|\hat{\mu}_x-\mu_x\|_{\mathcal{H}_{\mathcal{X}}} \\
       &\quad + 
       \|\hat{\gamma}-\gamma_0\|_{\mathcal{H}}\|\phi(1)\|_{\mathcal{H}_{\mathcal{S}}}\|\phi(d)\|_{\mathcal{H}_{\mathcal{D}}} \|\mu_x\|_{\mathcal{H}_{\mathcal{X}}}
      \\
      &\leq \kappa_s\kappa_d \cdot r_{\gamma}(n,\delta,c) \cdot r_{\mu}(n,\delta)+\kappa_s\kappa_d\cdot\|\gamma_0\|_{\mathcal{H}} \cdot r_{\mu}(n,\delta)+\kappa_s\kappa_d\kappa_x \cdot r_{\gamma}(n,\delta,c)\\
      &=O\left(n^{-\frac{1}{2}\frac{c-1}{c+1}}\right).
   \end{align*}
By the same argument, with probability $1-2\delta$
    \begin{align*}
   &|\hat{\theta}^{DS}(d,\tilde{\mathbb{P}})-\theta_0^{DS}(d,\tilde{\mathbb{P}})| \\
    &\leq \kappa_s\kappa_d \cdot r_{\gamma}(n,\delta,c) \cdot r_{\nu}(\tilde{n},\delta)+\kappa_s\kappa_d\cdot\|\gamma_0\|_{\mathcal{H}} \cdot r_{\nu}(\tilde{n},\delta)+\kappa_s\kappa_d\kappa_x \cdot r_{\gamma}(n,\delta,c)\\
      &=O\left( n^{-\frac{1}{2}\frac{c-1}{c+1}}+\tilde{n}^{-\frac{1}{2}}\right).
    \end{align*}
    Next, consider $\theta_0^{ATT}$.
    \begin{align*}
        &\hat{\theta}^{ATT}(d,d')-\theta_0^{ATT}(d,d') =\langle \hat{\gamma} ,\phi(1)\otimes \phi(d')\otimes \hat{\mu}_x(d) \rangle_{\mathcal{H}} - \langle \gamma_0 ,\phi(1)\otimes \phi(d')\otimes \mu_x(d) \rangle_{\mathcal{H}} \\
        &=\langle \hat{\gamma} ,\phi(1)\otimes \phi(d')\otimes[\hat{\mu}_x(d)-\mu_x(d)] \rangle_{\mathcal{H}} + \langle [\hat{\gamma}-\gamma_0],\phi(1)\otimes \phi(d') \otimes \mu_x(d) \rangle_{\mathcal{H}} \\
        &=\langle [\hat{\gamma}-\gamma_0], \phi(1)\otimes\phi(d')\otimes[\hat{\mu}_x(d)-\mu_x(d)] \rangle_{\mathcal{H}} \\
        &\quad + \langle \gamma_0, \phi(1)\otimes\phi(d')\otimes[\hat{\mu}_x(d)-\mu_x(d)] \rangle_{\mathcal{H}}\\
        &\quad +\langle [\hat{\gamma}-\gamma_0], \phi(1)\otimes\phi(d') \otimes \mu_x(d) \rangle_{\mathcal{H}}.
    \end{align*}
    Therefore by Lemmas~\ref{theorem:regression} and~\ref{theorem:conditional}, with probability $1-2\delta$
   \begin{align*}
       &|\hat{\theta}^{ATT}(d,d')-\theta_0^{ATT}(d,d')|
       \leq 
       \|\hat{\gamma}-\gamma_0\|_{\mathcal{H}}\|\phi(1)\|_{\mathcal{H}_{\mathcal{S}}}\|\phi(d')\|_{\mathcal{H}_{\mathcal{D}}} \|\hat{\mu}_x(d)-\mu_x(d)\|_{\mathcal{H}_{\mathcal{X}}}\\
       &\quad +\|\gamma_0\|_{\mathcal{H}}\|\phi(1)\|_{\mathcal{H}_{\mathcal{S}}}\|\phi(d')\|_{\mathcal{H}_{\mathcal{D}}}\|\hat{\mu}_x(d)-\mu_x(d)\|_{\mathcal{H}_{\mathcal{X}}} \\
       &\quad +
       \|\hat{\gamma}-\gamma_0\|_{\mathcal{H}}\|\phi(1)\|_{\mathcal{H}_{\mathcal{S}}}\|\phi(d')\|_{\mathcal{H}_{\mathcal{D}}} \|\mu_x(d)\|_{\mathcal{H}_{\mathcal{X}}}
      \\
      &\leq \kappa_s\kappa_d \cdot r_{\gamma}(n,\delta,c) \cdot r_{\mu}^{ATT}(n,\delta,c_1)+\kappa_s\kappa_d\cdot\|\gamma_0\|_{\mathcal{H}} \cdot r_{\mu}^{ATT}(n,\delta,c_1)
      +\kappa_s\kappa_d\kappa_x \cdot r_{\gamma}(n,\delta,c)
      \\
      &=O\left(n^{-\frac{1}{2}\frac{c-1}{c+1}}+n^{-\frac{1}{2}\frac{c_1-1}{c_1+1}}\right).
   \end{align*}
    Finally, consider $\theta_0^{CATE}$.
    \begin{align*}
        &\hat{\theta}^{CATE}(d,v)-\theta_0^{CATE}(d,v)=\langle \hat{\gamma} , \phi(1)\otimes\phi(d)\otimes \phi(v)\otimes \hat{\mu}_{x}(v) \rangle_{\mathcal{H}} - \langle \gamma_0 ,\phi(1)\otimes \phi(d )\otimes \phi(v) \otimes \mu_{x}(v) \rangle_{\mathcal{H}} \\
        &=\langle \hat{\gamma} ,\phi(1)\otimes \phi(d)\otimes \phi(v)\otimes[\hat{\mu}_{x}(v)-\mu_{x}(v)] \rangle_{\mathcal{H}} + \langle [\hat{\gamma}-\gamma_0],\phi(1)\otimes \phi(d)\otimes \phi(v) \otimes \mu_{x}(v) \rangle_{\mathcal{H}} \\
        &=\langle [\hat{\gamma}-\gamma_0],\phi(1)\otimes \phi(d)\otimes \phi(v)\otimes[\hat{\mu}_{x}(v)-\mu_{x}(v)] \rangle_{\mathcal{H}}  \\
        &\quad + \langle \gamma_0,\phi(1)\otimes \phi(d)\otimes \phi(v)\otimes[\hat{\mu}_{x}(v)-\mu_{x}(v)] \rangle_{\mathcal{H}}\\
        &\quad +\langle [\hat{\gamma}-\gamma_0],\phi(1)\otimes \phi(d)\otimes \phi(v) \otimes \mu_{x}(v) \rangle_{\mathcal{H}}.
    \end{align*}
    Therefore by Lemmas~\ref{theorem:regression} and~\ref{theorem:conditional}, with probability $1-2\delta$
   \begin{align*}
       &|\hat{\theta}^{CATE}(d,v)-\theta_0^{CATE}(d,v)|\leq 
       \|\hat{\gamma}-\gamma_0\|_{\mathcal{H}}\|\phi(1)\|_{\mathcal{H}_{\mathcal{S}}}\|\phi(d)\|_{\mathcal{H}_{\mathcal{D}}}\|\phi(v)\|_{\mathcal{H}_{\mathcal{V}}} \|\hat{\mu}_{x}(v)-\mu_{x}(v)\|_{\mathcal{H}_{\mathcal{X}}}\\
      &\quad+
       \|\gamma_0\|_{\mathcal{H}}\|\phi(1)\|_{\mathcal{H}_{\mathcal{S}}}\|\phi(d)\|_{\mathcal{H}_{\mathcal{D}}}\|\phi(v)\|_{\mathcal{H}_{\mathcal{V}}}\|\hat{\mu}_{x}(v)-\mu_{x}(v)\|_{\mathcal{H}_{\mathcal{X}}} \\
       &\quad+
       \|\hat{\gamma}-\gamma_0\|_{\mathcal{H}}\|\phi(1)\|_{\mathcal{H}_{\mathcal{S}}}\|\phi(d)\|_{\mathcal{H}_{\mathcal{D}}}\|\phi(v)\|_{\mathcal{H}_{\mathcal{V}}} \|\mu_{x}(v)\|_{\mathcal{H}_{\mathcal{X}}}
      \\
      &\leq \kappa_s\kappa_d\kappa_{v} \cdot r_{\gamma}(n,\delta,c) \cdot r_{\mu}^{CATE}(n,\delta,c_2)
      +\kappa_s\kappa_d\kappa_{v}\cdot\|\gamma_0\|_{\mathcal{H}} \cdot r_{\mu}^{CATE}(n,\delta,c_2)
      +\kappa_s\kappa_d\kappa_{v} \kappa_{x} \cdot r_{\gamma}(n,\delta,c)
      \\
      &=O\left(n^{-\frac{1}{2}\frac{c-1}{c+1}}+n^{-\frac{1}{2}\frac{c_2-1}{c_2+1}}\right).
   \end{align*}
    For incremental responses, replace $\phi(d)$ with $\nabla_d \phi(d)$ and hence replace $\|\phi(d)\|_{\mathcal{H}_{\mathcal{D}}}\leq \kappa_d$ with $\|\nabla_d \phi(d)\|_{\mathcal{H}_{\mathcal{D}}}\leq \kappa_d'$.
\end{proof}

\subsubsection{Dynamic sample selection}

To lighten notation, define
\begin{align*}
    \Delta_p&:=\frac{1}{n}\sum_{i=1}^n \{\phi(X_i)\otimes \hat{\mu}_{m}(d,X_i) \}-\int \{\phi(x)\otimes \mu_{m}(d,x) \}\mathrm{d}\mathbb{P}(x), \\
    \Delta_q&:=\frac{1}{\tilde{n}}\sum_{i=1}^{\tilde{n}} \{\phi(\tilde{X}_i)\otimes \hat{\nu}_{m}(d,\tilde{X}_i) \}-\int \{\phi(x)\otimes \nu_{m}(d,x) \}\mathrm{d}\tilde{\mathbb{P}}(x).
\end{align*}

\begin{proposition}\label{prop:delta_p}
Suppose Assumptions~\ref{assumption:RKHS} and~\ref{assumption:original} hold.
\begin{enumerate}
    \item If in addition Assumption~\ref{assumption:smooth_op} holds with with $\mathcal{A}_4=\mathcal{X}$ and $\mathcal{B}_4=\mathcal{D}\times \mathcal{X}$ then with probability $1-2\delta$
    \begin{align*}
        &\left\|\Delta_p\right\|_{\mathcal{H}_{\mathcal{X}}\otimes\mathcal{H}_{\mathcal{M}}}\leq \kappa_x \cdot r^{ATE}_{\mu}(n,\delta,c_4)+r^{ATE}_{\mu}(n,\delta).
\end{align*}
    \item If in addition Assumption~\ref{assumption:smooth_op} hold with $\mathcal{A}_5=\mathcal{X}$ and $\mathcal{B}_5=\mathcal{D}\times \mathcal{X}$ then with probability $1-2\delta$
    \begin{align*}
        &\left\|\Delta_q\right\|_{\mathcal{H}_{\mathcal{X}}\otimes\mathcal{H}_{\mathcal{M}}}\leq \kappa_x \cdot r^{DS}_{\nu}(\tilde{n},\delta,c_5)+r^{DS}_{\nu}(\tilde{n},\delta).
\end{align*}
\end{enumerate}
\end{proposition}

\begin{proof}
  By triangle inequality, 
\begin{align*}
    \left\|\Delta_p\right\|_{\mathcal{H}_{\mathcal{X}}\otimes\mathcal{H}_{\mathcal{M}}}
    &\leq \left\|\frac{1}{n}\sum_{i=1}^n \{\phi(X_i)\otimes \hat{\mu}_{m}(d,X_i) \}-\{\phi(X_i)\otimes \mu_{m}(d,X_i) \}\right\|_{\mathcal{H}_{\mathcal{X}}\otimes\mathcal{H}_{\mathcal{M}}} \\
    &\quad +  \left\|\frac{1}{n}\sum_{i=1}^n\{\phi(X_i) \otimes \mu_{m}(d,X_i)\}-\int \{\phi(x)\otimes \mu_{m}(d,x)\} \mathrm{d}\mathbb{P}(x)\right\|_{\mathcal{H}_{\mathcal{X}}\otimes \mathcal{H}_{\mathcal{M}}}.
\end{align*}
  Focusing on the former term, by Lemma~\ref{theorem:conditional}
  \begin{align*}
      &\left\|\frac{1}{n}\sum_{i=1}^n \{\phi(X_i)\otimes \hat{\mu}_{m}(d,X_i) \}-\{\phi(X_i)\otimes \mu_{m}(d,X_i) \}\right\|_{\mathcal{H}_{\mathcal{X}}\otimes\mathcal{H}_{\mathcal{M}}} \\
      &=\left\|\frac{1}{n}\sum_{i=1}^n  \phi(X_i) \otimes \{\hat{\mu}_{m}(d,X_i)-\mu_{m}(d,X_i)\} \right\|_{\mathcal{H}_{\mathcal{X}}\otimes\mathcal{H}_{\mathcal{M}}} \\
      &\leq \kappa_x \cdot \sup_{x\in\mathcal{X}}\left\| \hat{\mu}_{m}(d,x)-\mu_{m}(d,x)\right\|_{\mathcal{H}_{\mathcal{X}}} \\
      &\leq \kappa_x \cdot r^{ATE}_{\mu}(n,\delta,c_4).
  \end{align*}
  Focusing on the latter term, by Lemma~\ref{theorem:mean}
  \begin{align*}
      \left\|\frac{1}{n}\sum_{i=1}^n\{\phi(X_i) \otimes \mu_{m}(d,X_i)\}-\int \{\phi(x)\otimes \mu_{m}(d,x)\} \mathrm{d}\mathbb{P}(x)\right\|_{\mathcal{H}_{\mathcal{X}}\otimes \mathcal{H}_{\mathcal{M}}}\leq r^{ATE}_{\mu}(n,\delta).
  \end{align*}
  The argument for $\theta_0^{DS}$ is identical.
\end{proof}

\begin{proof}[Proof of Theorem~\ref{theorem:consistency_planning}]
I adapt the techniques of \cite[Theorem 7]{singh2021workshop} to the dynamic sample selection problem.

Consider $\theta_0^{ATE}$.
    \begin{align*}
        &\hat{\theta}^{ATE}(d)-\theta_0^{ATE}(d)\\
        &=\langle \hat{\gamma}, \phi(1)\otimes \phi(d)\otimes \frac{1}{n}\sum_{i=1}^n\{\phi(X_i)\otimes \hat{\mu}_{m}(d,X_i)\}  \rangle_{\mathcal{H}}  - \langle \gamma_0,\phi(1) \otimes  \phi(d)\otimes \int \phi(x)\otimes \mu_{m}(d,x) \mathrm{d}\mathbb{P}(x) \rangle_{\mathcal{H}} \\
        &=\langle \hat{\gamma}, \phi(1)\otimes  \phi(d)\otimes
        \Delta_p \rangle_{\mathcal{H}} +\langle (\hat{\gamma}-\gamma_0),  \phi(1)\otimes\phi(d) \otimes \int \phi(x)\otimes \mu_{m}(d,x) \mathrm{d}\mathbb{P}(x) \rangle_{\mathcal{H}}  \\
        &=\langle (\hat{\gamma}-\gamma_0), \phi(1)\otimes\phi(d)\otimes
        \Delta_p \rangle_{\mathcal{H}} 
        \\
       &\quad+\langle \gamma_0, \phi(1)\otimes\phi(d)\otimes
        \Delta_p \rangle_{\mathcal{H}} \\
        &\quad +\langle (\hat{\gamma}-\gamma_0),  \phi(1)\otimes\phi(d) \otimes \int \phi(x)\otimes \mu_{m}(d,x) \mathrm{d}\mathbb{P}(x) \rangle_{\mathcal{H}}.
    \end{align*}
    Therefore by Lemmas~\ref{theorem:regression},~\ref{theorem:mean}, and~\ref{theorem:conditional} as well as~Proposition~\ref{prop:delta_p}, with probability $1-3\delta$
    \begin{align*}
        &|\hat{\theta}^{ATE}(d)-\theta^{ATE}_0(d)|\\
        &\leq \|\hat{\gamma}-\gamma_0\|_{\mathcal{H}}\|\phi(1)\|_{\mathcal{H}_{\mathcal{S}}}\|\phi(d)\|_{\mathcal{H}_{\mathcal{D}}} 
        \left\|\Delta_p\right\|_{\mathcal{H}_{\mathcal{X}}\otimes\mathcal{H}_{\mathcal{M}}}
       \\
       &\quad +
       \|\gamma_0\|_{\mathcal{H}}\|\phi(1)\|_{\mathcal{H}_{\mathcal{S}}}\|\phi(d)\|_{\mathcal{H}_{\mathcal{D}}}
       \left\|\Delta_p\right\|_{\mathcal{H}_{\mathcal{X}}\otimes\mathcal{H}_{\mathcal{M}}}
       \\
       &\quad+
       \|\hat{\gamma}-\gamma_0\|_{\mathcal{H}}\|\phi(1)\|_{\mathcal{H}_{\mathcal{S}}}\|\phi(d)\|_{\mathcal{H}_{\mathcal{D}}}\times \left\|\int \{\phi(x)\otimes \mu_{m}(d,x) \}\mathrm{d}\mathbb{P}(x)\right\|_{\mathcal{H}_{\mathcal{X}}\otimes\mathcal{H}_{\mathcal{M}}}
      \\
      &\leq \kappa_s\kappa_d \cdot r_{\gamma}(n,\delta,c) \cdot \{\kappa_x \cdot r^{ATE}_{\mu}(n,\delta,c_4)+r^{ATE}_{\mu}(n,\delta)\}\\
      &\quad +\kappa_s\kappa_d\cdot\|\gamma_0\|_{\mathcal{H}} \cdot \{\kappa_x \cdot r^{ATE}_{\mu}(n,\delta,c_4)+r^{ATE}_{\mu}(n,\delta)\}\\
       &\quad+\kappa_s\kappa_d\kappa_x\kappa_m \cdot r_{\gamma}(n,\delta,c)
      \\
      &=O\left(n^{-\frac{1}{2}\frac{c-1}{c+1}}+n^{-\frac{1}{2}\frac{c_4-1}{c_4+1}}\right).
    \end{align*}
    By the same argument
  \begin{align*}
        &|\hat{\theta}^{DS}(d)-\theta^{DS}_0(d)|\\
      &\leq \kappa_s\kappa_d \cdot r_{\gamma}(n,\delta,c) \cdot \{\kappa_x \cdot r^{DS}_{\nu}(\tilde{n},\delta,c_5)+r^{DS}_{\nu}(\tilde{n},\delta)\}\\
      &\quad +\kappa_s\kappa_d\cdot\|\gamma_0\|_{\mathcal{H}} \cdot \{\kappa_x \cdot r^{DS}_{\nu}(\tilde{n},\delta,c_5)+r^{DS}_{\nu}(\tilde{n},\delta)\}\\
      &\quad +\kappa_s\kappa_d\kappa_x\kappa_m \cdot r_{\gamma}(n,\delta,c)
      \\
      &=O\left(n^{-\frac{1}{2}\frac{c-1}{c+1}}+\tilde{n}^{-\frac{1}{2}\frac{c_5-1}{c_5+1}}\right).
    \end{align*}
   For incremental responses, replace $\phi(d)$ with $\nabla_d \phi(d)$ and hence replace $\|\phi(d)\|_{\mathcal{H}_{\mathcal{D}}}\leq \kappa_d$ with $\|\nabla_d \phi(d)\|_{\mathcal{H}_{\mathcal{D}}}\leq \kappa_d'$.
\end{proof}

\subsubsection{Counterfactual distributions}

\begin{proof}[Proof of Theorem~\ref{theorem:consistency_dist}]
   The argument is analogous to Theorems~\ref{theorem:consistency_static} and~\ref{theorem:consistency_planning}, replacing $\|\gamma_0\|_{\mathcal{H}}$ with $\|E_8\|_{\mathcal{L}_2}$ or $\|E_9\|_{\mathcal{L}_2}$  and replacing $r_{\gamma}(n,\delta,c)$ with $r_E(n,\delta,c_8)$ or $r_E(n,\delta,c_9)$.
\end{proof}

\begin{proof}[Proof of Theorem~\ref{theorem:conv_dist}]
 I generalize \cite[Theorem A.7]{singh2020kernel} to sample selection problems.
   
    Fix $d$. By Theorem~\ref{theorem:consistency_dist}
    $$
    \|\hat{\theta}^{ATE}(d)-\check{\theta}_0^{ATE}(d)\|_{\mathcal{H}_{\mathcal{Y}}}=O_p\left(n^{-\frac{1}{2}\frac{c_8-1}{c_8+1}}\right).
    $$
    Denote the samples constructed by Algorithm~\ref{algorithm:herding} by $(\tilde{Y}_j)^m_{j=1}$. Then by \cite[Section 4.2]{bach2012equivalence}
    $$
    \left\|\hat{\theta}^{ATE}(d)-\frac{1}{m}\sum_{j=1}^m \phi(\tilde{Y}_j)\right\|_{\mathcal{H}_{\mathcal{Y}}}=O(m^{-\frac{1}{2}}).
    $$
    Therefore by triangle inequality
    $$
    \left\|\frac{1}{m}\sum_{j=1}^m \phi(\tilde{Y}_j)-\check{\theta}_0^{ATE}(d)\right\|_{\mathcal{H}_{\mathcal{Y}}}=O_p\left(n^{-\frac{1}{2}\frac{c_8-1}{c_8+1}}+m^{-\frac{1}{2}}\right).
    $$
    The desired result follows from \cite{sriperumbudur2016optimal}, as quoted by \cite[Theorem 1.1]{simon2020metrizing}. The arguments for other counterfactual distributions are identical.
\end{proof}
\section{Semiparametric inference proofs}\label{sec:inference_proof}

In this appendix, I (i) present technical lemmas for regression; (ii) appeal to these lemmas to prove $\sqrt{n}$ consistency, Gaussian approximation, and semiparametric efficiency of treatment effects for static and dynamic sample selection.

\subsection{Lemmas}

Recall the various nonparametric objects required for inference. For static sample selection,
\begin{align*}
    \gamma_0(s,d,x)&=\mathbb{E}(SY|S=s,D=d,M=m,X=x)
\end{align*}
and $ \alpha^{(m)}_0(s,d,x)$ is given in Theorem~\ref{theorem:dr_static}. For dynamic sample selection,
\begin{align*}
\gamma_0(s,d,x,m)&=\mathbb{E}(SY|S=s,D=d,X=x,M=m),\\
\pi_0(x)&=\mathbb{E}(D|X=x),\\
\rho_0(d,x,m)&=\mathbb{E}(S|D=d,X=x,M=m),\\
\omega_0(s,d;x)&=\int \gamma_0(s,d,x,m) \mathrm{d}\mathbb{P}(m|d,x).
\end{align*}
I present uniform and mean square rates for these nonparametric objects.

\subsubsection{Uniform rate}

Observe that $(\gamma_0,\pi_0,\rho_0)$ can be estimated by nonparametric regressions. I quote a classic result for kernel ridge regression then specialize it for these various nonparametric regressions. For the abstract result, write
\begin{align*}
    \gamma_0&=\argmin_{\gamma\in\mathcal{H}}\mathcal{E}(\gamma),\quad \mathcal{E}(\gamma)=\mathbb{E}[\{Y-\gamma(W)\}^2], \\
    \hat{\gamma}&=\argmin_{\gamma\in\mathcal{H}}\hat{\mathcal{E}}(\gamma),\quad \hat{\mathcal{E}}(\gamma)=\frac{1}{n}\sum_{i=1}^n\{Y_i-\gamma(W_i)\}^2+\lambda\|\gamma\|^2_{\mathcal{H}}.
\end{align*}

\begin{lemma}[Uniform regression rate; Proposition 5 of \cite{singh2021workshop}]\label{prop:unif}
Suppose Assumptions~\ref{assumption:RKHS},~\ref{assumption:original}, and \ref{assumption:smooth_gamma} hold. Then with probability $1-\delta$
$$
\|\hat{\gamma}-\gamma_0\|_{\infty}\leq r_{\gamma,{\infty}}(n,\delta,c):=\kappa_w\dfrac{ \sqrt{\zeta}(c+1)}{4^{\frac{1}{c+1}}} \left\{\dfrac{6 C \kappa_w \ln(2/\delta)}{ \sqrt{n\zeta}(c-1)}\right\}^{\frac{c-1}{c+1}}=O\left(n^{-\frac{1}{2}\frac{c-1}{c+1}}\right).
$$
To lighten notation, I write summarize the asymptotic rate as
$
r_{\gamma,{\infty}}(n,c)=n^{-\frac{1}{2}\frac{c-1}{c+1}}
$.
\end{lemma}

\begin{remark}\label{remark:unif}
In various applications, $(\kappa_w,c,\zeta,C)$ vary.
\begin{enumerate}
    \item Static sample selection
    \begin{enumerate}
        \item $\theta_0^{ATE}$, $\theta_0^{DS}$, $\theta_0^{ATT}$: $\|\hat{\gamma}-\gamma_0\|_{\infty}\leq r^{ATE}_{\gamma,{\infty}}(n,\delta,c)$ with $(\kappa_s\kappa_d\kappa_x,c,\zeta,C)$, summarized by $r^{ATE}_{\gamma,{\infty}}(n,c)$.
         \item $\theta_0^{CATE}$: $\|\hat{\gamma}-\gamma_0\|_{\infty}\leq r^{CATE}_{\gamma,{\infty}}(n,\delta,c)$ with $(\kappa_s\kappa_d\kappa_v\kappa_x,c,\zeta,C)$, summarized by $r^{CATE}_{\gamma,{\infty}}(n,c)$.
    \end{enumerate}
    \item Dynamic sample selection
        \begin{enumerate}
         \item $\|\hat{\gamma}-\gamma_0\|_{\infty}\leq r^{ATE}_{\gamma,{\infty}}(n,\delta,c)$ with $(\kappa_s\kappa_d\kappa_x\kappa_m,c,\zeta,C)$, summarized by $r^{ATE}_{\gamma,{\infty}}(n,c)$.
        \item $\|\hat{\pi}-\pi_0\|_{\infty}\leq r^{ATE}_{\pi,{\infty}}(n,\delta,c_6)$  with $(\kappa_x,c_6,\zeta_6,1)$, summarized by $r^{ATE}_{\pi,{\infty}}(n,c_6)$.
        \item $\|\hat{\rho}-\rho_0\|_{\infty}\leq r^{ATE}_{\rho,{\infty}}(n,\delta,c_7)$  with $(\kappa_d\kappa_x\kappa_m,c_7,\zeta_7,1)$, summarized by $r^{ATE}_{\rho,{\infty}}(n,c_7)$.
    \end{enumerate}
\end{enumerate}
\end{remark}

Next I provide a uniform rate for $\omega_0$, using nonparametric techniques developed in Appendix~\ref{section:consistency_proof}.

\begin{proposition}[Uniform $\omega$ rate]\label{prop:unif_omega}
Suppose Assumptions~\ref{assumption:RKHS},~\ref{assumption:original},~\ref{assumption:smooth_gamma}, and~\ref{assumption:smooth_op} hold with $\mathcal{A}_4=\mathcal{X}$ and $\mathcal{B}_4=\mathcal{D}\times \mathcal{X}$. Then with probability $1-2\delta$
\begin{align*}
    &\|\hat{\omega}-\omega_0\|_{\infty}\leq r^{ATE}_{\omega,\infty}(n,\delta,c,c_4)\\
    &:=\kappa_s\kappa_d\kappa_x \cdot r_{\gamma}(n,\delta,c) \cdot r^{ATE}_{\mu}(n,\delta,c_4)+\kappa_s\kappa_d\kappa_x \cdot \|\gamma_0\|_{\mathcal{H}}\cdot  r^{ATE}_{\mu}(n,\delta,c_4)+\kappa_s\kappa_d\kappa_x\kappa_m\cdot  r_{\gamma}(n,\delta,c).
\end{align*}
I summarize the rate as $r^{ATE}_{\omega,\infty}(n,c,c_4)=O\left(n^{-\frac{1}{2}\frac{c-1}{c+1}}+n^{-\frac{1}{2}\frac{c_4-1}{c_4+1}}\right)$.
\end{proposition}

\begin{proof}
I adapt the techniques of \cite[Proposition 6]{singh2021workshop} to the dynamic sample selection problem.

Fix $(s,d,x)$. Then
    \begin{align*}
        &\hat{\omega}(s,d;x)-\omega_0(s,d;x)\\
        &=\langle\hat{\gamma},\phi(s)\otimes \phi(d) \otimes \phi(x)\otimes \hat{\mu}_{m}(d,x)\rangle_{\mathcal{H}}
        -\langle \gamma_0,\phi(s)\otimes \phi(d) \otimes \phi(x)\otimes \mu_{m}(d,x) \rangle_{\mathcal{H}}  \\
        &=\langle\hat{\gamma},\phi(s)\otimes \phi(d) \otimes \phi(x)\otimes \{\hat{\mu}_{m}(d,x)-\mu_{m}(d,x)\}\rangle_{\mathcal{H}}
        \\
        &\quad +\langle (\hat{\gamma}-\gamma_0),\phi(s)\otimes \phi(d) \otimes \phi(x)\otimes \mu_{m}(d,x) \rangle_{\mathcal{H}} \\
        &=\langle(\hat{\gamma}-\gamma_0),\phi(s)\otimes \phi(d) \otimes \phi(x)\otimes \{\hat{\mu}_{m}(d,x)-\mu_{m}(d,x)\}\rangle_{\mathcal{H}}\\
        &\quad +\langle \gamma_0,\phi(s)\otimes \phi(d) \otimes \phi(x)\otimes \{\hat{\mu}_{m}(d,x)-\mu_{m}(d,x)\}\rangle_{\mathcal{H}}\\
        &\quad +\langle (\hat{\gamma}-\gamma_0),\phi(s)\otimes \phi(d) \otimes \phi(x)\otimes \mu_{m}(d,x) \rangle_{\mathcal{H}} .
    \end{align*}
    Therefore by Lemmas~\ref{theorem:regression} and~\ref{theorem:conditional}, with probability $1-2\delta$
    \begin{align*}
        &|\hat{\omega}(s,d;x)-\omega_0(s,d;x)| \\
        &\leq \|\hat{\gamma}-\gamma_0\|_{\mathcal{H}}\|\phi(s)\|_{\mathcal{H}_{\mathcal{S}}}\|\phi(d)\|_{\mathcal{H}_{\mathcal{D}}} \|\phi(x)\|_{\mathcal{H}_{\mathcal{X}}} \|\hat{\mu}_{m}(d,x)-\mu_{m}(d,x)\|_{\mathcal{H}_{\mathcal{M}}} \\
        &\quad + \|\gamma_0\|_{\mathcal{H}} \|\phi(s)\|_{\mathcal{H}_{\mathcal{S}}}\|\phi(d)\|_{\mathcal{H}_{\mathcal{D}}}\|\phi(x)\|_{\mathcal{H}_{\mathcal{X}}} \|\hat{\mu}_{m}(d,x)-\mu_{m}(d,x)\|_{\mathcal{H}_{\mathcal{M}}}  \\
        &\quad + \|\hat{\gamma}-\gamma_0\|_{\mathcal{H}}\|\phi(s)\|_{\mathcal{H}_{\mathcal{S}}}\|\phi(d)\|_{\mathcal{H}_{\mathcal{D}}} \|\phi(x)\|_{\mathcal{H}_{\mathcal{X}}} \|\mu_{m}(d,x)\|_{\mathcal{H}_{\mathcal{M}}} \\
        &\leq \kappa_s\kappa_d\kappa_x \cdot r_{\gamma}(n,\delta,c) \cdot r^{ATE}_{\mu}(n,\delta,c_4)+\kappa_s\kappa_d\kappa_x \cdot \|\gamma_0\|_{\mathcal{H}}\cdot  r^{ATE}_{\mu}(n,\delta,c_4)+\kappa_s\kappa_d\kappa_x\kappa_m\cdot  r_{\gamma}(n,\delta,c) \\
        &=O\left(n^{-\frac{1}{2}\frac{c-1}{c+1}}+n^{-\frac{1}{2}\frac{c_4-1}{c_4+1}}\right).
    \end{align*}
\end{proof}

\subsubsection{Mean square rate}

Observe that $(\pi_0,\rho_0)$ can be estimated by nonparametric regressions. I quote an abstract result for kernel ridge regression then specialize it for these nonparametric regressions. For the abstract result, write
\begin{align*}
    \gamma_0&=\argmin_{\gamma\in\mathcal{H}}\mathcal{E}(\gamma),\quad \mathcal{E}(\gamma)=\mathbb{E}[\{Y-\gamma(W)\}^2], \\
    \hat{\gamma}&=\argmin_{\gamma\in\mathcal{H}}\hat{\mathcal{E}}(\gamma),\quad \hat{\mathcal{E}}(\gamma)=\frac{1}{n}\sum_{i=1}^n\{Y_i-\gamma(W_i)\}^2+\lambda\|\gamma\|^2_{\mathcal{H}}.
\end{align*}

\begin{lemma}[Regression mean square rate; Theorems 1 and 2 of \cite{caponnetto2007optimal}]\label{theorem:mse}
Suppose Assumptions~\ref{assumption:RKHS},~\ref{assumption:original}, and~\ref{assumption:smooth_gamma} hold. Suppose $\lambda_j(k_{\mathcal{W}})\asymp j^{-b}$. Calibrate the ridge regularization sequence such that $$
\lambda=\begin{cases} n^{-\frac{1}{2}} & \text{ if }\quad b=\infty; \\
n^{-\frac{b}{bc+1}} & \text{ if }\quad b\in(1,\infty),\; c\in(1,2]; \\
\ln^{\frac{b}{b+1}}(n)\cdot n^{-\frac{b}{b+1}} & \text{ if }\quad b\in(1,\infty),\; c=1.
\end{cases}
$$
Then the kernel ridge regression $\hat{\gamma}$ trained on $(W_i)_{i=1}^n$ satisfies
$$
\lim_{\tau\rightarrow \infty} \lim_{n\rightarrow\infty} \sup_{P\in\mathcal{P}(b,c)} \mathbb{P}_{(W_i)\sim P^{n}}(\|\hat{\gamma}-\gamma_0\|_2^2>\tau \cdot  r^2_{\gamma,2}(n,b,c))=0,
$$
where $\|\hat{\gamma}-\gamma_0\|_2^2=\mathbb{E}[\{\hat{\gamma}(W)-\gamma_0(W)\}^2|(W_i)_{i=1}^n]$ and
$$
r^2_{\gamma,2}(n,b,c)=\begin{cases} n^{-1} & \text{ if }\quad b=\infty; \\
n^{-\frac{bc}{bc+1}} & \text{ if }\quad b\in(1,\infty),\; c\in(1,2]; \\
\ln^{\frac{b}{b+1}}(n)\cdot n^{-\frac{b}{b+1}} & \text{ if }\quad b\in(1,\infty),\; c=1.
\end{cases}
$$
Moreover, the rate is optimal when $b\in(1,\infty)$ and $c\in(1,2]$. It is optimal up to a logarithmic factor when $b\in(1,\infty)$ and $c=1$.
\end{lemma}

\begin{assumption}[Mean square continuity]\label{assumption:cont}
Assume that there exists $Q<\infty$ such that, for any function $\gamma \in \Gamma$, where $\Gamma$ is a subset of $L^2$,
$$
\int m(s,d,x;\gamma)^2 \mathrm{d}\mathbb{P}(s,d,x) \leq Q \int \gamma(s,d,x)^2 \mathrm{d}\mathbb{P}(d,x).
$$
For $\theta_0^{DS}$, I use $\tilde{\mathbb{P}}$. For $\theta_0^{CATE}$, I use $\gamma_0(s,d,v,x)=\mathbb{E}[SY|S=s,D=d,V=v,X=x]$.
\end{assumption}

\begin{proposition}[Verifying mean square continuity]\label{prop:cont}
Assumption~\ref{assumption:bounded_Riesz} implies Assumption~\ref{assumption:cont} for each functional in Theorem~\ref{theorem:dr_static}.
\end{proposition}

\begin{proof}
The argument uses the same proof technique as Theorem~\ref{theorem:dr_static}. For $\theta_0^{ATE}$, observe that by law of iterated expectations
\begin{align*}
    \gamma(1,d,X)^2
&=\mathbb{E}[S\cdot \mathbbm{1}_{D=d}\cdot \gamma(S,D,X)^2|S=1,D=d,X] \\
&=\mathbb{E}\left[\frac{S\cdot \mathbbm{1}_{D=d}\cdot\gamma(S,D,X)^2}{\mathbb{P}(S=1,D=d|X)}\mid S=1,D=d,X\right]\mathbb{P}(S=1,D=d|X) \\
&=\mathbb{E}\left[\frac{S\cdot \mathbbm{1}_{D=d}\cdot \gamma(S,D,X)^2}{\mathbb{P}(S=1,D=d|X)}\mid X\right] \\
&=\mathbb{E}\left[\frac{S\cdot \mathbbm{1}_{D=d}\cdot \gamma(S,D,X)^2}{\mathbb{P}(S=1|d,X)\mathbb{P}(d|X)}\mid X\right]
\end{align*}
Hence by law of iterated expectation and Assumption~\ref{assumption:bounded_Riesz},
\begin{align*}
    \mathbb{E}[\gamma(1,d,X)^2 ] 
    &=\mathbb{E}\left[\mathbb{E}\left[\frac{S\cdot \mathbbm{1}_{D=d}\cdot \gamma(S,D,X)^2}{\mathbb{P}(S=1|d,X)\mathbb{P}(d|X)}\mid X\right]\right]\\
    &=\mathbb{E}\left[\frac{S\cdot \mathbbm{1}_{D=d}\cdot \gamma(S,D,X)^2}{\mathbb{P}(S=1|d,X)\mathbb{P}(d|X)}\right]\\
    &\leq \frac{1}{\epsilon^2}\mathbb{E}\left[\gamma(S,D,X)^2\right].
\end{align*}
The other arguments are similar extensions of Theorem~\ref{theorem:dr_static}.
\end{proof}

\begin{lemma}[Riesz representer mean square rate; Theorem 7.6 of \cite{singh2020kernel}]\label{theorem:mse_RR}
Suppose Assumptions~\ref{assumption:RKHS},~\ref{assumption:original},~\ref{assumption:smooth_Riesz},~\ref{assumption:spectral_Riesz}, and~\ref{assumption:cont} hold. Calibrate the ridge regularization sequence such that $$
\lambda_3=\begin{cases} n^{-\frac{1}{2}} & \text{ if }\quad b_3=\infty; \\
n^{-\frac{b_3}{b_3c_3+1}} & \text{ if }\quad b_3\in(1,\infty),\; c_3\in(1,2]; \\
\ln^{\frac{b_3}{b_3+1}}(n)\cdot n^{-\frac{b_3}{b_3+1}} & \text{ if }\quad b_3\in(1,\infty),\; c_3=1.
\end{cases}
$$
Then the kernel ridge Riesz representer $\hat{\alpha}^{(m)}$ trained on $(W_i)_{i=1}^n$ satisfies
$$
\lim_{\tau\rightarrow \infty} \lim_{n\rightarrow\infty} \sup_{P\in\mathcal{P}(b_3,c_3)} \mathbb{P}_{(W_i)\sim P^{n}}(\|\hat{\alpha}^{(m)}-\alpha^{(m)}_0\|_2^2>\tau \cdot  r^2_{\alpha,2}(n,b_3,c_3))=0,
$$
where $\|\hat{\alpha}^{(m)}-\alpha^{(m)}_0\|_2^2=\mathbb{E}[\{\hat{\alpha}^{(m)}(W)-\alpha^{(m)}_0(W)\}^2|(W_i)_{i=1}^n]$ and
$$
r^2_{\alpha,2}(n,b_3,c_3)=\begin{cases} n^{-1} & \text{ if }\quad b_3=\infty; \\
n^{-\frac{b_3c_3}{b_3c_3+1}} & \text{ if }\quad b_3\in(1,\infty),\; c_3\in(1,2]; \\
\ln^{\frac{b_3}{b_3+1}}(n)\cdot n^{-\frac{b_3}{b_3+1}} & \text{ if }\quad b_3\in(1,\infty),\; c_3=1.
\end{cases}
$$
\end{lemma}

\begin{remark}\label{remark:mse}
In various applications, $(b,c)$ vary.
\begin{enumerate}
  \item $\|\hat{\alpha}^{(m)}-\alpha^{(m)}_0\|_2\leq r^{(m)}_{\alpha,2}(n,b_3,c_3)$.
  \item $\|\hat{\pi}-\pi_0\|_2\leq r^{ATE}_{\pi,2}(n,b_6,c_6)$.
        \item $\|\hat{\rho}-\rho_0\|_2\leq r^{ATE}_{\rho,2}(n,b_7,c_7)$.
\end{enumerate}
\end{remark}

Recall that $b=\infty$ is the finite dimensional regime; $b\in(1,\infty),\; c\in(1,2]$ is the infinite dimensional regime with polynomial spectral decay and additional smoothness; $b\in(1,\infty),\; c=1$ is the infinite dimensional regime with polynomial spectral decay and no additional smoothness.

\subsection{Main results}

Appealing to Remark~\ref{remark:unif}, Proposition~\ref{prop:unif_omega}, and Remark~\ref{remark:mse}, I now prove inference for (i) static sample selection and (ii) dynamic sample selection. I build on the theory of debiased machine learning.

\subsubsection{Static sample selection}

Fix $d\in\{0,1\}$. Slightly abusing notation, write
\begin{align*}
    \|\hat{\gamma}-\gamma_0\|_2&=\left(\mathbb{E}[\{\hat{\gamma}_{\ell}(S,D,X)-\gamma_0(S,D,X)\}^2|(W_i)_{i\in I^c_{\ell}}]\right)^{\frac{1}{2}}, \\
    \|\hat{\alpha}^{(m)}-\alpha^{(m)}_0\|_2&=\left(\mathbb{E}[\{\hat{\alpha}^{(m)}_{\ell}(S,D,X)-\alpha^{(m)}_0(S,D,X)\}^2|(W_i)_{i\in I^c_{\ell}}]\right)^{\frac{1}{2}}.
\end{align*}

\begin{lemma}[Gaussian approximation; Theorem 5.1 of \cite{chernozhukov2021simple}]\label{thm:dml}Suppose Assumption~\ref{assumption:cont} holds,
$$
\int \{y-\gamma_0(w)\}^2\mathrm{d}\mathbb{P}(y|w) \leq \bar{\sigma}^2,\quad \|\alpha_0^{(m)}\|_{\infty}\leq\bar{\alpha}.
$$
Then with probability $1-\epsilon$, the DML meta algorithm for the abstract treatment effect in Algorithm~\ref{algorithm:static_dml} (with any choice of regression and Riesz representer estimators) has a finite sample Gaussian approximation given by
$$
\sup_{z\in\mathbb{R}} \left|\mathbb{P}\left(\frac{\sqrt{n}}{\sigma}(\hat{\theta}^{(m)}-\theta^{(m)}_0)\leq z\right)-\Phi(z)\right|\leq c^{BE}\left(\frac{\eta}{\sigma}\right)^3 n^{-\frac{1}{2}}+\frac{\Delta}{\sqrt{2\pi}}+\epsilon
$$
where $\Phi(z)$ is the standard Gaussian cumulative distribution function and
$$
\Delta=\frac{3 L}{\epsilon\cdot  \sigma}\left\{(\sqrt{Q}+\bar{\alpha})\|\hat{\gamma}_{\ell}-\gamma_0\|_2+\bar{\sigma}\|\hat{\alpha}^{(m)}-\alpha^{(m)}_0\|_2+\sqrt{n} \|\hat{\gamma}_{\ell}-\gamma_0\|_2\cdot \|\hat{\alpha}^{(m)}-\alpha^{(m)}_0\|_2\right\}.
$$
\end{lemma}

\begin{lemma}[Variance estimation; Theorem 5.3 of \cite{chernozhukov2021simple}]\label{thm:var}
Suppose Assumption~\ref{assumption:cont} holds,
$$
\int \{y-\gamma_0(w)\}^2\mathrm{d}\mathbb{P}(y|w) \leq \bar{\sigma}^2,\quad \|\hat{\alpha}_{\ell}\|_{\infty}\leq\bar{\alpha}'.
$$
Then with probability $1-\epsilon'$, the DML meta algorithm for the variance in Algorithm~\ref{algorithm:static_dml} (with any choice of regression and Riesz representer estimators) has the finite sample bound given by
$$
|\hat{\sigma}^2-\sigma^2|\leq \Delta'+2\sqrt{\Delta'}(\sqrt{\Delta''}+\sigma)+\Delta''
$$
where
$$
    \Delta'=4(\hat{\theta}^{(m)}-\theta^{(m)}_0)^2+\frac{24 L}{\epsilon'}\{\{Q+(\bar{\alpha}')^2\}\|\hat{\gamma}-\gamma_0\|^2_2+\bar{\sigma}^2\|\hat{\alpha}^{(m)}-\alpha^{(m)}_0\|^2_2\},\quad 
    \Delta''=\sqrt{\frac{2}{\epsilon'}}\chi^2 n^{-\frac{1}{2}}.
$$
\end{lemma}

\begin{proof}[Proof of Theorem~\ref{theorem:inference_static}]
I verify the conditions of Lemmas~\ref{thm:dml} and~\ref{thm:var} appealing to Lemmas~\ref{prop:unif} and~\ref{theorem:mse_RR}. In particular, I verify 
$$
(\sqrt{Q}+\bar{\alpha}/\sigma+\bar{\alpha}')\|\hat{\gamma}_{\ell}-\gamma_0\|_2\overset{p}{\rightarrow} 0,\quad \bar{\sigma}\|\hat{\alpha}^{(m)}-\alpha^{(m)}_0\|_2\overset{p}{\rightarrow} 0,\quad\sqrt{n}\|\hat{\gamma}_{\ell}-\gamma_0\|_2\cdot \|\hat{\alpha}^{(m)}-\alpha^{(m)}_0\|_2/\sigma \overset{p}{\rightarrow}0.
$$
\begin{enumerate}
    \item $\bar{\sigma}$. By Assumption~\ref{assumption:original}
    $$
\int \{y-\gamma_0(w)\}^2\mathrm{d}\mathbb{P}(y|w) \leq 2\int y^2 \mathrm{d}\mathbb{P}(y|w) +2\|\gamma_0\|_{\infty}^2\leq 2(C^2+\kappa^2_w\|\gamma_0\|^2_{\mathcal{H}}).
$$
    \item $\bar{\alpha}$. Since $\alpha_0^{(m)}\in\mathcal{H}$ by hypothesis,
    $$
    |\alpha^{(m)}_0(w)|=|\langle\alpha^{(m)}_0,\phi(w)\rangle_{\mathcal{H}}|\leq \kappa_w \|\alpha^{(m)}_0\|_{\mathcal{H}}.
    $$
    \item $\bar{\alpha}'$.  If I censor extreme evaluations of the kernel ridge Riesz representer at, say, $\bar{\alpha}$ as prescribed in Assumption~\ref{assumption:bounded_Riesz}, then $\|\hat{\alpha}^{(m)}_{\ell}\|_{\infty}\leq \bar{\alpha}'=\bar{\alpha}$.
    \item Relative rates. By Proposition~\ref{prop:cont}, $\sqrt{Q}<\infty$ is a constant. By arguments above, $\bar{\alpha}'=\bar{\alpha}<\infty$ and $\bar{\sigma}<\infty$ are also constants. Therefore the remaining conditions to verify are 
$$
\|\hat{\gamma}_{\ell}-\gamma_0\|_2\overset{p}{\rightarrow} 0,\quad \|\hat{\alpha}^{(m)}-\alpha^{(m)}_0\|_2\overset{p}{\rightarrow} 0,\quad\sqrt{n}\|\hat{\gamma}_{\ell}-\gamma_0\|_2\cdot \|\hat{\alpha}^{(m)}-\alpha^{(m)}_0\|_2 \overset{p}{\rightarrow}0.
$$ 
The first two conditions follow immediately from Lemmas~\ref{prop:unif} and~\ref{theorem:mse_RR}.

Finally I turn to the product rate condition. Let $b_3\in(1,\infty)$ and $c_3\in(1,2]$. In this case,
\begin{align*}
    \sqrt{n}\|\hat{\gamma}_{\ell}-\gamma_0\|_2\cdot \|\hat{\alpha}^{(m)}-\alpha^{(m)}_0\|_2 
    &\leq  \sqrt{n}\|\hat{\gamma}_{\ell}-\gamma_0\|_{\infty}\cdot \|\hat{\alpha}^{(m)}-\alpha^{(m)}_0\|_2  \\
    &= n^{\frac{1}{2}-\frac{1}{2}\frac{c-1}{c+1}-\frac{1}{2}\frac{b_3c_3}{b_3c_3+1}}.
\end{align*}
    Observe that
    $$
    \frac{1}{2}-\frac{1}{2}\frac{c-1}{c+1}-\frac{1}{2}\frac{b_3c_3}{b_3c_3+1}<0\iff \frac{c-1}{c+1}+\frac{b_3c_3}{b_3c_3+1}>1.
    $$
\end{enumerate}
\end{proof}

\begin{remark}[Double spectral robustness]
Alternatively, one could verify the rate conditions reversing the norms. For example,
 \begin{align*}
            \sqrt{n} \|\hat{\gamma}-\gamma_0\|_2 \cdot \|\hat{\alpha}^{(m)}-\alpha^{(m)}_0\|_2 
            &\leq \sqrt{n} \|\hat{\gamma}-\gamma_0\|_2 \cdot \|\hat{\alpha}^{(m)}-\alpha^{(m)}_0\|_{\infty}  \\
            &=\sqrt{n} \cdot r_{\gamma,2}(n,b,c)\cdot r^{(m)}_{\alpha,\infty}(n,c_3) \\
            &=n^{\frac{1}{2}-\frac{1}{2}\frac{bc}{bc+1}-\frac{1}{2}\frac{c_3-1}{c_3+1}},
        \end{align*}
        so I require
        $$
        \frac{1}{2}-\frac{1}{2}\frac{c_3-1}{c_3+1}-\frac{1}{2}\frac{bc}{bc+1}<0\iff \frac{c_3-1}{c_3+1}+\frac{bc}{bc+1}>1.
        $$
In such case, one could place the polynomial spectral decay assumption on the kernel for $\gamma_0$ rather than $\alpha^{(m)}_0$. Since only one of the two kernels in the product rate condition needs to have polynomial spectral decay, this amounts to double spectral robustness. Reversing the norms require a uniform rate on $\alpha^{(m)}_0$, which I conjecture to be $n^{-\frac{1}{2}\frac{c_3-1}{c_3+1}}$ and pose as a question for future work. 
\end{remark}

\subsubsection{Dynamic sample selection}

Fix $d\in\{0,1\}$. Slightly abusing notation, write
\begin{align*}
    \|\hat{\gamma}-\gamma_0\|_2&=\left(\mathbb{E}[\{\hat{\gamma}_{\ell}(S,D,X,M)-\gamma_0(S,D,X,M)\}^2|(W_i)_{i\in I^c_{\ell}}]\right)^{\frac{1}{2}}, \\
    \|\hat{\omega}-\omega_0\|_2&=\left(\mathbb{E}[\{\hat{\omega}_{\ell}(S,D;X)-\omega_0(S,D;X)\}^2|(W_i)_{i\in I^c_{\ell}}]\right)^{\frac{1}{2}}, \\
    \|\hat{\pi}-\pi_0\|_2&=\left(\mathbb{E}[\{\hat{\pi}_{\ell}(X)-\pi_0(X)\}^2|(W_i)_{i\in I^c_{\ell}}]\right)^{\frac{1}{2}}, \\
    \|\hat{\rho}-\rho_0\|_2&=\left(\mathbb{E}[\{\hat{\rho}_{\ell}(D,X,M)-\rho_0(D,X,M)\}^2|(W_i)_{i\in I^c_{\ell}}]\right)^{\frac{1}{2}}.
\end{align*}

\begin{lemma}[Gaussian approximation for dynamic sample selection; Theorem 4 of \cite{bia2020double}]\label{thm:dml_planning}
Suppose Assumptions~\ref{assumption:selection_planning},~\ref{assumption:original},~\ref{assumption:resid_planning}, and~\ref{assumption:bounded_propensity_planning} hold. Further suppose there exists some $q>2$, $\epsilon_n\rightarrow 0$, and $\delta_n\rightarrow 0$ such that, with probability $1-\delta_n$,
\begin{enumerate}
    \item Nuisances have bounded moments: 
    $$\|\hat{\gamma}-\gamma_0\|_q ,\|\hat{\omega}-\omega_0\|_q,\|\hat{\pi}-\pi_0\|_q,\|\hat{\rho}-\rho_0\|_q\leq C.$$
    \item Nuisances have learning rates:
    $$\|\hat{\gamma}-\gamma_0\|_2 ,\|\hat{\omega}-\omega_0\|_2,\|\hat{\pi}-\pi_0\|_2,\|\hat{\rho}-\rho_0\|_2\leq \epsilon_n.$$
    \item These learning rates satisfy product rate conditions
    \begin{align*}
     \sqrt{n} \|\hat{\gamma}-\gamma_0\|_2 \cdot \|\hat{\pi}-\pi_0\|_2 &\leq \epsilon_n,  \\
            \sqrt{n} \|\hat{\gamma}-\gamma_0\|_2 \cdot\|\hat{\rho}-\rho_0\|_2 &\leq \epsilon_n, \\
           \sqrt{n} \|\hat{\omega}-\omega_0\|_2 \cdot \|\hat{\pi}-\pi_0\|_2  &\leq \epsilon_n.
        \end{align*}
\end{enumerate}
Then the meta-algorithm (with black-box nuisances) in Algorithm~\ref{algorithm:planning_dml} satisfies $\sqrt{n}\{\hat{\theta}^{ATE}(d)-\theta^{ATE}_0(d)\}\overset{d}{\rightarrow}\mathcal{N}\{0,\sigma^2(d)\}$ where
\begin{align*}
\sigma^2(d)
    &=\mathbb{E}\bigg(\bigg[
    \omega_0(1,d;X) \\
    &\quad + \frac{\1_{D=d}S}{\pi_0(d;X)\rho_0(1;d,X,M)}\{Y-\gamma_0(1,d,X,M)\} \\
    &\quad + \frac{\1_{D=d}}{\pi_0(d;X)}\left\{\gamma_0(1,d,X,M)-\omega_0(1,d;X)\right\}-\theta_0^{ATE}(d)\bigg]^2
\bigg).
\end{align*}
\end{lemma}

\begin{proof}[Proof of Theorem~\ref{theorem:inference_planning}]
I verify the conditions of Lemma~\ref{thm:dml_planning}.
\begin{enumerate}
    \item Learning rates and bounded moments. By Jensen's inequality as well as Lemma~\ref{prop:unif} and Proposition~\ref{prop:unif_omega}, 
    \begin{enumerate}
        \item $\|\hat{\gamma}-\gamma_0\|_2 \leq \|\hat{\gamma}-\gamma_0\|_q \leq \|\hat{\gamma}-\gamma_0\|_{\infty}\leq r^{ATE}_{\gamma,\infty}(n,c)$;
        \item $\|\hat{\omega}-\omega_0\|_2\leq\|\hat{\omega}-\omega_0\|_q\leq \|\hat{\omega}-\omega_0\|_{\infty}\leq r^{ATE}_{\omega,\infty}(n,c,c_4)$;
        \item $\|\hat{\pi}-\pi_0\|_2\leq \|\hat{\pi}-\pi_0\|_q \leq \|\hat{\pi}-\pi_0\|_{\infty}\leq r^{ATE}_{\pi,\infty}(n,c_6)$;
        \item $\|\hat{\rho}-\rho_0\|_2 \leq \|\hat{\rho}-\rho_0\|_q \leq \|\hat{\rho}-\rho_0\|_{\infty}\leq r^{ATE}_{\rho,\infty}(n,c_7)$.
    \end{enumerate}
    \item Product rate conditions
     \begin{enumerate}
        \item $\sqrt{n} \|\hat{\gamma}-\gamma_0\|_2 \cdot \|\hat{\pi}-\pi_0\|_2$. By Lemmas~\ref{prop:unif} and~\ref{theorem:mse}, in the case that $b_6\in(1,\infty)$ and $c_6\in(1,2]$
        \begin{align*}
            \sqrt{n} \|\hat{\gamma}-\gamma_0\|_2 \cdot \|\hat{\pi}-\pi_0\|_2
            &\leq \sqrt{n} \|\hat{\gamma}-\gamma_0\|_{\infty} \cdot \|\hat{\pi}-\pi_0\|_2 \\
            &=\sqrt{n} \cdot r^{ATE}_{\gamma,\infty}(n,c)\cdot r^{ATE}_{\pi,2}(n,b_6,c_6) \\
            &=n^{\frac{1}{2}-\frac{1}{2}\frac{c-1}{c+1}-\frac{1}{2}\frac{b_6c_6}{b_6c_6+1}},
        \end{align*}
        so I require
        $$
        \frac{1}{2}-\frac{1}{2}\frac{c-1}{c+1}-\frac{1}{2}\frac{b_6c_6}{b_6c_6+1}<0\iff \frac{c-1}{c+1}+\frac{b_6c_6}{b_6c_6+1}>1.
        $$
        \item $\sqrt{n} \|\hat{\gamma}-\gamma_0\|_2 \cdot\|\hat{\rho}-\rho_0\|_2 $.  By Lemmas~\ref{prop:unif} and~\ref{theorem:mse}, in the case that $b_7\in(1,\infty)$ and $c_7\in(1,2]$
        \begin{align*}
            \sqrt{n} \|\hat{\gamma}-\gamma_0\|_2 \cdot \|\hat{\rho}-\rho_0\|_2
            &\leq \sqrt{n} \|\hat{\gamma}-\gamma_0\|_{\infty} \cdot \|\hat{\rho}-\rho_0\|_2 \\
            &=\sqrt{n} \cdot r^{ATE}_{\gamma,\infty}(n,c)\cdot r^{ATE}_{\rho,2}(n,b_7,c_7) \\
            &=n^{\frac{1}{2}-\frac{1}{2}\frac{c-1}{c+1}-\frac{1}{2}\frac{b_7c_7}{b_7c_7+1}},
        \end{align*}
        so I require
        $$
        \frac{1}{2}-\frac{1}{2}\frac{c-1}{c+1}-\frac{1}{2}\frac{b_7c_7}{b_7c_7+1}<0\iff \frac{c-1}{c+1}+\frac{b_7c_7}{b_7c_7+1}>1.
        $$
        \item $\sqrt{n} \|\hat{\omega}-\omega_0\|_2 \cdot \|\hat{\pi}-\pi_0\|_2 $. By Proposition~\ref{prop:unif_omega} and Lemma~\ref{theorem:mse}, in the case that $b_6\in(1,\infty)$ and $c_6\in(1,2]$
        \begin{align*}
            \sqrt{n} \|\hat{\omega}-\omega_0\|_2 \cdot \|\hat{\pi}-\pi_0\|_2
            &\leq \sqrt{n} \|\hat{\omega}-\omega_0\|_{\infty} \cdot \|\hat{\pi}-\pi_0\|_2 \\
            &=\sqrt{n} \cdot r^{ATE}_{\omega,\infty}(n,c,c_4)\cdot r^{ATE}_{\pi,2}(n,b_6,c_6) \\
            &=n^{\frac{1}{2}-\frac{1}{2}\frac{c_{\omega}-1}{c_{\omega}+1}-\frac{1}{2}\frac{b_6c_6}{b_6c_6+1}},
        \end{align*}
        where $c_{\omega}=\min(c,c_4)$ so I require
        $$
        \frac{1}{2}-\frac{1}{2}\frac{c_{\omega}-1}{c_{\omega}+1}-\frac{1}{2}\frac{b_6c_6}{b_6c_6+1}<0\iff \frac{c_{\omega}-1}{c_{\omega}+1}+\frac{b_6c_6}{b_6c_6+1}>1.
        $$
    \end{enumerate}
\end{enumerate}
Consistency of the variance estimator follows from \cite[Theorem 3.2]{chernozhukov2018original}.
\end{proof}

\begin{remark}[Double spectral robustness]
Alternatively, one could verify the rate conditions reversing the norms. For example,
 \begin{align*}
            \sqrt{n} \|\hat{\gamma}-\gamma_0\|_2 \cdot \|\hat{\pi}-\pi_0\|_2
            &\leq \sqrt{n} \|\hat{\gamma}-\gamma_0\|_2 \cdot \|\hat{\pi}-\pi_0\|_{\infty} \\
            &=\sqrt{n} \cdot r^{ATE}_{\gamma,2}(n,b,c)\cdot r^{ATE}_{\pi,\infty}(n,c_6) \\
            &=n^{\frac{1}{2}-\frac{1}{2}\frac{bc}{bc+1}-\frac{1}{2}\frac{c_6-1}{c_6+1}},
        \end{align*}
        so I require
        $$
        \frac{1}{2}-\frac{1}{2}\frac{c_6-1}{c_6+1}-\frac{1}{2}\frac{bc}{bc+1}<0\iff \frac{c_6-1}{c_6+1}+\frac{bc}{bc+1}>1.
        $$
In such case, one could place the polynomial spectral decay assumption on the kernel for $\gamma_0$ rather than $\pi_0$. Since only one of the two kernels in a product rate condition needs to have polynomial spectral decay, this amounts to double spectral robustness. Reversing the norms of the other product rate conditions would require a mean square rate on $\omega_0$, which I pose as a question for future work. 
\end{remark}
\section{Tuning}\label{section:tuning}

\subsection{Simplified setting}

I propose a family of novel estimators that are combinations of kernel ridge regressions. As such, the same two kinds of hyperparameters that arise in kernel ridge regressions arise in the estimators: ridge regression penalties and kernel hyperparameters. In this section, I recall practical tuning procedures for such hyperparameters. To simplify the discussion, I focus on the regression of $Y$ on $W$. Recall that the closed form solution of the regression estimator using all observations is
$
\hat{f}(w)=K_{wW}(K_{WW}+n\lambda I)^{-1}\mathbf{Y}.
$

\subsection{Ridge penalty}

It is convenient to tune $\lambda$ by leave-one-out cross validation (LOOCV), since the validation loss has a closed form solution.

\begin{algorithm}[Ridge penalty tuning; Algorithm F.1 of \cite{singh2020kernel}]
Construct the matrices
$$
H_{\lambda}:=I-K_{WW}(K_{WW}+n\lambda I)^{-1}\in\mathbb{R}^{n\times n},\quad \tilde{H}_{\lambda}:=diag(H_{\lambda})\in\mathbb{R}^{n\times n},
$$
where $\tilde{H}_{\lambda}$ has the same diagonal entries as $H_{\lambda}$ and off diagonal entries of 0. Then set
$$
\lambda^*=\argmin_{\lambda \in\Lambda} \frac{1}{n}\|\tilde{H}_{\lambda}^{-1}H_{\lambda} Y\|_2^2,\quad \Lambda\subset\mathbb{R}.
$$
\end{algorithm}

\subsection{Kernel}

The Gaussian kernel is the most popular kernel among machine learning practitioners:
$
k(w,w')=\exp\left(-\frac{1}{2}\frac{\|w-w'\|^2_{\mathcal{W}}}{\sigma^2}\right)
$.
Importantly, this kernel satisfies the required properties: it is continuous, bounded, and characteristic.

Observe that the Gaussian kernel has a hyperparameter: the \textit{lengthscale} $\sigma$. A convenient heuristic is to set the lengthscale equal to the median interpoint distance of $(W_i)^n_{i=1}$, where the interpoint distance between observations $i$ and $j$ is $\|W_i-W_j\|_{\mathcal{W}}$. When the input $W$ is multidimensional, I use the kernel obtained as the product of scalar kernels for each input dimension. For example, if $\mathcal{W}\subset \mathbb{R}^d$ then
$
k(w,w')=\prod_{j=1}^d \exp\left\{-\frac{1}{2}\frac{(w_j-w_j')^2}{\sigma_j^2}\right\}.
$
Each lengthscale $\sigma_j$ is set according to the median interpoint distance for that input dimension.

In principle, one could instead use LOOCV to tune kernel hyperparameters in the same way that I use LOOCV to tune ridge penalties. However, given the choice of product kernel, this approach becomes impractical in high dimensions. For example, $D\in\mathbb{R}$ and $X\in\mathbb{R}^{100}$ leads to a total of 101 lengthscales $(\sigma_j)$. Even with a closed form solution for LOOCV, searching over this high dimensional grid becomes cumbersome.

\end{document}